\documentclass[3p,12pt]{elsarticle}
\usepackage{lineno,hyperref}
\usepackage{amsmath,bm,geometry}
\usepackage{amsfonts}
\usepackage{color}
\usepackage{nomencl,etoolbox,ragged2e,siunitx}
\usepackage{mathtools}
\usepackage{multirow}
\usepackage{caption}
\usepackage{parskip}
\usepackage{calligra}
\usepackage{makeidx}
\usepackage{tikz}
\usepackage{subcaption}
\usepackage[utf8]{inputenc}
\usepackage[english]{babel}
\usepackage{amsthm}
\newtheorem{proposition}{Proposition}[section]
\usetikzlibrary{matrix}
\makeindex
\DeclareMathAlphabet{\mathcalligra}{T1}{calligra}{m}{n}
\begin{document}
\title{Unified Gas-kinetic Wave-Particle Method IV: Multi-species Gas Mixture and Plasma Transport}
\author[ad1]{Chang Liu}
\ead{cliuaa@connect.ust.hk}
\author[ad1,ad2]{Kun Xu\corref{cor1}}
\ead{makxu@ust.hk}
\address[ad1]{Department of Mathematics, Hong Kong University of Science and Technology, Hong Kong, China}
\address[ad2]{Shenzhen Research Institute, Hong Kong University of Science and Technology, Shenzhen, China}
\cortext[cor1]{Corresponding author}

\begin{abstract}
In this paper, we extend the unified gas-kinetic wave-particle (UGKWP) method
to the multi-species gas mixture and multiscale plasma transport.
The construction of the scheme is based on the direct modeling on the mesh size and time step scales, and the local cell's
Knudsen number determines the flow physics.
The proposed scheme has the multiscale and asymptotic complexity diminishing properties.
The multiscale property means that according to cell's Knudsen number the scheme can capture the non-equilibrium flow physics in the rarefied flow regime,
and preserve the asymptotic Euler, Navier-Stokes, and magnetohydrodynamics limit in the continuum regime.
The asymptotic complexity diminishing property means that the total
degree of freedom of the scheme automatically decreases as cell's Knudsen number decreases.
In the continuum regime, the scheme automatically degenerates from a kinetic solver to a hydrodynamic solver.
In UGKWP, the evolution of
microscopic velocity distribution is coupled with the evolution of macroscopic
variables, and the particle evolution as well as the macroscopic fluxes are
modeled from the time accumulating solution up to a time step scale from the kinetic model equation.
 For plasma transport,
current scheme provides a smooth transition from particle in cell (PIC) method in the rarefied regime
to the magnetohydrodynamic (MHD) solver in the continuum regime.
In the continuum limit, the cell size and time step of the UGKWP method is not restricted to be less than the mean free path and
mean collision time. In the highly magnetized regime, the cell size and time step are not
restricted by the Debye length and plasma cyclotron period.
The multiscale and asymptotic complexity diminishing properties of the scheme are
verified by numerical tests in multiple flow regimes.
\end{abstract}

\maketitle
\section{Introduction}
Gas mixture and plasma widely exit in the universe and
are extensively applied in the industry of aerospace, chemical, and nuclear engineering.
Both gas mixture and plasma transport have multiscale flow dynamics.
For the gas mixture, the flow regime various from the rarefied to
continuum regime according to the Knudsen number.
In the rarefied regime, the fundamental governing equation is the multi-species Boltzmann equation \cite{chapman1990mathematical},
which resolves the physics on the mean free path and mean collision time scale.
The complex five-fold integral collision operator makes the Boltzmann equation difficult for
both mathematical analysis and numerical simulation. Therefore, many kinetic models
have been proposed, for example the McCormack model \cite{mccormack1973construction}
that linearizes the nonlinear collision term with the assumption for the distribution function
to slightly deviate from equilibrium;
the Andries-Aoki-Perthame model \cite{AAP} in which the collision in modeled by
a single relaxation term; and other modified models that can
recover transport coefficients correctly \cite{brull2015ellipsoidal,liu2016asymptotic}.
Although the kinetic equation resolves small scale flow physics, the high dimension of the equation puts barrier in
practical 3D engineering applications.
The hydrodynamic model, namely the Euler or Navier-Stokes (NS) equations are mostly used in the continuum regime.
For the plasma transport, the flow regime varies from rarefied regime to continuum regime according to the Knudsen number,
and varies from the two fluid regime to magnetohydrodynamic (MHD) regime according to the normalized Larmor radius and Debye length.
In the rarefied flow regime with large Knudsen number, the plasma flow physics
is described by the kinetic Fokker-Planck-Landau equation coupled with the Maxwell equation \cite{chen1984introduction}.
In the hydrodynamic regime at small Knudsen number, the two-fluid hydrodynamic system coupled with Maxwell equation can describe
the plasma flow dynamics in a more effective way, which takes into account the Hall effect, electron inertia effect,
resistive effect, etc. \cite{hakim2006high}.
In the highly magnetized flow regime where the normalized Larmor radius approaches zero and Debye length on the order of the reciprocal of
the speed of light, a single fluid ideal MHD can be used to approximate the large scale plasma flow dynamics \cite{xu1999gas}.
For both multiscale gas mixture and plasma transport, the hydrodynamic models are more effective, but limited in the continuum regime;
while the kinetic models capture small scale physics, but have complex form and high dimension. Therefore, in order to capture flow physics in different regime in the corresponding most efficient way, the construction of multiscale model
for gas mixture and plasma transport is highly demanded.

In general, the numerical methods for gas mixture and plasma transport can be categorized into the deterministic method and stochastic method.
The deterministic discrete ordinate method (DVM) has great advantage in the simulation of low speed flow as it does not suffer from any statistical noise \cite{liu2020aia}.
In the last several years, many deterministic numerical methods have
been developed for multi-species gas mixture \cite{jin2010micro,zhang2018discrete,zhang2019discrete,brull2020local}, as well as
plasma transport \cite{qiu2010conservative,liu2017unified}. On the other hand, when dealing with the high speed flow and 3D flow,
the stochastic particle method shows advantage in term of computation efficiency. The direct simulation Monte Carlo (DSMC) method has been
extended to gas mixture and chemical reaction \cite{thomasopen}. For the simulation of plasma transport, the particle in cell (PIC) method
has been developed and widely applied in industry \cite{tskhakaya2007particle}. For the traditional DVM, DSMC, and PIC methods, the numerical
cell size is usually required to be less than the mean free path and Deybye length, and the time step is required to be less than the mean
collision time. The cell size and time step constraints reduce the computational efficiency of the traditional DVM, DSMC, and PIC methods
and it becomes impossible to use them in the continuum regime. In order to remove the constraints, the asymptotic preserving schemes have been proposed that can preserve the
flow dynamics in the collisionless and Euler limiting regime \cite{jin2010micro,degond2010asymptotic}.

The unified gas-kinetic scheme (UGKS) proposed by Xu et al. is a multiscale numerical numerical method for the simulation of gas flow \cite{xu2010unified,xu2015direct}.
In the last decade, the UGKS has been well developed and extended to the field of multiscale photon transport \cite{sun2015asymptotic},
plasma transport \cite{liu2017unified}, gas-particle multiphase flow \cite{liu2019unified}, neutron transport \cite{shuang2019parallel}, etc.
The two important ingredients of UGKS are: firstly, the evolution of velocity distribution function is coupled with the evolution of the macroscopic
conservative variables; secondly, the numerical flux of UGKS is constructed from the integral solution of the kinetic equation
which takes into account both particle free stream and collision effects. The UGKS has been proved to be a second order unified preserving
scheme that can accurately capture the NS solution with cell size and time step being much larger than the mean free path and mean collision time \cite{liu2020unified}, the same as traditional NS solvers in discretizing the macroscopic equations directly.
To improve the efficiency of UGKS in the simulation of high speed flow, the unified gas-kinetic wave-particle (UGKWP) method has been proposed and
applied in the simulation of multiscale gas dynamics and photon transport \cite{liu2020unified,zhu2019unified,li2020unified}.
The construction of UGKWP method follows the direct modeling methodology of UGKS: the evolution of microscopic simulation particle is coupled with
the evolution macroscopic variables, and the multiscale particle evolution equation is derived from the integral solution of the kinetic equation.
The propose of this work is to extend the UGKWP method to the field of multiscale gas mixture and plasma transport.

The rest of the paper is organized as following. The governing equations for gas mixture and plasma transport are discussed
in Section \ref{governingequations}. In Section \ref{ugkwp-mixture}, the UGKWP methods for gas mixture and plasma transport are proposed.
The unified preserving and asymptotic complexity diminishing properties of UGKWP are discussed in Section \ref{discussion}. The numerical
examples are shown in Section \ref{tests}, and the last Section \ref{conclusion} gives the conclusion.

\section{Governing equations for multi-species gas mixture and plasma transport}\label{governingequations}
This section is to present the governing equations based on which the scheme is constructed. The multi-species Boltzmann equation is first reviewed, and then the kinetic model equation proposed by Andries, et al. \cite{AAP} will be discussed, including its asymptotic behavior in continuum regime. The two fluid kinetic-Maxwell system, as well as the Hall-MHD equations will be presented as well.
\subsection{Multi-species Boltzmann equation}
A gas mixture composed of $m$ species can be modeled by the multi-species Boltzmann equations,
\begin{equation}\label{Boltzmann}
  \partial_t f_\alpha+\vec{v}\cdot\nabla f_\alpha
  =\sum_{k=1}^{m}\mathcal{Q}_{\alpha k}(f_\alpha,f_k),
\end{equation}
where $f_{\alpha}(t,\vec{x},\vec{v})$ is the velocity distribution function of species $\alpha$, and the collision between species $\alpha$ and $k$ follows the integral operator
\begin{equation}\label{Boltzmann-collision}
  \mathcal{Q}_{\alpha k}(f_\alpha,f_k)=
  \int_{\mathcal{R}^3}\int_{\mathcal{S}^2}
  (f_\alpha' f_k^{*'}-f_\alpha f_k^*)
  B^{\alpha,k}(\vec{v}_r\cdot\vec{n},|\vec{v}_r|)d\vec{n}dv^*
\end{equation}
where $f_k^*=f_k(t,\vec{x},\vec{v}^*)$, $f_\alpha'=f_\alpha(t,\vec{x},\vec{v}')$,
$f_k^{*'}=f_k(t,\vec{x},\vec{v}^{*'})$. The post collision velocity $\vec{v}$ and $\vec{v}^{*'}$ follow
\begin{equation*}
\begin{aligned}
&\vec{v}'=\vec{v}-2\frac{\mu_{\alpha k}}{m_\alpha}(\vec{v}_r\cdot\vec{n})\vec{n},\\
&\vec{v}^{*'}=\vec{v}^*+2\frac{\mu_{\alpha k}}{m_k}(\vec{v}_r\cdot\vec{n})\vec{n},
\end{aligned}
\end{equation*}
where $\mu_{\alpha k}=\frac{m_\alpha m_k}{m_\alpha+m_k}$ is the reduced mass, and $\vec{n}$ is the unit vector joining the centers of the two colliding spheres. The collision kernel $B^{\alpha,k}$ depends on relative velocity $\vec{v}_r=\vec{v}-\vec{v}^*$.
The macroscopic density $\rho_\alpha$, velocity $\vec{U}_\alpha$, temperature $T_\alpha$, and energy $E_\alpha$ of species $\alpha$ can be calculated by taking the moments of the velocity distribution $f_\alpha$,
\begin{equation*}
\begin{aligned}
  &\rho_\alpha=m_\alpha n_\alpha=\int_{R^3}f_\alpha d\vec{v},
  \quad \rho_\alpha \vec{U}_\alpha=\int_{R^3}\vec{v}f d\vec{v},\\
  &T_\alpha=\frac{1}{3n_\alpha k_B}\int_{R^3}(\vec{v}-\vec{U}_\alpha)^2 f d\vec{v},
  \quad E_\alpha=\frac12\rho_\alpha|\vec{U}_\alpha|^2+\frac32n_\alpha k_B T_\alpha,
\end{aligned}
\end{equation*}
where $m_\alpha$ and $n_\alpha$ are the molecular mass and number density of species $\alpha$. The total density $\rho$, total number density $n$, total momentum $\rho \vec{U}$, and total energy $E$ satisfy
\begin{equation}\label{total}
\begin{aligned}
  &\rho=\sum^m_{k=1} \rho_k,\quad n=\sum^m_{k=1}n_k,\\
  &\rho \vec{U}=\sum_{k=1}^m \rho_k\vec{U}_k,\quad E=\sum_{k=1}^m E_k.
\end{aligned}
\end{equation}
Boltzmann equation is a fundamental equation that describes the mean free path level flow physics,
however, the five-fold collision operator is costly numerically.
Simplified kinetic model equations are developed in the literature
\cite{sirovich1962kinetic,garzo1989kinetic,asinari2008consistent,groppi2011kinetic},
including a relaxation-type kinetic model proposed by Andries, et al. \cite{AAP}.
Andries' model will be introduced in the next section, based on which the numerical schemes for multi-species gas mixture and plasma are constructed.

\subsection{Kinetic model equation for multi-species gas mixture}
The relaxation-type kinetic model equation that originally proposed by Gross and Krook \cite{gross1956model} has been widely used in the numerical simulation of rarefied gas dynamics due to its simple formulation.
Such BGK-type operator has been extended to model the multi-species collision by Andries, Aoki, and Perthame \cite{AAP}, which can be written as
\begin{equation}\label{AAP}
  \partial_t f_\alpha +\vec{v}\cdot\nabla_xf_\alpha=\frac{g_\alpha-f_\alpha}{\tau_\alpha},
\end{equation}
where the post collision distribution function is a Maxwellian distribuion
\begin{equation}\label{postcollision}
g_\alpha=\rho_\alpha\left(\frac{m_\alpha}{2\pi kT^*_\alpha}\right)^{3/2}\exp
\left(-\frac{m_\alpha}{2k_BT^*_\alpha}(\vec{v}-\vec{U}^*_\alpha)^2\right),
\end{equation}
and the parameters $T_\alpha^*$ and $\vec{U}^*_\alpha$ are chosen to recover the exchanging relations for Maxwell
molecule, which takes the form
\begin{equation}
\begin{aligned}
  \vec{U}_\alpha^*&=\vec{U}_\alpha+\tau_\alpha\sum_{k=1}^N 2\mu_{\alpha k}\chi_{\alpha k}n_k(\vec{U}_k-\vec{U}_\alpha),\\
  T^*_\alpha&=T_\alpha-\frac{m_\alpha}{3k_B}(\vec{U}^*_\alpha-\vec{U}_\alpha)^2
  +\tau_\alpha\sum_{k=1}^N\frac{4\mu_\alpha \chi_{\alpha k}n_k}{m_\alpha+m_k}\left(T_k-T_\alpha+\frac{m_k}{3k_B}(\vec{U}_k-\vec{U}_\alpha)^2\right).
\end{aligned}
\end{equation}
For Maxwell molecules, the interaction coefficient $\chi$ and relaxation parameter $\tau$ satisfy
\begin{equation*}
  \frac{1}{\tau_\alpha}=\sum_{k=1}^{N} \chi_{\alpha k}n_k,
  \quad \chi_{\alpha k}=0.422\pi\left(\frac{a_{\alpha k}(m_\alpha+m_r)}{m_\alpha m_r}\right)^{\frac12},
\end{equation*}
where $a_{\alpha k}$ is the constant of proportionality in the intermolecular force law \cite{morse1963energy}.
The advantage of Andries' kinetic model is that it satisfies the indifferentiability principle, entropy condition, and can recover the exchanging relationship of Maxwell molecules with such a simple relaxation form \cite{AAP}.

Based on Andries' model, Liu et al. proposed a BGK-Maxwell system for fully ionized plasma transport \cite{liu2017unified}, which can be written as
\begin{equation}\label{liu-plasma}
\begin{aligned}
  &\frac{\partial f_\alpha}{\partial t}+\vec{v}\cdot \nabla_x f_\alpha+\vec{a}_\alpha\cdot\nabla_xf_\alpha=\frac{g_\alpha-f_\alpha}{\tau_\alpha}\\
  &\frac{\partial \vec{B}}{\partial t}+\nabla_x\times\vec{E}=0\\
  &\frac{\partial \vec{E}}{\partial t}-\vec{c}^2\nabla_x\times \vec{B}=\frac{1}{\epsilon_0}\vec{j},
\end{aligned}
\end{equation}
where the velocity distribution $f_\alpha(t,\vec{x},\vec{v})$ of species $\alpha$ ($\alpha=i$ for ion and $\alpha=e$ for electron) is governed by a kinetic equation that coupled with the Maxwell equations for electromagnetic wave. In the Maxwell equation, $\vec{E}$ and $\vec{B}$ are the electric and magnetic field, $\vec{c}$ is speed of light, and $\epsilon_0$ is the vacuum permittivity. In the kinetic equation, the Lorenz acceleration $\vec{a}_\alpha$ takes the form
\begin{equation*}
\vec{a}_\alpha=\frac{e(\vec{E}+\vec{v}\times\vec{B})}{m_\alpha},
\end{equation*}
where $e$ is electric charge, and $m_\alpha$ is the particle mass of species $\alpha$.
The post collision distribution $g_\alpha$ takes the same form of the Andries' model as given in Eq.\eqref{postcollision}, however the interspecies interaction coefficient $\chi_{ie}$ is determined by the plasma electrical conductivity $\sigma_p$ \cite{liu2017unified}
\begin{equation}\label{conductivity}
 \chi_{ie}\sigma_p=\frac{n_ie^2(m_i+m_e)}{2m_im_e}.
\end{equation}
The hydrodynamic equations such as the Navier-Stokes equations, the Euler equations, and magnetohydrodynamic equations can be derived in the continuum regime. The asymptotic behavior of above kinetic model Eq.\eqref{AAP} and Eq.\eqref{liu-plasma} will be briefly discussed in the next subsection.

\subsection{Asymptotic behavior of the kinetic system}
In this section, the asymptotic analysis is applied to give the corresponding hydrodynamic limits of the Andries' and BGK-Maxwell equations. Given the reference variables
length $L_\infty$, temperature $T_\infty$, mass $m_\infty$,  number density $n_\infty$, and magnetic field strength $B_\infty$, the following reference variables can be deduced,
\begin{equation*}
\begin{aligned}
&V_\infty=\sqrt{2k_B T_\infty/m_\infty}, \quad t_\infty=L_\infty/V_\infty,\quad \rho_\infty=m_\infty n_\infty,\\ & E_\infty=B_\infty V_\infty, \quad a_\infty=eB_\infty V_\infty/m_\infty,\quad f_\infty=m_\infty n_\infty/V^3_\infty,
\end{aligned}
\end{equation*}
which are the reference velocity, time, density, electric field, acceleration, and velocity distribution respectively. Based on above reference variables, the Andries' kinetic model can be re-scaled as
\begin{equation*}
  \partial_{\tilde{t}} \tilde{f}_\alpha +\tilde{\vec{v}}\cdot\nabla_{\tilde{x}}\tilde{f}_\alpha
  =\frac{\tilde{g}_\alpha-\tilde{f}_\alpha}{\tilde{\tau}_\alpha},
\end{equation*}
and the BGK-Maxwell system can be re-scaled as
\begin{equation*}
\begin{aligned}
  &\frac{\partial \tilde{f}_\alpha}
  {\partial \tilde{t}}+\tilde{\vec{v}}\cdot \nabla_{\tilde{x}} \tilde{f}_\alpha+
  \frac{1}{\tilde{r}}\vec{\tilde{a}}_\alpha\cdot\nabla_{\tilde{x}}
  \tilde{f}_\alpha=\frac{\tilde{g}_\alpha-\tilde{f}_\alpha}
  {\tilde{\tau}_\alpha}\\
  &\frac{\partial \tilde{\vec{B}}}{\partial \tilde{t}}
  +\nabla_{\tilde{x}}\times\tilde{\vec{E}}=0\\
  &\frac{\partial \tilde{\vec{E}}}{\partial \tilde{t}}
  -\tilde{c}^2\nabla_{\tilde{x}}\times \tilde{\vec{B}}=\frac{1}{\tilde{\lambda}_D^2 \tilde{r}}\vec{j},
\end{aligned}
\end{equation*}
where the variables with a tilde stand for the re-scaled variables, and especially $\tilde{r}$ and $\tilde{\lambda}_D$ are the normalized Larmor radius and Debye length,
\begin{equation*}
  \tilde{r}=\frac{eB_\infty L_\infty}{m_\infty V_\infty},\quad
  \tilde{\lambda}_D=\sqrt{\frac{\epsilon_0m_\infty V_\infty^2}{ne^2}}\frac{m_\infty V_\infty}{eB_\infty}.
\end{equation*}
For the sake of simplicity, the tilde is omitted in the following parts of the paper.

In the continuum regime, the Andries' kinetic equation recovers the Navier-Stokes and Euler equations as $\tau\to 0$,
\begin{equation}\nonumber
\text{Andries' equation}
\xrightarrow{\tau,\chi \ll 1}\text{Navier-Stokes equations}
\xrightarrow{\tau,\chi \to 0}\text{Euler equations}.
\end{equation}
According to the Chapman-Enskog theory \cite{chapman1990mathematical}, the distribution of Andries' kinetic model can be expanded as
\begin{equation}\label{ceexpansion}
f_\alpha=g-\tau_\alpha(\partial_t \bar{g}_\alpha+\vec{v}\cdot\nabla_x \bar{g}_\alpha)+O(\tau^2),
\end{equation}
where $\bar{g}$ is the Maxwellian distribution of the averaged quantities of all species that are evaluated from Eq.\eqref{total}.
The zero-th order expansion with respect to $\tau_\alpha$ gives the Euler equations \cite{AAP},
\begin{equation}\label{aap-euler}
\begin{aligned}
  &\partial_t \rho_\alpha+\nabla\cdot(\rho_\alpha \vec{U})=0,\\
  &\partial_t(\rho \vec{U})+\nabla\cdot(p\mathbb{I}+\rho\vec{U}\vec{U})=0,\\
  &\partial_tE+\nabla\cdot(E\vec{U}+p\vec{U})=0,
\end{aligned}
\end{equation}
and the first order expansion gives the Navier-Stokes equations \cite{AAP},
\begin{equation}\label{aap-ns}
\begin{aligned}
  &\partial_t \rho_\alpha+\nabla\cdot(\rho_\alpha \vec{U}+\vec{J_\alpha})=0,\\
  &\partial_t(\rho \vec{U})+\nabla\cdot(p\mathbb{I}+\rho\vec{U}\vec{U}+\sigma)=0,\\
  &\partial_tE+\nabla\cdot(E\vec{U}+p\vec{U}+\sigma\cdot\vec{U}+\vec{q})=0.
\end{aligned}
\end{equation}
The mass diffusion flux $\vec{J}_\alpha$ is
\begin{equation*}
\vec{J}_\alpha=-\sum_{k=1}^m L_{\alpha k}\frac{\nabla_x(n_\alpha k_B T)}{\rho_\alpha},
\end{equation*}
and the shear stress $\sigma$ and heat flux $q$ satisfy
\begin{equation*}
  \begin{aligned}
    &\sigma=-\mu(\nabla_x\vec{U}+(\nabla_x\vec{U})^T-\frac23\nabla\cdot\vec{U}\mathbb{I}),\\
    &\vec{q}=\frac52k_B T\sum^n_{k=1} \frac{J_k}{m_k}-\kappa \nabla_xT,
  \end{aligned}
\end{equation*}
with the viscous coefficient $\mu=k_BT\sum_{k=1}^m \tau_{\alpha k}n_k$, and the heat conduction coefficient $\kappa=\frac{52}k^2_B T\sum_{k=1}^m \frac{\tau_k n_k}{m_k}$.
In the mass flux, the Soret and the Dufour coefficients are equal to zero, which stand only for Maxwell particles.

The BGK-Maxwell system converges to the two-fluid system and magnetohydrodynamics system as $\tau\to 0$ and $r\to 0$,
\begin{equation*}\nonumber
\text{BGK-Maxwell equation}
\xrightarrow[\chi_{ie}\sim 1]{\tau_\alpha \ll 1}\text{two-fluid system}
\xrightarrow{r\ll 1}\text{MHD equations}.
\end{equation*}
In the continuum regime with $\tau_\alpha \ll 1$, and $\chi_{ie}\sim 1$, the distribution of BGK-Maxwell system can be expanded as
\begin{equation*}
f_\alpha=g_\alpha-\tau_\alpha(\partial_t g_\alpha+\vec{v}\cdot\nabla_x g_\alpha)+O(\tau^2),
\end{equation*}
where $g_\alpha$ is the Maxwellian distribution of the macroscopic quantities of species $\alpha$.
The first order expansion gives hydrodynamic two-fluid equations
\begin{equation}\label{two-fluid}
\begin{aligned}
  &\partial_t \rho_\alpha+
  \nabla_x\cdot(\rho_\alpha\vec{U}_\alpha)=0,\\
  &\partial_t(\rho_\alpha \vec{U}_\alpha)
  +\nabla_x\cdot(\rho_\alpha \vec{U}_\alpha\vec{U}_\alpha+p_\alpha \mathbb{Id}-\mu \sigma(\vec{U}_\alpha))
  =\frac{n_\alpha}{r_{L_i}}
  (\vec{E}+\vec{U}_\alpha\times \vec{B})+S_\alpha,\\
  &\partial_t {E}_\alpha
  +\nabla_x\cdot(({E}_\alpha
  +p_\alpha)\vec{U}_\alpha-\mu \sigma(\vec{U}_\alpha) \vec{U}+\kappa\nabla_x T)
  =\frac{n_\alpha}{r_{L_i}}
  \vec{U}_\alpha\cdot \vec{E}+Q_\alpha,
\end{aligned}
\end{equation}
where the strain rate tensor $\sigma(\vec{U})$ is
\begin{equation}\label{strain}
  \sigma(\vec{U}_\alpha)=\left(\nabla_x \vec{U}_\alpha+(\nabla_x \vec{U}_\alpha)^T\right)-\frac23 \text{div}_x \vec{U}_\alpha \mathbb{Id}.
\end{equation}
The viscosity $\mu_\alpha$ and the thermal conductivity $\kappa_\alpha$ can be expressed by the relaxation parameter $\tau_\alpha$ as
\begin{equation*}
\mu_\alpha=\tau_\alpha n_\alpha k_B T_\alpha,\quad
\kappa_\alpha=\tau_\alpha\frac52\frac{k_B}{m} nk_B T.
\end{equation*}
In above two-fluid system, $S_i=-S_e$ and $Q_i=-Q_e$ are the corresponding momentum and energy exchange between electron and ion,
\begin{equation*}
\begin{aligned}
  &S_\alpha=\sum_r\frac{2m_\alpha m_r}{m_\alpha+m_r}
   n_\alpha\chi_{\alpha r}
  (\vec{U}_r-\vec{U}_\alpha),\\
  &Q_\alpha=\sum_r\frac{4m_\alpha m_r}{(m_\alpha+m_r)^2}
  n_r \chi_{\alpha r}
  \left(\frac32k_BT_r-\frac32k_BT_\alpha +\frac{m_r}{2}(\vec{U}_r-\vec{U}_\alpha)^2\right).
\end{aligned}
\end{equation*}
In the magnetohydrodynamic regime with
$m_e\ll m_i$,
the first order with respect to $r$,
the zero-th order with respect of $\tau_\alpha$ and $m_e/m_i$
of the two-fluid system give the Hall-MHD equations,
\begin{equation}\label{Hall-MHD}
  \begin{aligned}
  &\partial_t\rho+\nabla_x\cdot(\rho \vec{U})=0,\\
  &\partial_t(\rho \vec{U})+
  \nabla_x\cdot(\rho \vec{U} \vec{U}+p \mathbb{I})
  =\lambda^2_Dc^2\nabla_x\times\vec{B}\times \vec{B},\\
  &\partial_t E_\alpha+\nabla_x\cdot((E_\alpha+p_\alpha)\vec{U}_\alpha)
  =\lambda^2_Dc^2n\vec{U}\cdot(\nabla_x\times\vec{B}\times \vec{B}),\\
  &\partial_t \vec{B}+\nabla_x \times \vec{E}=0,\\
  &\vec{E}+\vec{U}\times \vec{B}=
  \underbrace{\frac{r}{\sigma} \vec{j}}_{\text{Resistive term}}+\underbrace{\frac{r}{n}\lambda^2_Dc^2\nabla_x\times\vec{B}\times \vec{B}}_{\text{Hall term}},
  \end{aligned}
\end{equation}
where
$\vec{j}=e n_i\vec{U}_i-en_e\vec{U}_e$ is the current density, and $\sigma$ is the electrical conductivity that relates to the interspecies interaction coefficient $\chi_{ie}$ as given in Eq.\eqref{conductivity}.
In the limit where $\lambda_D=c^{-1}$ and $r\to0$, one gets the ideal MHD equations,
\begin{equation}\label{ideal-mhd}
  \begin{aligned}
  &\partial_t\rho+\nabla_x\cdot(\rho \vec{U})=0,\\
  &\partial_t(\rho \vec{U})+\nabla_x\cdot(\rho \vec{U}\vec{U}+pI)=\nabla_x\times \vec{B}\times \vec{B},\\
  &\partial_t E+\nabla_x\cdot
  ((E+p)\vec{U})=\vec{U}\cdot(\nabla_x\times \vec{B}\times \vec{B}),\\
  &\partial_t \vec{B}+\nabla_x \times (\vec{U}\times \vec{B})=0.
  \end{aligned}
\end{equation}
The asymptotic behavior of the Andries' and BGK-Maxwell system is given in above discussion.
In the next section, the unified gas-kinetic wave-particle method for gas mixture and plasma transport will be proposed.

\section{Unified Gas-kinetic Wave-Particle Method}\label{ugkwp-mixture}
\subsection{UGKWP method for multi-species gas mixture}
The unified gas-kinetic wave-particle method is a multiscale numerical method that preserves the asymptotic limits of the Andries' kinetic equations. The UGKWP method couples the evolution of the velocity distribution $f_\alpha$ and the macroscopic quantities $\vec{W}_\alpha$. The evolution of microscopic distribution and macroscopic variables will be given in the following subsections.

\subsubsection{The evolution of microscopic velocity distribution function}
Similar to the UGKWP method for single species gas \cite{liu2020unified}, in current scheme the velocity distribution function is partially represented by an analytical distribution $g^{+,c}_\alpha$ and partially represented by stochastic particles $P_{\alpha k}=(m_{\alpha k}, \vec{x}_{\alpha k}, \vec{v}_{\alpha k})$, which is as shown in Fig. \ref{wpdiagram}. Here $m_{\alpha k}$ is the mass of simulation particle $P_{\alpha k}$, which represents a cluster of real gas particles of species $\alpha$, and $\vec{x}_{\alpha k}$, $\vec{v}_{\alpha k}$ stand for the position and velocity of simulation particle $P_{\alpha k}$.
The evolution of the microscopic velocity distribution function follows the integral solution of the kinetic equation \eqref{AAP}.
With initial condition
$f_\alpha(0,\vec{x},\vec{v})=f_{\alpha,0}(\vec{x},\vec{v})$,
the integral solution at $(\vec{x},t)$ can be written as
\begin{equation}\label{integalsolution}
  f_\alpha(\vec{x},t,\vec{v})=\frac{1}{\tau}\int_0^t e^{-(t-t')/\tau}g_\alpha(\vec{x}',t',\vec{v})dt'
  +e^{-t/\tau}f_{\alpha,0}(\vec{x_0},\vec{v}),
\end{equation}
where the equilibrium distribution is integrated along the characteristics $\vec{x}'=\vec{x}+\vec{v}(t'-t)$.
Substituting the second order Taylor expansion of equilibrium
\begin{equation*}
  g_\alpha(\vec{x}',t',\vec{v})=g_\alpha(\vec{x},t,\vec{v})
  +\nabla_xg_\alpha(\vec{x},t,\vec{v})\cdot(\vec{x}'-\vec{x})
  +\partial_tg_\alpha(\vec{x},t,\vec{v})(t'-t)
  +O\left((\vec{x}'-\vec{x})^2,(t'-t)^2\right),
\end{equation*}
into the integral solution, the numerical multiscale evolution solution for simulation particle can be obtained,
\begin{equation}\label{particlems}
  f_\alpha(\vec{x},t,\vec{v})=(1-e^{-t/\tau})g_\alpha^+(\vec{x},t,\vec{v})
  +e^{-t/\tau}f_{0,\alpha}(\vec{x}_0,\vec{v}),
\end{equation}
where
\begin{equation}\label{g+}
  g_\alpha^+(\vec{x},t,\vec{v})=g_\alpha(\vec{x},t,\vec{v})
  +\left(\frac{te^{-t/\tau}}{1-e^{-t/\tau}}-\tau\right)
  \left(\partial_tg_\alpha(\vec{x},t,\vec{v})
  +\vec{v}\cdot\nabla_xg_\alpha(\vec{x},t,\vec{v})\right).
\end{equation}
A physical interpretation of Eq.\eqref{particlems} is that a particle has a probability $e^{-t/\tau}$
to free stream in a time period $[0,t]$, and has a probability $1-e^{-t/\tau}$
to interact with other particles and reach a velocity distribution $g_\alpha^+$.
The free stream particles are kept and the collisional particles get re-sampled from distribution $g^+_\alpha$.
The cumulative distribution function of the free streaming time $t_f$ is
\begin{equation}\label{tc-distribution}
  F(t_f<t)=\exp(-t/\tau),
\end{equation}
from which $t_f$ can be sampled as $t_f=-\tau\ln(\eta)$ with $\eta$ a uniform distribution $\eta\sim U(0,1)$. For a time step $\Delta t$, the particles with $t_f\ge\Delta t$ will be collisionless particles, and the particles with $t_f<\Delta t$ will be collisional particles.
The procedure of updating particles in UGKWP method is
\begin{description}
  \item[Step 1:] Sample free streaming time $t_{f,\alpha k}$ for each particle $P_{\alpha k}$, and stream particle $P_{\alpha k}$ for a time period of $\min(\Delta t, t_{f,\alpha k})$;
  \item[Step 2:] Keep collisionless particles, and remove collisional particles. Calculate the total conservative quantities of collisional particles $\vec{W}_{i,\alpha}^h$ from the updated conservative quantities $\vec{W}_{i,\alpha}$ as $\vec{W}_{i,\alpha}^h=\vec{W}_{i,\alpha}-\vec{W}_{i,\alpha}^p$;
  \item[Step 3:] Rebuild the microscopic velocity distribution. Calculate the analytical distribution $g_\alpha^{+,c}$ and re-sample collisional particles from distribution $g_\alpha^{+,f}$.
\end{description}
In above particle updating procedure, total conservative quantities of collisionless particles in cell $\Omega_i$ is denoted as $\vec{W}_{i,\alpha}^p$, and the total conservative quantities of collisional particle in cell $\Omega_i$ is denoted as $\vec{W}_{i,\alpha}^h$.
In the distribution rebuilding process, the $\vec{W}_{i,\alpha}^h$ is divided into $g_\alpha^{+,c}=(1-e^{-\Delta t/\tau^{n+1}})g_\alpha^{+}$ and $g_\alpha^{+,f}=e^{-\Delta t/\tau^{n+1}}g_\alpha^{+}$, which coresponding to the collisional and collisionless particles in the next time step from $t^n$ to $t^{n+1}$. The distribution $g_\alpha^{+,c}$ is recorded analytically, and the distribution $g_\alpha^{+,f}$ is re-sampled into stachastic particles.
Above discussion gives the evolution of particles, and in the next subsection we will give the evolution of the conservative variables.

\subsubsection{The evolution of macroscopic quantities}
The evolution of macroscopic quantities is under the framework of finite volume scheme. The cell averaged conservative variables $\vec{W}_{i,\alpha}=(\rho_{i,\alpha},\rho_{i,\alpha}\vec{U}_{i,\alpha},
\rho_{i,\alpha}E_{i,\alpha})$ on a physical cell $\Omega_i$ is defined as
\begin{equation*}
\vec{W}_{i,\alpha}=\frac{1}{|\Omega_i|}\int_{\Omega_i}\vec{W}_\alpha(\vec{x}) d\vec{x}.
\end{equation*}
The finite volume scheme of $\vec{W}_{i,\alpha}$ follows
\begin{align}
&\vec{W}^{n+1}_{i,\alpha}=\vec{W}^n_{i,\alpha}-\sum_s\frac{\Delta t }{|\Omega_{i}|}|l_s|F_{s,\alpha}+\frac{\Delta t}{\tau}(\vec{W}^{* n+1}_{i,\alpha}-\vec{W}^{n+1}_{i,\alpha}),\label{UGKSma}
\end{align}
where $l_s\in \partial \Omega_i$ is the cell interface with center $\vec{x}_s$ and outer unit normal vector $\vec{n}_s$.
The numerical flux of conservative variables $F_{s,\alpha}$  at $\vec{x}_s$ can be written as
\begin{equation*}
  F_{s,\alpha}=\frac{1}{\Delta t}\int_{t^n}^{t_{n+1}}\int \vec{v}\cdot\vec{n}_s f_\alpha(\vec{x}_s,t,\vec{v}) \vec{\Psi} d\Xi dt,
\end{equation*}
where $\vec{\Psi}=\left(1,\vec{v},\frac12(\vec{v}^2+\xi^2)\right)$ is the conservative moments of distribution function with $\xi$ the internal degree of freedom.
The time dependent distribution function $f_\alpha(\vec{x}_s,t,\vec{v})$ at cell interface is constructed from the integral solution of kinetic equation as given in Eq.\eqref{integalsolution}.
The above UGKWP flux for conservative variables can be split into the equilibrium flux
\begin{equation}\label{fluxwg}
  F_{s,\alpha}^{g}=\frac{1}{\Delta t}\int_{t^n}^{t_{n+1}}\int \vec{v}\cdot\vec{n}_s \left\{\frac{1}{\tau}\int_0^t e^{-(t-t')/\tau}g_\alpha(\vec{x}',t',\vec{v})dt'\right\} \vec{\Psi} d\Xi dt,
\end{equation}
 and the free streaming flux
\begin{equation}\label{fluxwf}
  F_{s,\alpha}^{f}=\frac{1}{\Delta t}\int_{t^n}^{t_{n+1}}\int \vec{v}\cdot\vec{n}_s e^{-t/\tau}f_{\alpha,0}(\vec{x_0},\vec{v}) \vec{\Psi} d\Xi dt,
\end{equation}
First, we consider the equilibrium flux $F_s^{eq}$ which can be calculated directly form the macroscopic flow field.
Assume $\vec{x}_s=0$ and $t^n=0$, the equilibrium $g$ can be expanded as
\begin{equation}\label{gtaylor}
  g_\alpha(\vec{x},t,\vec{v})=g_{0,\alpha}+\nabla_{x}g_{0,\alpha}\cdot\vec{x}+\partial_tg_{0,\alpha}t,
\end{equation}
where $g_{0,\alpha}=g_\alpha(0,0,\vec{v})$.
The initial equilibrium $g_{0,\alpha}$ and its spatial and time derivatives can be obtained from the micro-macro consistency
\begin{equation}\label{mimacon}
    \int g_\alpha \vec{\Psi} d\Xi=\int_{\vec{v}\cdot\vec{n}>0} g^l_\alpha \vec{\Psi} d\Xi+\int_{\vec{v}\cdot\vec{n}<0} g^r_\alpha \vec{\Psi} d\Xi,
    \quad \int \nabla_x g_\alpha\vec{\Psi}d\Xi=\nabla_x \vec{W}_\alpha,
\end{equation}
and compatible condition
\begin{equation}\label{compatible}
\int \partial_t g_\alpha \vec{\Psi}d\Xi=-\int\vec{v}\cdot\nabla_x g_\alpha\vec{\Psi} d\Xi,
\end{equation}
where $g^l_\alpha$ and $g^r_\alpha$ are the equilibrium distributions according to the reconstructed left and right side conservative variables at cell interface $\vec{W}^l_\alpha$, $\vec{W}^r_\alpha$, and $\nabla_x \vec{W}_\alpha$ is the reconstructed spatial derivative of conservative variables at cell interface. In this paper, the van Leer limiter is used to achieve a second order accurate space reconstruction. Substitute the reconstructed equilibrium distribution Eq.\eqref{gtaylor} into the equilibrium flux Eq.\eqref{fluxwg}, and we have
\begin{equation*}
  F_{s,\alpha}^g=\int \vec{v}\cdot\vec{n}_s\left(C_1g_{0,\alpha}+C_2\vec{v}\cdot \nabla_x g_{0,\alpha}+C_3\partial_t g_{0,\alpha}\right)\vec{\Psi}d\Xi,
\end{equation*}
where the time integration coefficients are
 \begin{equation*}
   \begin{aligned}
     &C_1=1-\frac{\tau_\alpha}{\Delta t}\left(1-e^{-\Delta t/\tau_\alpha}\right),\\
     &C_2=-\tau_\alpha+\frac{2\tau_\alpha^2}{\Delta t}-e^{-\Delta t/\tau_\alpha}\left(\frac{2\tau_\alpha^2}{\Delta t}+\tau_\alpha\right),\\
     &C_3=\frac12\Delta t-\tau_\alpha+\frac{\tau_\alpha^2}{\Delta t}\left(1-e^{-\Delta t/\tau_\alpha}\right).
   \end{aligned}
 \end{equation*}
Next we consider the free stream flux $F_{s,\alpha}^f$. As stated in the last subsection, the initial distribution is represented partially by an analytical distribution $g_\alpha^{+,c}$, and partially by particles, and therefore the free stream flux $F_{s,\alpha}^f$ is also calculated partially from the reconstructed analytical distribution as $F_{s,\alpha}^{f,w}$, and partially from particles as $F_{s,\alpha}^{f,p}$.
The initial analytical distribution $g_{\alpha}^{+,c}$ is reconstructed as
\begin{equation}\label{f-reconst}
  g_{0,\alpha}^{+,c}(\vec{x},\vec{v})=g_{0,\alpha}^{+,c}+\nabla_x g_{0,\alpha}^{+,c}\cdot\vec{x},
\end{equation}
which gives
\begin{equation*}
  F_{s,\alpha}^{f,w}=\int \vec{v}\cdot\vec{n}\left(C_4 g^+_{0,\alpha}+C_5\vec{v}\cdot\nabla_xg^+_{0,\alpha}\right)\vec{\Psi}d\Xi,
\end{equation*}
where the time integration coefficients are
 \begin{equation*}
   \begin{aligned}
     &C_4=\frac{\tau_\alpha}{\Delta t}\left(1-e^{-\Delta t/\tau_\alpha}\right)-\Delta t\left(1-e^{-\Delta t/\tau_\alpha}\right),\\
     &C_5=\tau e^{-\Delta t/\tau_\alpha}-\frac{\tau_\alpha^2}{\Delta t}\left(1-e^{-\Delta t/\tau_\alpha}\right)
     -\frac{\Delta t^2}{2}\left(1-e^{-\Delta t/\tau_\alpha}\right).
   \end{aligned}
 \end{equation*}
The net particle flux $F_{s,\alpha}^{f,p}$ is calculated as
\begin{equation*}
  F_{s,\alpha}^{f,p}=\frac{1}{\Delta t}\left(\sum_{k\in P_{\partial \Omega_i^+},\alpha}\vec{W}_{P_{k,\alpha}}-\sum_{k\in P_{\partial \Omega_i^-},\alpha}\vec{W}_{P_{k,\alpha}}\right),
\end{equation*}
where $\vec{W}_{P_{k,\alpha}}=\left(m_{k,\alpha},m_{k,\alpha}\vec{v}_{k,\alpha},\frac12m_{k,\alpha}\vec{v}^2_{k,\alpha}\right)$, $P_{\partial \Omega_i^-,\alpha}$ is the index set of the particles stream out cell $\Omega_i$ during a time step, and
$P_{\partial \Omega_i^+,\alpha}$ is the index set of the particles stream in cell $\Omega_i$.
Finally, the finite volume scheme for conservative variables is
\begin{equation}\label{UGKWP-w}
  \vec{W}^{n+1}_{i,\alpha}=\vec{W}^{n}_{i,\alpha}-\sum_s \frac{\Delta t}{|\Omega_i|}|l_s|F^{eq}_{s,\alpha}-\sum_s \frac{\Delta t}{|\Omega_i|}|l_s|F^{fr,w}_{s,\alpha}
  +\frac{\Delta t}{|\Omega_i|}F_{s,\alpha}^{f,p}
  +\frac{\Delta t}{\tau}(\vec{W}^{* n+1}_{i,\alpha}-\vec{W}^{n+1}_{i,\alpha})
\end{equation}
To solve $\vec{W}^{n+1}_{i,\alpha}$ from Eq.\eqref{UGKWP-w}, the following two linear system needs to be solved. The first is the $m\times m$ linear system for $m$ species velocity vector
$\vec{V}^{n+1}_{i}=(\vec{U}_{i,1}^{n+1},\vec{U}_{i,2}^{n+1},...,\vec{U}_{i,m}^{n+1})$,
\begin{equation*}
A_i\vec{V}^{n+1}_i=B_i,
\end{equation*}
where
$B_{\alpha,i}=\rho_{i,\alpha}^n\vec{U}^n_{i,\alpha}-\sum_s\frac{\Delta t }{|\Omega_{i}|}|l_s|F_{s,\alpha}^{\rho u},$
and the matrix $A_i$ read
\begin{equation*}
\begin{aligned}
  &(A_i)_{\alpha \alpha}=\rho_{i,\alpha}^{n+1}+2\Delta tn^{n+1}_{i,\alpha}\sum_{\beta=1 \atop \beta\neq\alpha}^m \mu_{\alpha \beta}\kappa_{\alpha \beta}n_{i,\beta}^{n+1}\\
  &(A_i)_{\alpha \beta}=-2\Delta tn_{i,\alpha}^{n+1}\mu_{\alpha \beta}\kappa_{\alpha \beta}n_{i,\beta}^{n+1}.
\end{aligned}
\end{equation*}
The second $m\times m$ linear system is for $m$ species internal energies $\vec{e}^{n+1}_i=(e^{n+1}_{i,1},e^{n+1}_{i,2},...e^{n+1}_{i,m})$
\begin{equation*}
C_i\vec{e}^{n+1}_i=D_i,
\end{equation*}
where
\begin{equation*}
\begin{aligned}
  D_{i,\alpha}=&E_{i,\alpha}^n-\sum_s\frac{\Delta t }{|\Omega_{i}|}|l_s|F_{s,\alpha}^{E}
  -\frac12\rho_{i,\alpha}^{n+1}(\vec{U}_{i,\alpha}^{n+1})^2
  +\frac{\Delta t \rho_{i,\alpha}^{n+1}}{2\tau_{i,\alpha}^{n+1}}
  \left((\vec{U}_{i,\alpha}^{* n+1})^2-(\vec{U}_{i,\alpha}^{n+1})^2\right)\\
  &-\frac{\Delta t \rho^{n+1}_{i,\alpha}}{2\tau_{i,\alpha}^{n+1}}
  (\vec{U}_{i,\alpha}^{* n+1}-\vec{U}_{i,\alpha}^{n+1})^2
  +\Delta t n_{i,\alpha}^{n+1} \sum_{\beta=1 \atop \beta\neq\alpha}^{m}\mu_{\alpha \beta}\kappa_{\alpha \beta}
  \frac{2\rho_{i,\beta}^{n+1}}{m_\alpha+m_\beta}(\vec{U}_{i,\alpha}^{ n+1}-\vec{U}_{i,\alpha}^{n+1})^2,
\end{aligned}
\end{equation*}
and
\begin{equation*}
\begin{aligned}
  &(C_i)_{\alpha \alpha}=n_{i,\alpha}^{n+1}+\Delta t n_{i,\alpha}^{n+1} \sum_{\beta=1 \atop \beta\neq\alpha}^{m}\mu_{\alpha \beta}\kappa_{\alpha \beta}\frac{4n_{i,\beta}^{n+1}}{m_\alpha+m_\beta},\\
  &(C_i)_{\alpha \beta}=-\Delta t n_{i,\alpha}^{n+1}\mu_{\alpha \beta} \kappa_{\alpha \beta}\frac{4n_{i,\beta}^{n+1}}{m_\alpha+m_\beta}.
\end{aligned}
\end{equation*}
Under the assumption of non-vacuum solutions ($\rho_{i,\alpha}^n>0$), each system admits a unique solution.
The evolution of the microscopic velocity distribution and macroscopic quantities compose the UGKWP method for multi-species gas mixture.

\subsection{UGKWP method for plasma transport}
In this subsection, the UGKWP method for plasma transport will be proposed, which is the UGKWP method for multi-species coupled with the electromagnetic field.
We split the BGK-Maxwell equations into the transport equations and the interaction equations.
The transport equations including the electron ion transport and electromagnetic wave transport can be written as
\begin{equation*}
\begin{aligned}
  &\partial_t f_\alpha
  +\vec{v}\cdot \nabla_x f_{\alpha}
 =\frac{g_\alpha-f_\alpha}{\tau_\alpha},\\
  &\frac{\partial \vec{B}}{\partial t}+\nabla_x \times \vec{E}=0,\\
  &\frac{\partial \vec{E}}{\partial t}-c^2\nabla_x \times \vec{B}=0,
\end{aligned}
\end{equation*}
and the interaction equations are
\begin{equation}\label{interaction}
\begin{aligned}
  &\partial_t f_\alpha
  +\frac{e_\alpha}{ m_\alpha r}(\vec{E}+\vec{v}\times \vec{B})
  \cdot\nabla_v f_{\alpha}=0,\\
  &\frac{\partial \vec{E}}{\partial t}=
  -\frac{1}{\hat{\lambda}_D^2 r} \vec{j}.
\end{aligned}
\end{equation}
In the next two subsections, the numerical evolution equations for the transport equations and interaction equations will be presented respectively.

\subsubsection{Evolution equations for the transport equations}
In the transport equations, the electron and ion transport is decoupled from the electromagnetic wave transport. The numerical  evolution equation for the electron and ion transport is the UGKWP method presented in Section \ref{ugkwp-mixture}. The Yee-grid based Crank-Nicolson scheme proposed by Yang el al. is used as the evolution equation for the electromagnetic wave transport \cite{yang2006unconditionally}.
The semi-implicit discretization of transverse electric wave equation on Yee mesh can be written as
\begin{equation}\label{TE1}
\begin{aligned}
  E_x^{n+1}\left(x_{i+\frac12},y_j\right)=
  &E^n_x\left(x_{i+\frac12},y_j\right)+
  \frac{\Delta t c^2}{2\Delta y}
  \left(B_z^{n+1}(x_{i+\frac12},y_{j+\frac12})-
  B_z^{n+1}(x_{i+\frac12},y_{j-\frac12})\right)\\
  &+\frac{\Delta t c^2}{2\Delta y}\left(B_z^{n}(x_{i+\frac12},y_{j+\frac12})-
  B_z^{n}(x_{i+\frac12},y_{j-\frac12})\right),
\end{aligned}
\end{equation}
\begin{equation}\label{TE2}
\begin{aligned}
  E_y^{n+1}\left(x_i,y_{i+\frac12}\right)=
  &E^n_x\left(x_i,y_{j+\frac12}\right)-
  \frac{\Delta t c^2}{2\Delta x}
  \left(B_z^{n+1}(x_{i+\frac12},y_{j+\frac12})-
  B_z^{n+1}(x_{i-\frac12},y_{j+\frac12})\right)\\
  &-\frac{\Delta t c^2}{2\Delta x}
  \left(B_z^{n}(x_{i+\frac12},y_{j+\frac12})-
  B_z^{n}(x_{i-\frac12},y_{j+\frac12})\right).
\end{aligned}
\end{equation}
And the semi-implicit discretization of magnetic wave equation is
\begin{equation}\label{TE3}
\begin{aligned}
  B_z^{n+1}\left(x_{i+\frac12},y_{i+\frac12}\right)=
  &B_z^{n+1}\left(x_{i+\frac12},y_{i+\frac12}\right)+
  \frac{\Delta t}{2\Delta y}
  \left(E_x^{n+1}(x_{i+\frac12},y_{j+1})-
  E_x^{n+1}(x_{i+\frac12},y_{j})\right)\\
  &+\frac{\Delta t}{2\Delta y}
  \left(E_x^{n}(x_{i+\frac12},y_{j+1})-
  E_x^{n}(x_{i+\frac12},y_{j})\right)\\
  &-\frac{\Delta t}{2\Delta x}
  \left(E_y^{n+1}(x_{i+1},y_{j+\frac12})-
  E_y^{n+1}(x_{i},y_{j+\frac12})\right)\\
  &+\frac{\Delta t}{2\Delta x}
  \left(E_y^{n}(x_{i+1},y_{j+\frac12})-
  E_y^{n}(x_{i},y_{j+\frac12})\right)
\end{aligned}
\end{equation}
Substituting Eq.\eqref{TE1} and Eq.\eqref{TE2} into Eq.\eqref{TE3}, an implicit equation for $B_z$ can be derived as
\begin{equation*}
\left[1-\frac{c^2\Delta t^2}{4}(D_{2x}+D_{2y})\right]B^{n+1}_z(x_{i+\frac12},y_{j+\frac12})=
\left[1+\frac{c^2\Delta t^2}{4}(D_{2x}+D_{2y})\right]B^{n}_z(x_{i+\frac12},y_{j+\frac12})+f(E_x^n,E_y^n),
\end{equation*}
which can be effectively solved by Douglas-Gunn algorithm.
The advantage of the Yee-grid based Crank-Nicolson scheme is that the divergence constraint of the Maxwell equation is numerically preserved; the dispersion and dissipation error is lower than the FDTD method; and the scheme is unconditionally stable, which removes the CFL constraint on time step.
\subsubsection{Evolution equations for the interaction equations}
Taking conservative moments on Eq.\eqref{interaction}, one gets the macroscopic interaction equations
\begin{equation*}
\left\{
\begin{aligned}
  &\frac{\rho_\alpha\vec{U}}{\partial t}=\frac{e_\alpha n_\alpha}{r}(\vec{E}+\vec{U}\times\vec{B}),\\
  &\frac{\partial \vec{E}}{\partial t}=
  -\frac{1}{\hat{\lambda}_D^2 r} \vec{j}.
\end{aligned}
\right.
\end{equation*}
The implicit discretization of the macroscopic interaction equations gives the following linear system,
\begin{equation}\label{implicit1}
\left\{
\begin{aligned}
&\rho^{n+1}_i \vec{U}_i^{n+1}-\rho^{n+1}_i \vec{U}_i^{n}=\frac{\Delta t}{r}n_i^{n+1}( \vec{E}^{n+1}+ \vec{U}^{n+1}_i \times \vec{B}^{n+1}),\\
&\rho^{n+1}_e \vec{U}_e^{n+1}-\rho^{n+1}_e \vec{U}_e^{n}=-\frac{\Delta t}{r}n_e^{n+1}(\vec{E}^{n+1}+\vec{U}^{n+1}_e \times \vec{B}^{n+1}),\\
&\vec{E}^{n+1}-\vec{E}^n=-\frac{\Delta t}{\lambda^2_D r}\left(\vec{j}_i^{n+1}+\vec{j}_e^{n+1}\right),
\end{aligned}
\right.
\end{equation}
from which the electromagnetic field and macroscopic flow variables are updated to $t^{n+1}$, and the velocity of the simulation particles is updated by
\begin{equation*}
\vec{v}^{n+1}_{k,\alpha}=\vec{v}^{n}_{k,\alpha}+\frac{\Delta te_\alpha}{ m_\alpha r}(\vec{E}^{n+1}+\vec{v}_{k,\alpha}\times \vec{B}^{n+1}).
\end{equation*}
The evolutions of the transport equations and interaction equations compose of the UGKWP method for plasma transport.

\section{Analysis and discussion}\label{discussion}
\subsection{Unified preserving and asymptotic complexity diminishing properties of UGKWP method}
In this section, the multiscale property of UGKWP method will be discussed, and the computational complexity will be estimated. Guo et al. proposes the unified preserving property which assesses the accuracy of a kinetic scheme in continuum regime \cite{guo2019unified}. Crestetto et al. proposes the asymptotic complexity diminishing property of a kinetic scheme which assesses the computational complexity of a kinetic scheme in continuum regime \cite{crestetto2019asymptotically}. In the following proposition, we show that the UGKWP method is a second order UP scheme and an asymptotic complexity diminishing scheme.
\begin{proposition}\label{up}
Holding the mesh size and time step, the UGKWP method satisfies:
\begin{enumerate}
  \item The scheme degenerates to collisionless Boltzmann equation as the local relaxation parameter $\tau\to\infty$.
  \item The scheme becomes a second order scheme for Navier-Stokes equations $\tau\to0$.
  \item The total degree of freedom of the scheme $N_f\to N_f^h$ as $\tau\to0$, where $N_f^h$ is the freedom of the hydrodynamic equations.
\end{enumerate}
\end{proposition}
\begin{proof}
\begin{enumerate}
  \item In the collisionless limit, we have
  \begin{equation}
    \lim_{\tau \to \infty} t_{f,\alpha}=\lim_{\tau \to \infty} (-\tau \ln(\eta))\to \infty.
  \end{equation}
  Therefore, all particles will be streamed for
  $\min(\Delta t,t_{f,\alpha})=\Delta t$.
    And the UGKWP method solves collisionless Boltzmann equation in collisionless regime.
\item In the continuum regime when $\tau\to0$, we have
 \begin{equation*}
\begin{aligned}
  g^+_\alpha(\vec{x},t,\vec{v})&=g_\alpha(\vec{x},t,\vec{v})
  +\left(\frac{te^{-t/\tau}}{1-e^{-t/\tau}}-\tau\right)
  \left(\partial_tg_\alpha(\vec{x},t,\vec{v})
  +\vec{v}\cdot\nabla_xg_\alpha(\vec{x},t,\vec{v})\right)\\
&=g_\alpha(\vec{x},t,\vec{v})
  -\tau\left(\partial_tg_\alpha(\vec{x},t,\vec{v})
  +\vec{v}\cdot\nabla_xg_\alpha(\vec{x},t,\vec{v})\right)+O(e^{-\Delta t/\tau})\\
&=g_\alpha(\vec{x},t,\vec{v})
  -\tau\left(\partial_t\bar{g}(\vec{x},t,\vec{v})
  +\vec{v}\cdot\nabla_x\bar{g}(\vec{x},t,\vec{v})\right)+O(\tau^2)
\end{aligned}
\end{equation*}
The analytic flux $F^{an}_\alpha$ of macroscopic variables, namely the equilibrium flux and free streaming flux by analytic distribution function satisfies
\begin{equation*}
  \begin{aligned}
   F^{an}_\alpha=&F^{eq}_\alpha+F^{f,w}_\alpha\\
   =&\int \vec{v}\cdot\vec{n}\left(C_1g_{0,\alpha}+C_2\vec{v}\cdot \nabla_x g_{0,\alpha}+C_3\partial_t g_{0,\alpha}\right)\vec{\Psi}d\Xi
   +\int \vec{v}\cdot\vec{n}\left(C_4 g^+_\alpha+C_5\vec{v}\cdot\nabla_xg^+_\alpha\right)\vec{\Psi}d\Xi\\
  =&\int \vec{v}\cdot\vec{n}\left((C_1+C_4)g_{0,\alpha}+(C_2-\tau C_4+C_5)\vec{v}\cdot \nabla_x g_{0,\alpha}+(C_3-\tau C_4)\partial_t g_{0,\alpha}\right) \vec{\Psi}d\Xi\\
  =&\int \vec{v}\cdot\vec{n}\left(g_{0,\alpha}-\tau\vec{v}\cdot \nabla_x g_{0,\alpha}+\left(\frac12\Delta t-\tau\right)\partial_t g_{0,\alpha}\right) \vec{\Psi}d\Xi+O(e^{-\Delta t/\tau})\\
  =&\int \vec{v}\cdot\vec{n}\left(g_{0,\alpha}-\tau\vec{v}\cdot \nabla_x \bar{g}_{0}+\left(\frac12\Delta t-\tau\right)\partial_t \bar{g}_{0}\right) \vec{\Psi}d\Xi+O(\tau^2)
  \end{aligned}
 \end{equation*}
 The sampled particle mass in UGKWP method is
   $e^{-\Delta t/\tau}\rho^h_\alpha \Omega_x$
 and therefore the net free streaming flow contributed by
particles passing through the cell interface, $F_{s,\alpha}^{f,p}\sim O(e^{-\Delta t/\tau})$, diminishes.
As $\tau\to0$, Eq.\eqref{UGKWP-w} exponentially converges to
\begin{equation}\label{ugkwp-ns}
\begin{aligned}
  \vec{W}^{n+1}_{i,\alpha}=&\vec{W}^{n}_{i,\alpha}-\sum_s \frac{\Delta t}{|\Omega_i|}|l_s| \int \vec{v}\cdot\vec{n}\left(g_{0s,\alpha}-\tau(\partial_t \bar{g}_{0,s}+\vec{v}\cdot \nabla_x \bar{g}_{0,s})+\frac12\Delta t\partial_t \bar{g}_{0,s}\right) \vec{\Psi}d\Xi\\
&+\frac{\Delta t}{\tau}(\vec{W}^{* n+1}_{i,\alpha}-\vec{W}^{n+1}_{i,\alpha})
\end{aligned}
  \end{equation}
It can be observed that the numerical flux of conservative variables is consistent with the Navier-Stokes flux given by first order Chapman-Enskog expansion Eq.\eqref{ceexpansion}. Therefore in the continuum regime, the UGKWP method converges to Eq.\eqref{ugkwp-ns}, which is a second order gas-kinetic Navier-Stokes solver \cite{gks-2001}, i.e., the same as the direct macroscopic NS solver in smooth flow region.
\item As $\tau\to0$, the total mass of simulation particle $ e^{-\Delta t/\tau}\rho^h\Omega_x\to0$, and therefore the number of simulation particles $N_p\to 0$ in continuum regime. As $\tau\to0$, the total degree of freedom $N_f=N_f^h+N_p\to N_f^h$, and the UGKWP method is an asymptotic complexity diminishing scheme.
\end{enumerate}
\end{proof}

\subsection{Asymptotic preserving property of UGKWP method for plasma transport}
Property \ref{up} states that the UGKWP method for plasma transport preserves the two fluid model in the hydrodynamic regime. In this subsection, the behavior of the UGKWP method in the highly magnetized regime is discussed.
\begin{proposition}
In the highly magnetized regime as $r \to 0,\lambda_D=c^{-1}$, the linear system Eq.\eqref{implicit1} is consistent to the magnetohydrodynamic equations.
\end{proposition}
\begin{proof}
The Crank-Nicolson scheme for electromagnetic wave propagation gives
\begin{equation*}
\vec{E}^{n+1}-\vec{E}^n=\Delta t c^2\nabla\times\vec{B}+O(\Delta t^2,\Delta x^2).
\end{equation*}
The implicit discretization of the macroscopic interaction equations Eq.\eqref{implicit1} gives
\begin{equation*}
\begin{aligned}
\vec{j}^{n+1}_i+\vec{j}^{n+1}_e&=\frac{\Delta t}{r}\lambda_{D}^2c^2(\vec{E}^{n+1}-\vec{E}^n)\\
&= r\nabla\times\vec{B}+O(\Delta t^2,\Delta x^2),
\end{aligned}
\end{equation*}
and therefore the total momentum equation gives
\begin{equation*}
\begin{aligned}
\rho^{n+1} \vec{U}^{n+1}-\rho^{n+1} \vec{U}^{n}&=\frac{\Delta t}{r}\left[ (\vec{j}^{n+1}_i+\vec{j}^{n+1}_e) \times \vec{B}^{n+1}\right]\\
&=\Delta t \nabla\times\vec{B}\times\vec{B}+O(\Delta t^2,\Delta x^2),
\end{aligned}
\end{equation*}
which converges to a consistent MHD scheme.
\end{proof}
\section{Numerical tests}\label{tests}
Five numerical tests are carried out in this section to verify the performance of the UGKWP method in various flow regime, including
three 1D and two 2D tests. Firstly, the shock structure of binary gas mixture is calculated to show the capability of the UGKWP method in capturing the flow non-equilibrium in the rarefied regime. The second test is the Landau damping and two steam instability, showing that the scheme can accurately capture the interaction between plasma and electromagnetic field. The Brio-Wu and Orszag-Tang tests verifies the performance of the UGKWP method in different flow regimes. In the last, the scheme is applied to the magnetic reconnection problem to study how the electron ion collision affects the reconnection rate.
\subsection{Shock structure of binary gas mixture}
Normal shock structure is a standard test that verifies the ability of the scheme in capturing the non-equilibrium effect in rarefied regime. In this test, the Mach number is set as $\text{M}=1.5$, the mass ratio of gas mixture is  $m_B/m_A=0.5$, diameter ratio $d_B/d_A=1$, and the component concentration of B is $\chi_B=0.1$. The hard sphere model is used and the reference mean free path is defined by
\begin{equation*}
\lambda_\infty=\frac{1}{\sqrt{2}\pi d_A^2n_1},
\end{equation*}
For each component, the upstream and downstream conditions are related through Rankine-Hugoniot condition.
The cell size is chosen to be $\Delta x=0.5\lambda_\infty$, and CFL number is 0.95.
The mass of simulation particle is $m_{p,\alpha}=10^{-2}$, which corresponds to around a hundred of simulation particles per cell.
The normalized density, velocity and temperature are compared to the reference UGKS solution \cite{wang2015unified}, as shown in Fig.\ref{shock}. The UGKWP results well agree with the UGKS solution, which shows the capability of UGKWP in capturing the non-equilibrium flow physics.

\subsection{Landau damping and two steam instability}
The Landau damping and two steam instability are two classical phenomenons that have been well studied theoretically, and therefore we choose these two cases to test the accuracy of UGKWP method in capturing the interaction between plasma and electromagnetic field.
First we consider the Landau damping.
Consider a Vlasov-Poisson system that perturbed by a weak signal.
The linear theory of Landau damping can be applied to predict the linear decay of electric energy with time \cite{chen1984introduction}.
The initial condition of linear Landau damping is
\begin{equation}\label{landaudamping-initial}
  f_0(x,u)=\frac{1}{\sqrt{2\pi}}\left(1+\alpha\cos(kx)\right)\text{e}^{-\frac{u^2}{2}},
\end{equation}
with $\alpha=0.01$.
The length of the domain in the x direction is $L=2\pi/k$.
The background ion distribution function is fixed,
uniformly chosen so that the total net charge density for the system is zero.
When perturbation parameter $\alpha=0.01$ is small enough,
the Vlasov-Poisson system can be approximated by linearization around the Maxwellian equilibrium.
The analytical damping rate of electric field can be derived accordingly. Numerical cell number in physical space is $N_x=128$, and the particle number in each cell is $N_p=1000$.
We test our scheme with different wave numbers and compare the numerical damping rates with theoretical values.
For wave numbers $k=0.3$ and $k=0.4$, the evolution of the $L^2$ norm electric field is plotted in Fig. \ref{Landau1}.
It can be observed that the decay rates and oscillating frequencies $\omega=1.16, 1.29$ agree well with theoretical data.

Once a larger perturbation $\alpha=0.5$ and $k=0.5$ is applied, the linear theory breaks down, and the nonlinear phenomenon occurs. The evolution of the electric energy calculated by UGKWP method is shown in Fig. \ref{Landau2} (a),
The linear decay rate of electric energy is approximately equal to $\gamma_1=-0.287$,
which agrees well to the values obtained by Heath et al. \cite{heath}.
The growth rate predicted by UGKWP method is approximately $\gamma_2=0.078$, which is between
the value of $0.0815$ computed by Rossmanith and Seal and $0.0746$ by Heath et al. \cite{rossmanith}.

Next we consider the linear two stream instability problem with initial distribution function:
\begin{equation}\label{instability-initial1}
  f(x,u,t=0)=\frac{2}{7\sqrt{2\pi}}(1+5v^2)(1+\alpha((\cos(2kx)+\cos(3kx))/1.2+\cos(kx)))\text{e}^{-\frac{u^2}{2}},
\end{equation}
with $\alpha=0.001$ and $k=0.2$.
The length of the domain in the x direction is $L=\frac{2\pi}{k}$.
The background ion distribution function is
fixed, uniformly to balance the charge density of electron.
After an initial transition, a linear growth rate of electric field can be theoretically predicted \cite{chen1984introduction}.
We apply the UGKWP method to calculate this two stream instability problem with physical cell number $N_x=512$ and simulation particle number $N_p=1000$ per cell. The electric energy result is shown in Fig. \ref{Landau2} (b), and good agreement between UGKWP solution and theoretical value can be observed. The phase space contour at $t=70$ is shown in Fig. \ref{Landau3}. Compared to the UGKS solution, the UGKWP solution provides more detailed flow structure.

\subsection{Brio-Wu shock tube}
The Brio-Wu shock tube is originally designed for MHD solvers in continuum regime. Here we calculate the Brio-Wu problem in rarefied (Kn=1), transitional (Kn=$10^{-2}$), and continuum (Kn=$10^{-4}$) regimes.
The same initial condition as the Brio-Wu one is shown in Fig. \ref{Brio-initial}.
The ion to electron mass ratio is set to be 1836, and the ionic charge state is set to be unity.
The normalized Debye length is $\lambda_D=0.001$, the normalized Larmor radius is $r=0.001$, and the normalized speed of light is $1000$.
The grid points in physical space are $N_x=1000$.
And the simulation particle mass is set as $m_p=10^{-5}$.
The UGKWP solutions in rarefied and transitional flow regimes are shown in Fig. \ref{brio1} and \ref{brio2}, and compared to the reference UGKS solution.
In Fig. \ref{brio3}, the UGKWP solution in continuum regime is compared to the reference UGKS solution.
The UGKWP solutions have good agreement with reference solutions in different flow regime.
Especially, it can be observed that the statistical noise significantly reduces as Knudsen number decreases thanks to the asymptotic complexity diminishing property of UGKWP.
In the MHD regime, the UGKWP solution is shown in Fig. \ref{brio4} and compared to the reference ideal-MHD solution.

\subsection{Orszag-Tang Vortex}
The Orszag-Tang Vortex problem was originally designed to study the MHD turbulence \cite{orszag1979small}.
In this work, the problem is calculated in rarefied (Kn=1) and continuum (Kn=$10^{-4}$) regimes to verify the multiscale and asymptotic complexity diminishing property of UGKWP.
The initial data for the current study is
\begin{equation}\nonumber
\begin{aligned}
&m_i/m_e=25, n_i=n_e=\gamma^2, P_i=P_e=\gamma, \ B_y=\sin(2x),\\ &u_{i,x}=u_{e,x}=-\sin(y), \ u_{i,y}=u_{e,y}=\sin(x),
\end{aligned}
\end{equation}
where $\gamma=5/3$ and $r=0.001$.
The computation domain is $[0,2\pi]\times[0,2\pi]$ with a uniform mesh of $200\times200$ cells.
The mass of simulation particle is $m_p=10^{-5}$.
The UGKWP result in rarefied regime is shown in Fig. \ref{orzag1} compared to the reference UGKS solution,
and the UGKWP result in continuum regime is shown in Fig. \ref{orzag2}.
A better agreement and lower noise can be observed in the continuum regime due to the asymptotic preserving and the asymptotic complexity diminishing property of UGKWP.
In the MHD regime, the UGKWP solution with Kn=$10^{-4}$ and $r=0$ are shown in Fig. \ref{orzag3}-\ref{orzag5},
where in Fig. \ref{orzag5}(b) the UGKWP pressure distribution along $y=0.625\pi$ is compared to the MHD solution \cite{tang}.

\subsection{Magnetic reconnection}
Magnetic reconnection is an important phenomenon that transfers magnetic energy into flow energy by topological change of magnetic lines. In this test case, the UGKWP method is used to study the reconnection phenomenon in different flow regime, and study how the particle collision affects the collision rate as well as the topology of magnetic line.
The simulation uses the same initial conditions as the GEM challenge problem \cite{birn2001geospace}.
The initial magnetic field is given by
\begin{equation}\nonumber
  \vec{B}(y)=B_0\tanh(y/\lambda)\vec{e}_x,
\end{equation}
and a corresponding current sheet is carried by the electrons
\begin{equation}\nonumber
  \vec{J}_e=-\frac{B_0}{\lambda} \text{sech} ^2(y/\lambda)\vec{e_z}.
\end{equation}
The initial number densities of electron and ion are
\begin{equation}\nonumber
  n_e=n_i=1/5+\text{sech}^2(y/\lambda).
\end{equation}
The electron and ion pressures are set to be
\begin{equation}\nonumber
  P_i=5P_e=\frac{5B_0}{12}n(y),
\end{equation}
where $B_0=0.1$, $m_i=25m_e$ and $\lambda=0.5$.
The computational domain is $[-L_x/2, L_x/2]\times[-L_y/2,L_y/2]$ with $L_x=8\pi$, $L_y=4\pi$, which is divided into $200\times 100$ cells.
Periodic boundaries are applied at $x=\pm L_x/2$ and conducting wall boundaries at $y=\pm L_y/2$.
To initiate reconnection, the magnetic field is perturbed with $\delta \vec{B}=\vec{e}_z\times\nabla_x\psi$, where
\begin{equation}\nonumber
  \psi(x,y)=0.1B_0\cos(2\pi x/L_x)\cos(\pi y/L_y).
\end{equation}
Two Knudsen numbers are considered, Kn=$10^{-3}$ in the transitional regime and Kn=$10^{-4}$ in the continuum regime.
The magnetic field topology as well as the distribution of flow variables in transitional regime are shown
in Fig. \ref{reconnection1}-\ref{reconnection2} at $\omega_{pi}t=15$ and $\omega_{pi}t=30$,
and the magnetic field topology in continuum regime as well as the magnetic reconnection rate are shown in Fig. \ref{reconnection3}.
In the continuum regime, the topology of the magnetic field is symmetric,
while in the transitional regime a magnetic island appears in the middle region at $\omega_p t=15$
and merges into the big right island at $\omega_p t=30$. Due to the magnetic island,
two x-shape reconnection points form and the reconnection rate in the transitional regime is significantly increased $15<\omega_p t<30$.
After $\omega_p t=30$ when the middle magnetic island merges with the right one,
the reconnection rate slows down to the same reconnection intensity as in the continuum regime, which is shown in Fig. \ref{reconnection3}(b).

\section{Conclusion}\label{conclusion}
In this work, we extend the unified gas-kinetic wave-particle method to the field
of multi-species gas mixture and multiscale plasma transport. The construction of numerical scheme for multiscale transport is
based on the direct modeling methodology \cite{xu2015direct}, where the flow physics is modeled according to the cell size and time
step scales. In the unified framework, the evolution of
 microscopic velocity distribution function is coupled with the evolution of macroscopic quantities in a discretized space.
The evolution solution of microscopic distribution function is modeled from the accumulating effect of particle transport and collision
within a time step, from which numerical fluxes for both macroscopic flow variables and particle distribution function are obtained.
The intrinsic governing equation underlying the unified scheme depends on the local cell's Knudsen number.
A smooth transition from the kinetic particle transport to the continuum hydrodynamic flow evolution can be recovered with the
variation of the cell's Knudsen number.
For the multispecies and plasma transport, the UGKWP has the properties of  second order unified preserving as well as the computational complexity asymptotic  dimensioning. In plasma transport, the UGKWP method provides a smooth transition from PIC method in the kinetic scale to the
MHD flow solver in the continuum regime, all all kinds of MHD equations, such as the two fluid models, become subsets in the UGKWP modeling.
Compared to the discrete velocity method (DVM), the UGKWP is much efficient in the numerical
 simulation of highly non-equilibrium and high-dimensional flow problems.
 In conclusion, the UGKWP method has great potential to solve  multiscale transport problems in rarefied flow \cite{liu2020unified,zhu2019unified}, radiative transfer \cite{li2020unified,shi}, and plasma physics.

\section*{Acknowledgment}
The current research is supported by National Numerical Windtunnel project and  National Science Foundation of China 11772281, 91852114.

\section*{References}
\bibliography{ugkwpplasma}
\bibliographystyle{elsarticle-num}
\biboptions{numbers,sort&compress}

\clearpage

\begin{figure}
\centering
\includegraphics[width=0.9\textwidth]{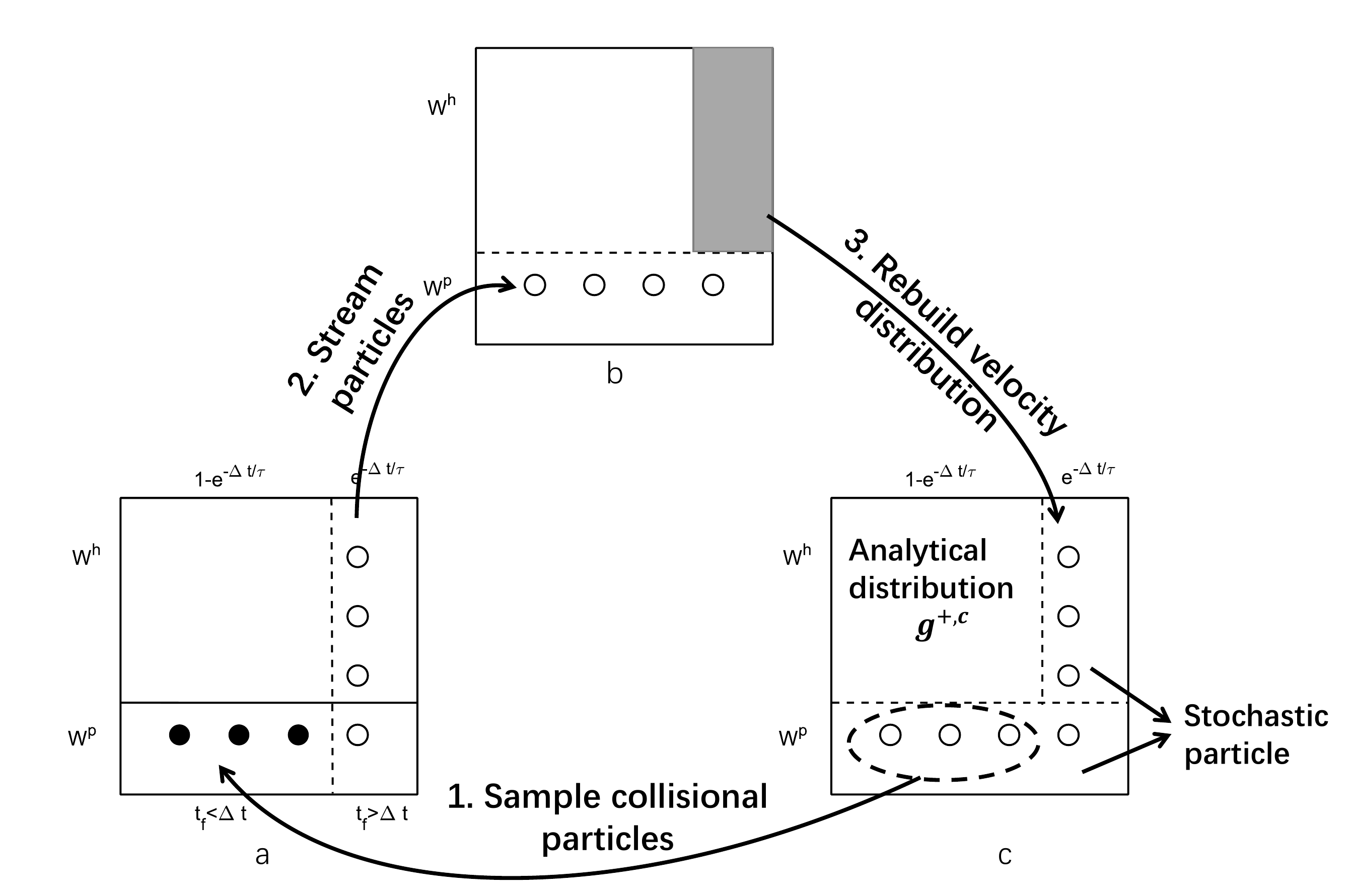}
\caption{The diagram of particle updating procedure: 1. Sample particle free stream time $t_{c,\alpha}$; 2. Stream particles and update macroscopic quantities; 3. Rebuild velocity distribution.}
\label{wpdiagram}
\end{figure}

\begin{figure}
\centering
\includegraphics[width=0.48\textwidth]{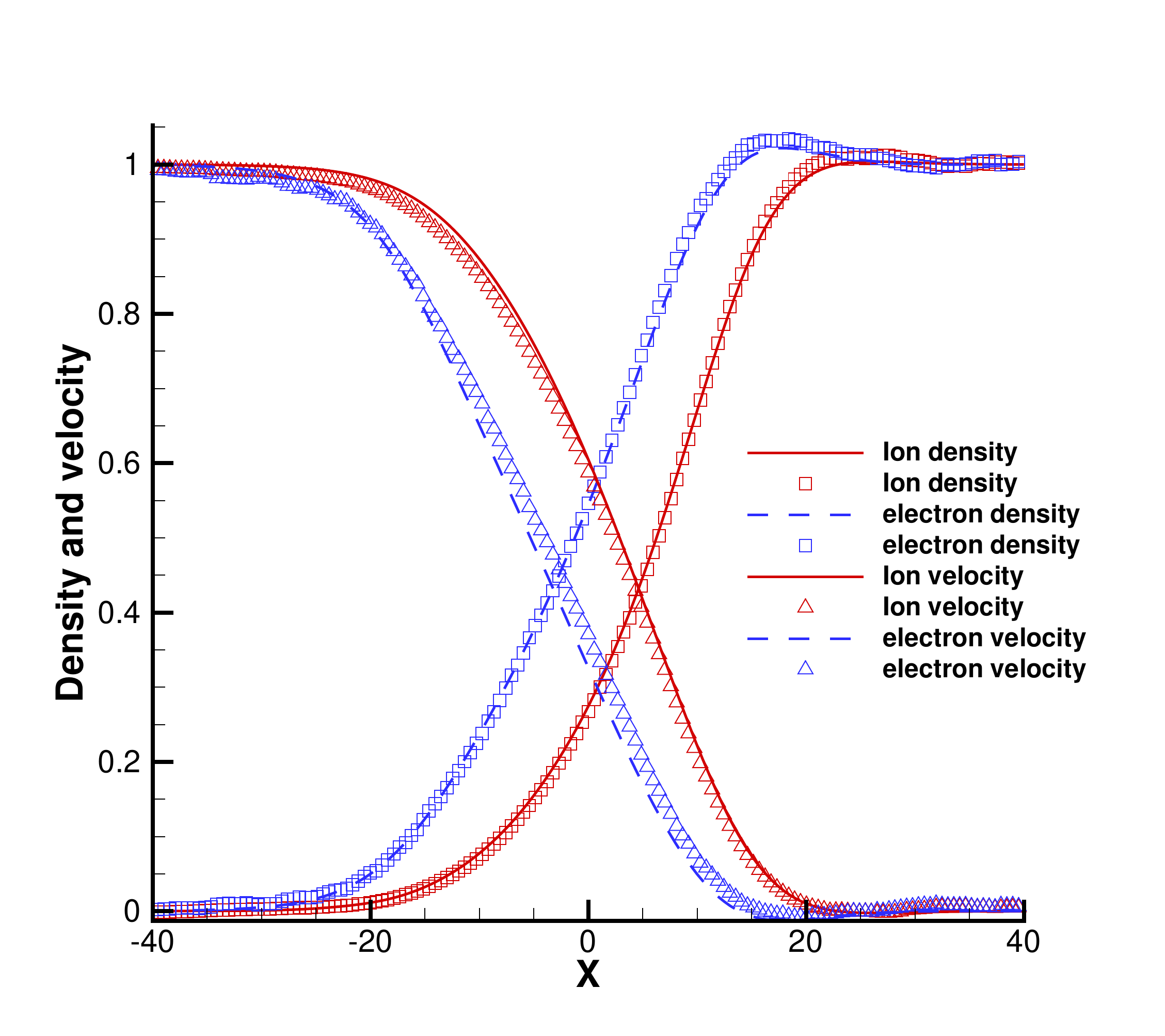}
\includegraphics[width=0.48\textwidth]{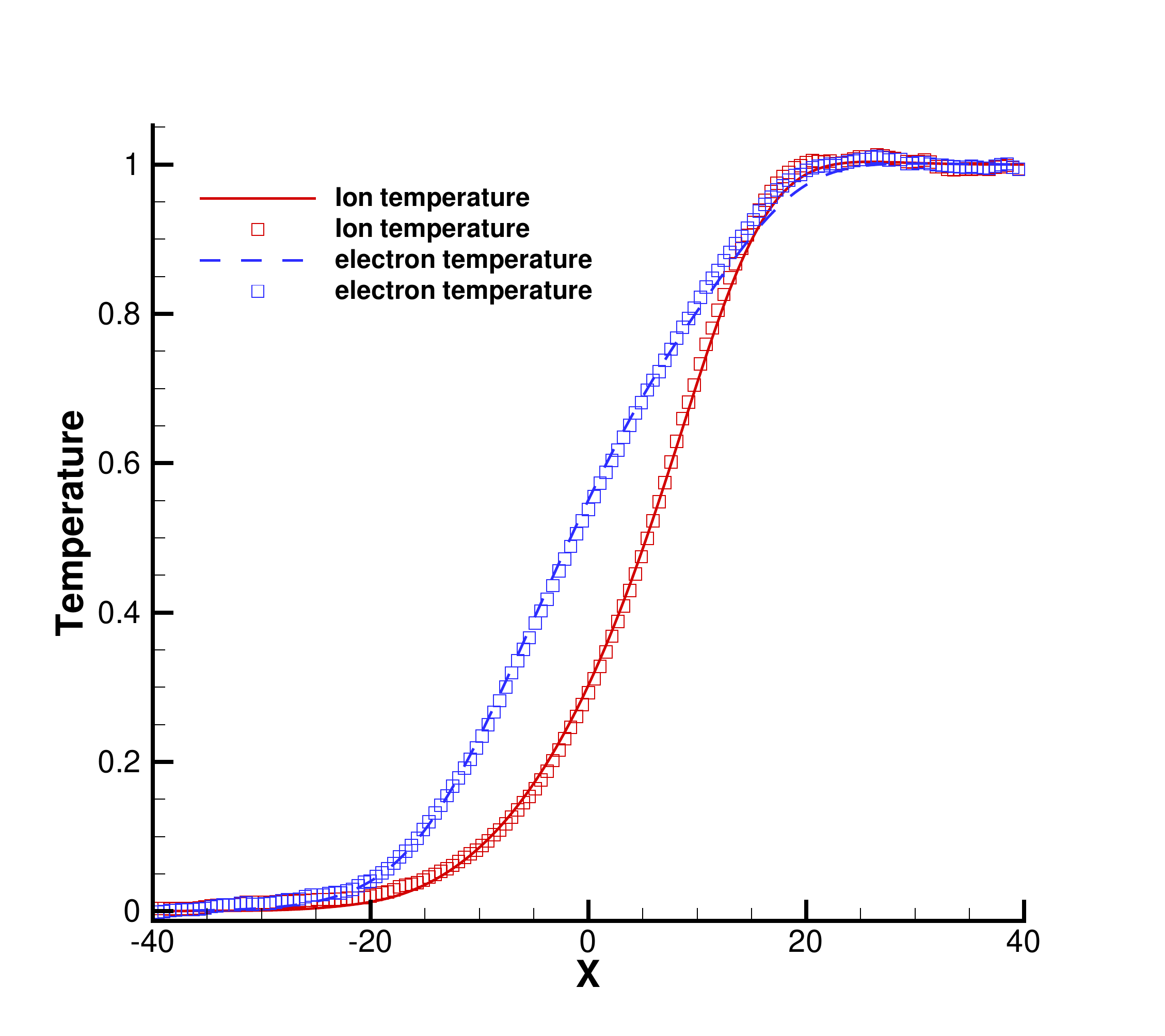}
\caption{Results of the shock wave in a binary gas mixture. Left figure shows the normalized density and velocity. Right figure shows the normalized temperature. Symbols shows the UGKWP solution and lines shows the reference UGKS solution.}
\label{shock}
\end{figure}

\begin{figure}
\centering
\includegraphics[width=0.48\textwidth]{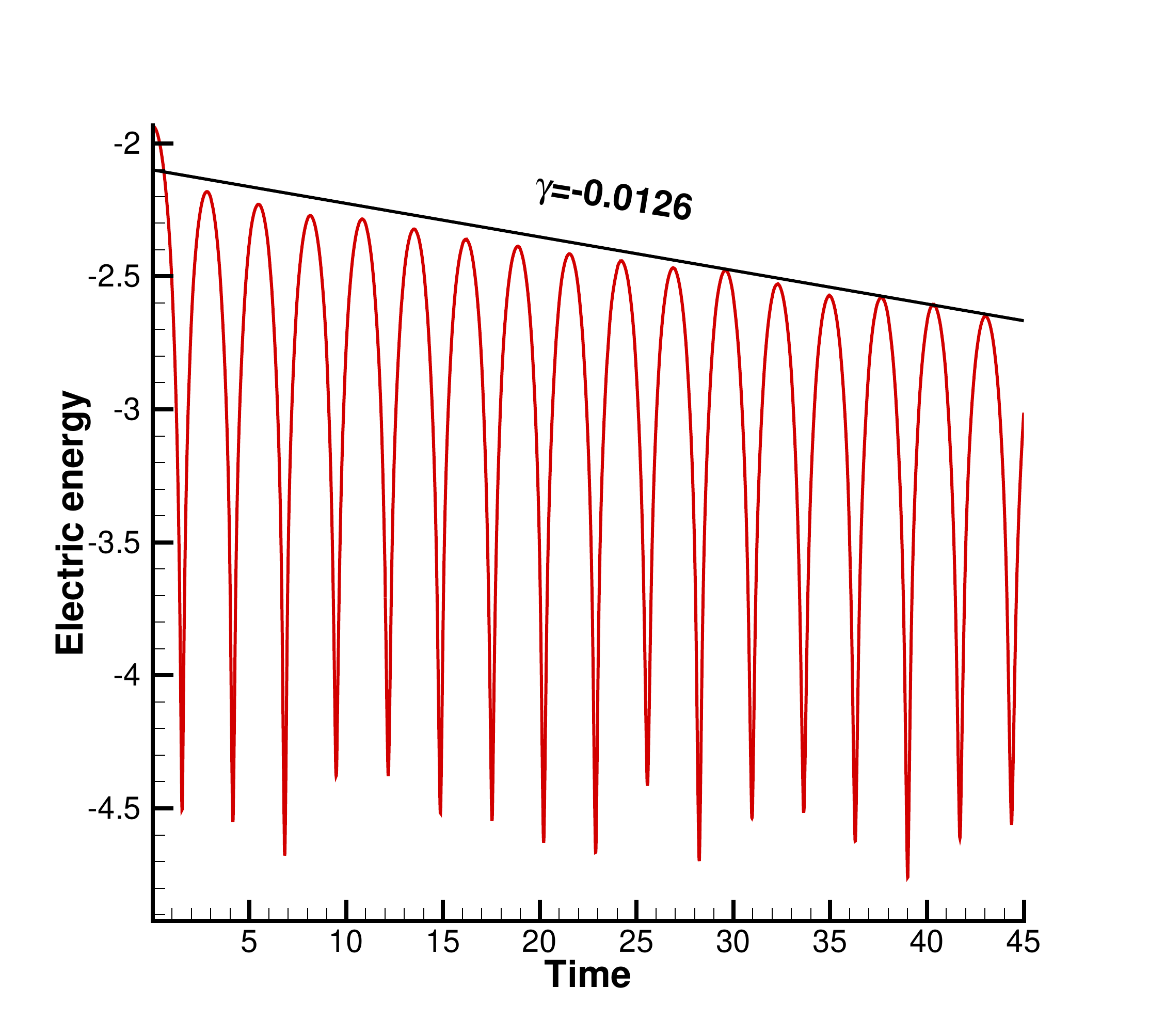}
\includegraphics[width=0.48\textwidth]{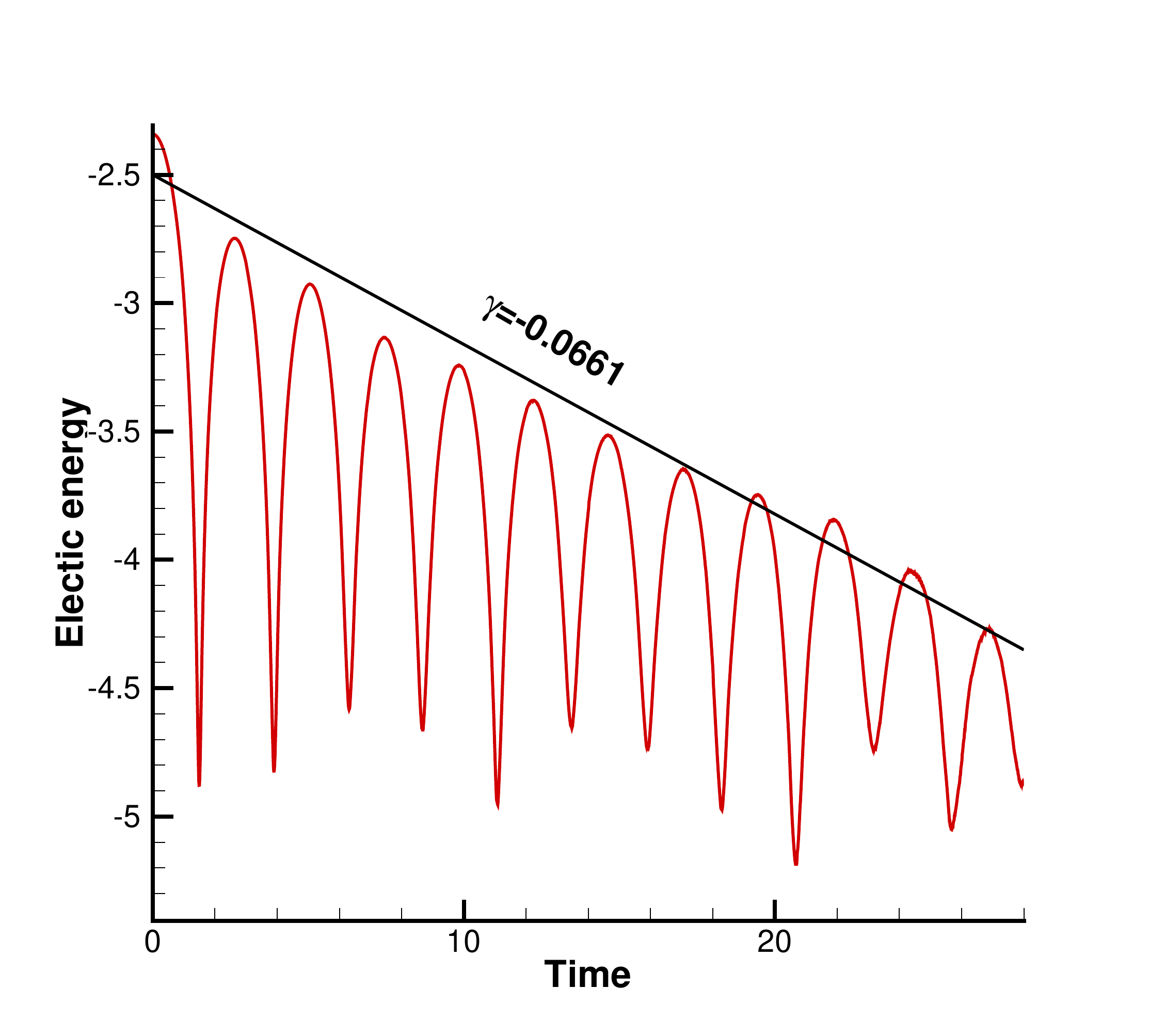}
\caption{Results of Landau damping. Red line shows the time evolution of the electric energy calculated by UGKWP method, and the black line gives the theoretical prediction of the electric energy decay rate.}
\label{Landau1}
\end{figure}

\begin{figure}
\centering
\includegraphics[width=0.48\textwidth]{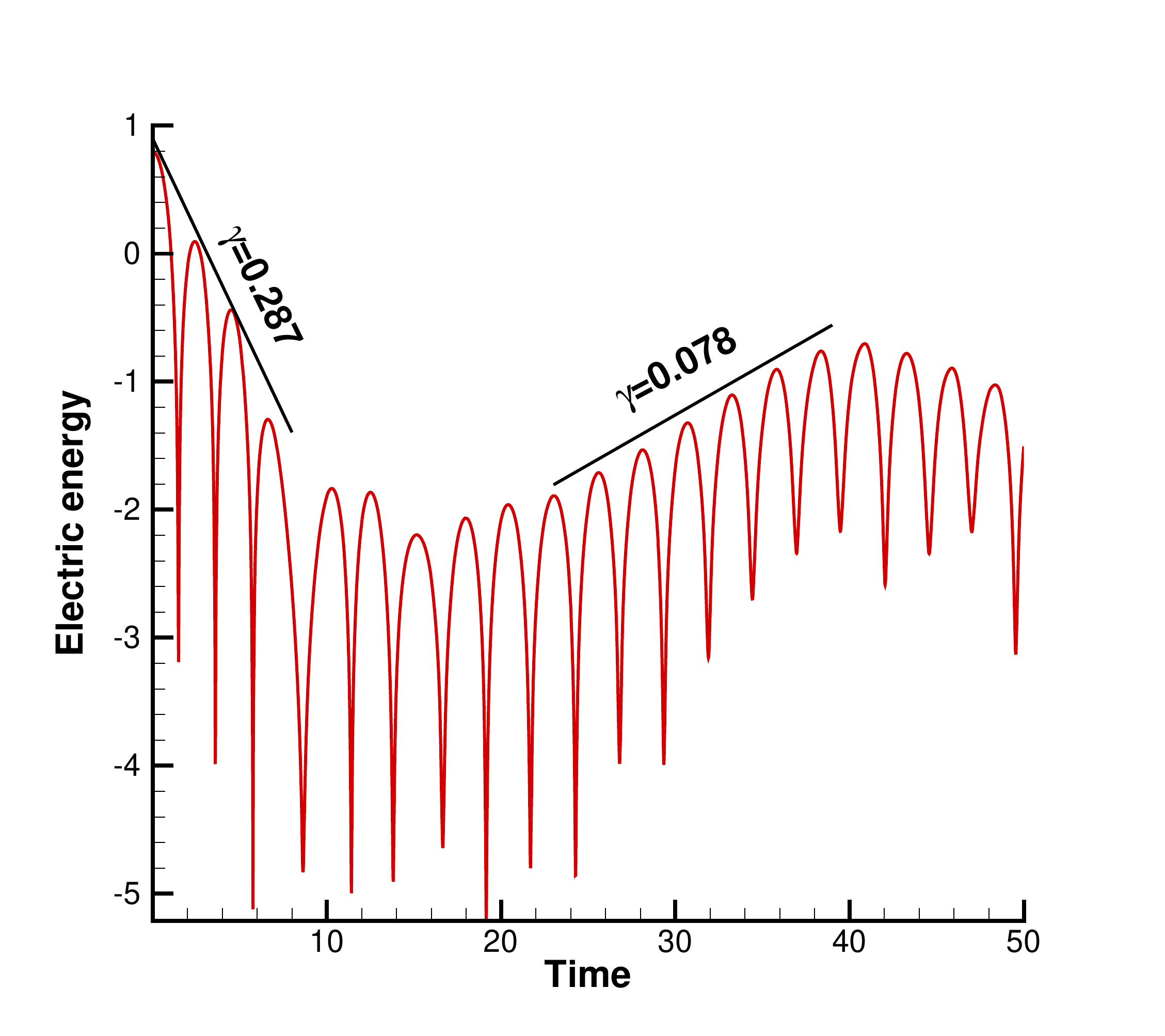}
\includegraphics[width=0.48\textwidth]{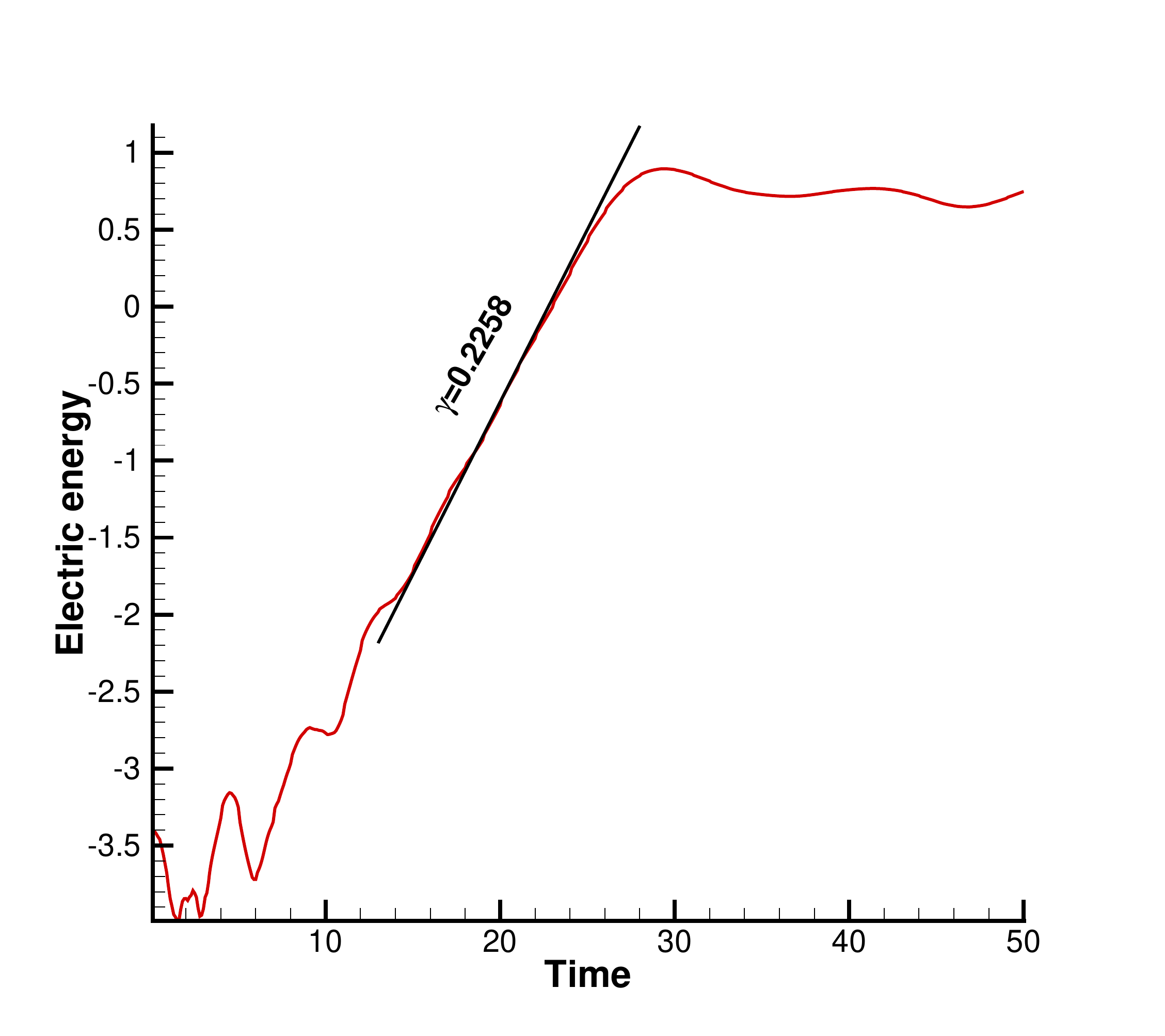}
\caption{Left figure shows the result of nonlinear Landau damping, Red line shows the time evolution of the electric energy calculated by UGKWP method, and the black line shows the reference solutions. Right figure shows the result of two stream instability. Red line shows the time evolution of the electric energy calculated by UGKWP method, and the black line gives the theoretical prediction of the electric energy increase rate.}
\label{Landau2}
\end{figure}

\begin{figure}
\centering
\includegraphics[width=0.49\textwidth]{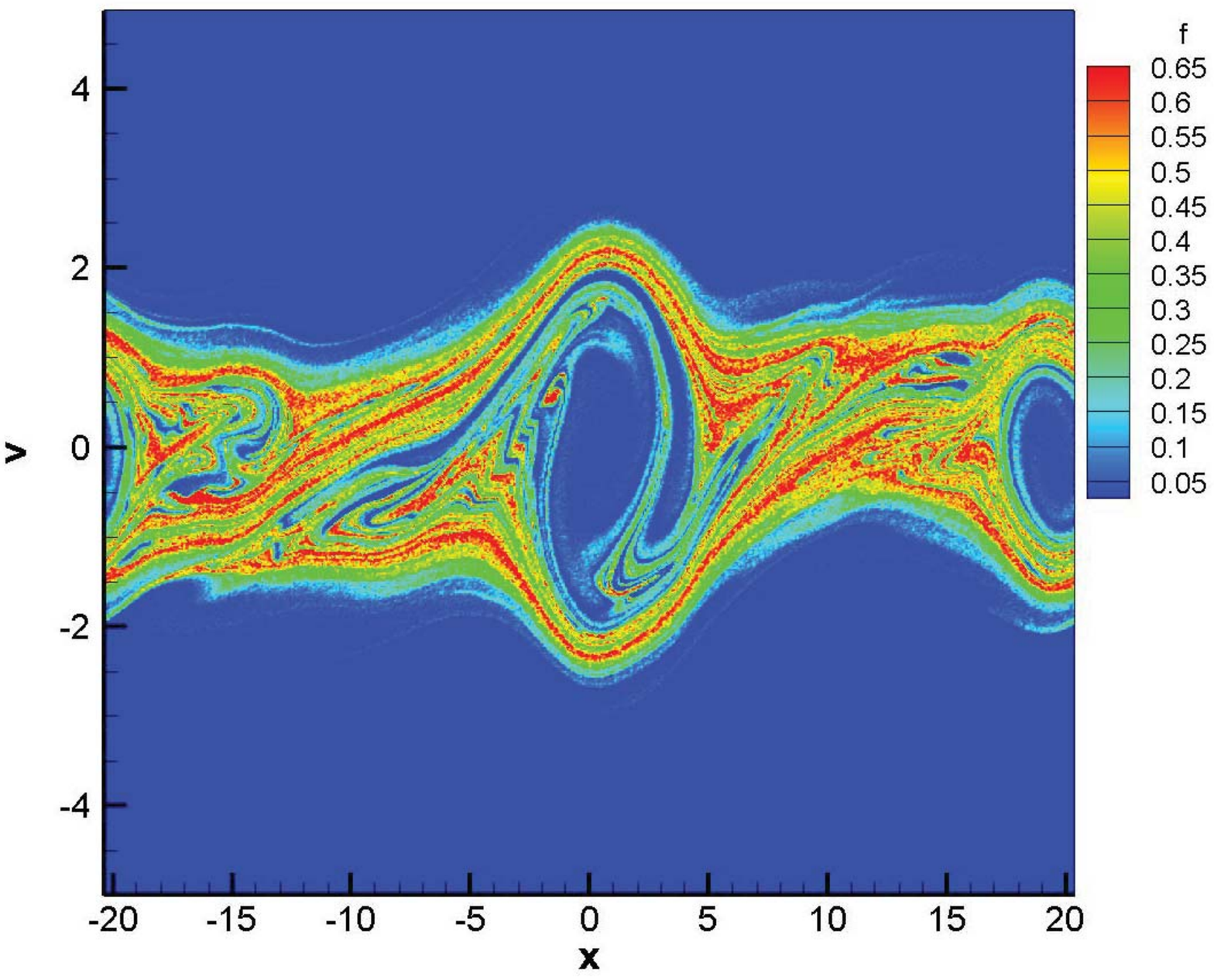}
\includegraphics[width=0.49\textwidth]{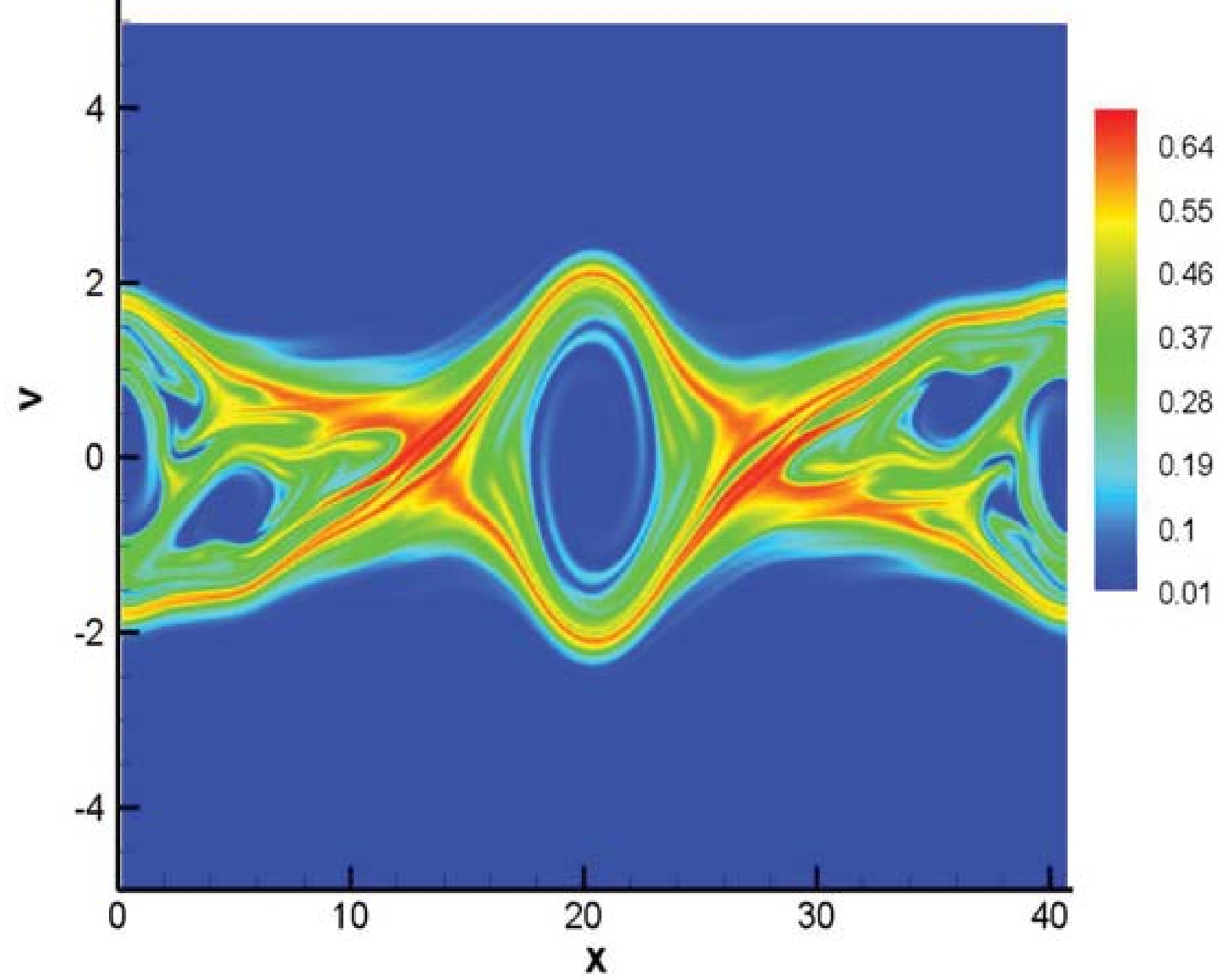}
\caption{Phase space contour results of the two stream instability at $t=70$. Left figure shows the UGKWP result and right figure shows the UGKS solution.}
\label{Landau3}
\end{figure}

\begin{figure}
  \centering
  \includegraphics[width=0.8\textwidth]{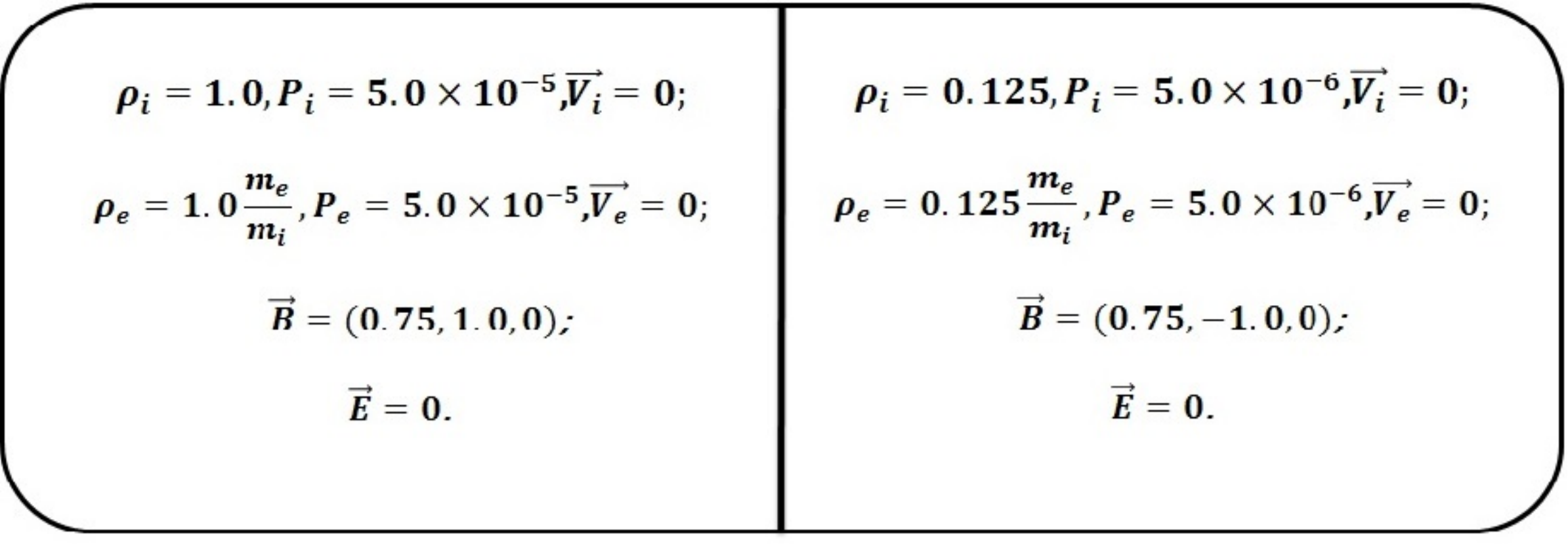}\\
  \caption{Initial condition for the Brio-Wu shock tube problem.}\label{Brio-initial}
\end{figure}

\begin{figure}
     \centering
     \begin{subfigure}[b]{0.48\textwidth}
         \centering
         \includegraphics[width=\textwidth]{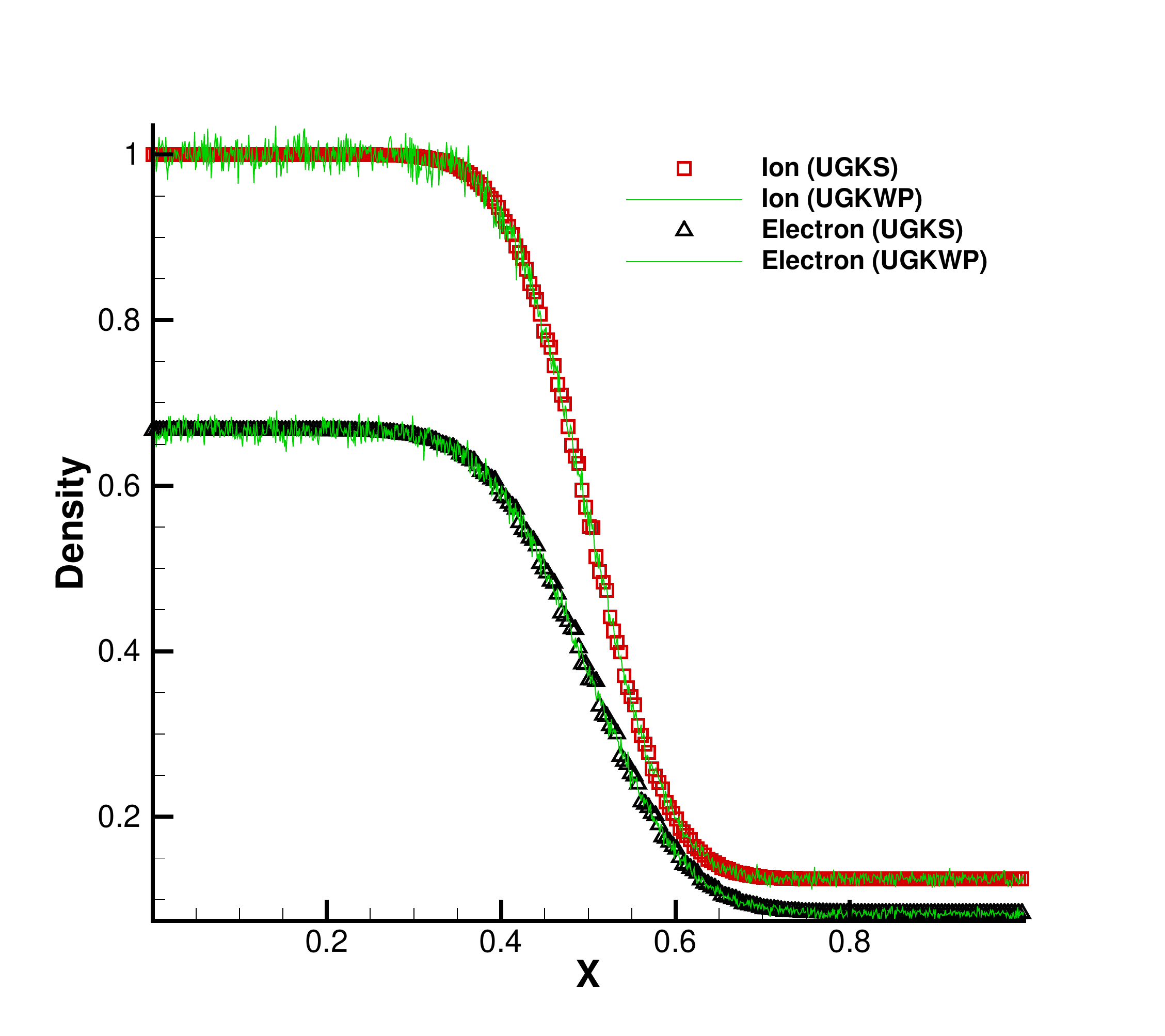}
         \caption{Density distribution}
     \end{subfigure}
     \hfill
     \begin{subfigure}[b]{0.48\textwidth}
         \centering
         \includegraphics[width=\textwidth]{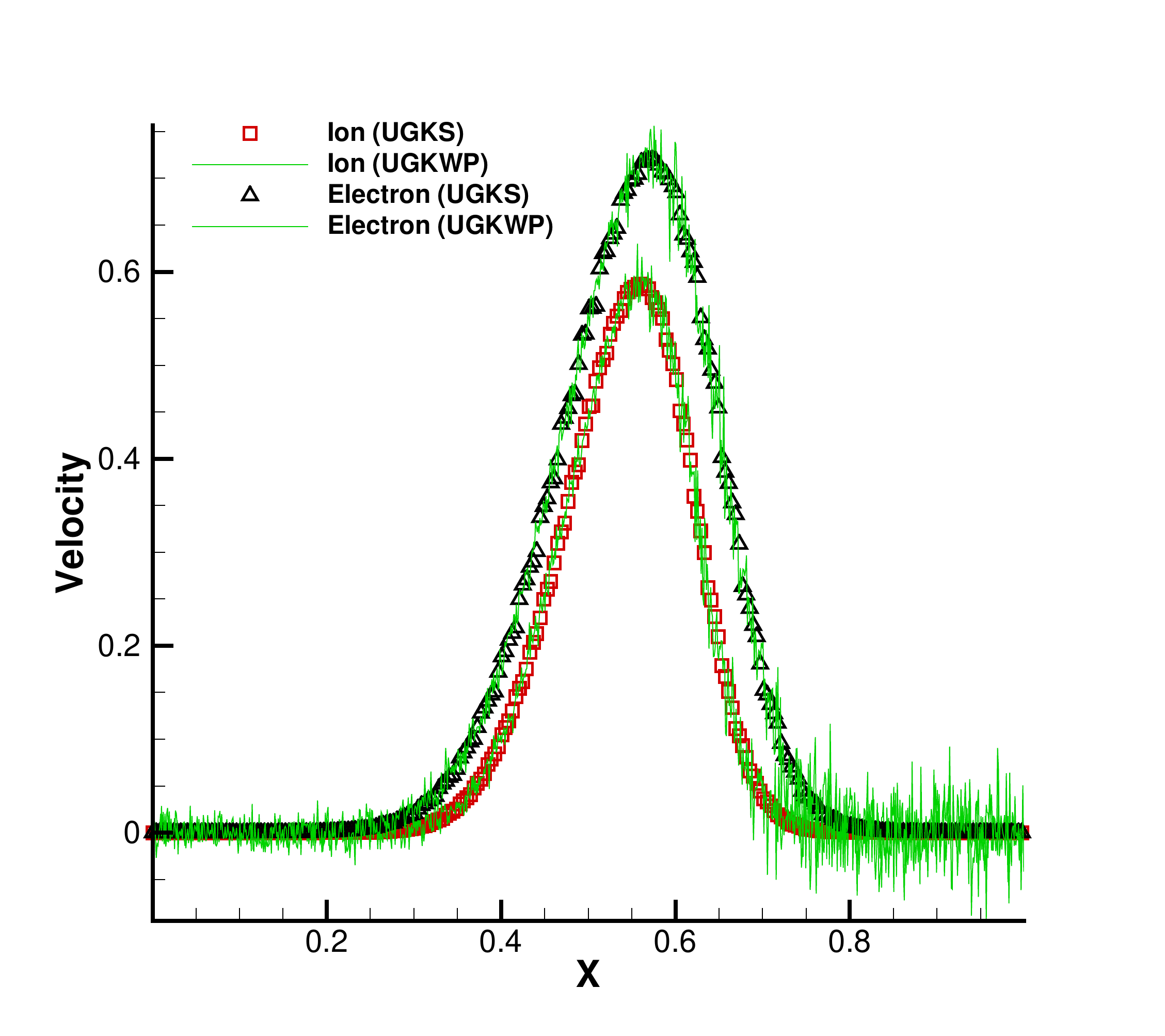}
         \caption{X-velocity distribution}
     \end{subfigure}
     \vfill
     \begin{subfigure}[b]{0.48\textwidth}
         \centering
         \includegraphics[width=\textwidth]{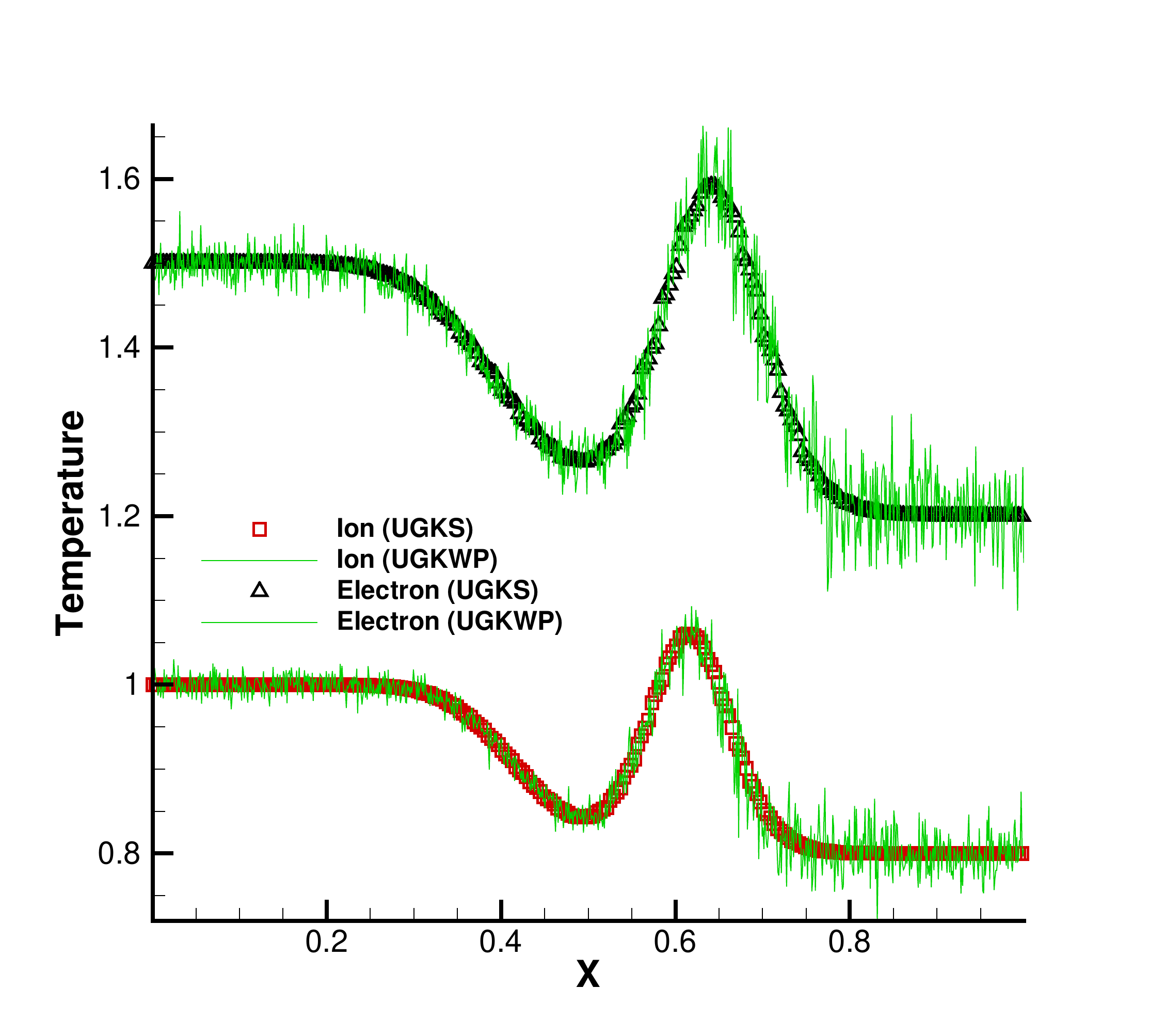}
         \caption{Pressure distribution}
     \end{subfigure}
     \hfill
     \begin{subfigure}[b]{0.48\textwidth}
         \centering
         \includegraphics[width=\textwidth]{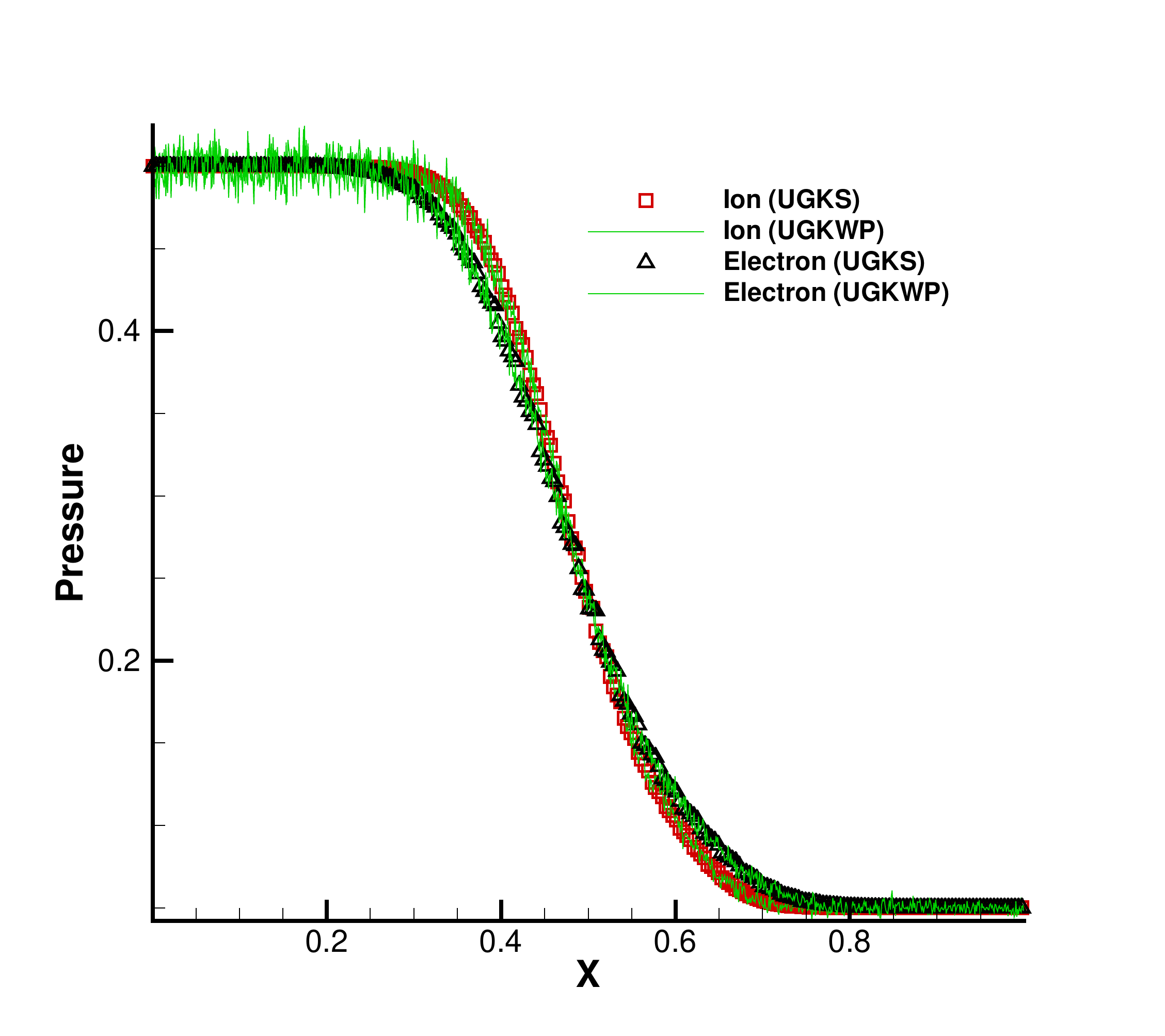}
         \caption{Y-magnetic field distribution}
     \end{subfigure}
\caption{Multiscale Brio-Wu shock tube problem with Kn=1 and $r=10^{-3}$.
Lines show the UGKWP solution and symbols show the reference UGKS solution.}
\label{brio1}
\end{figure}

\begin{figure}
     \centering
     \begin{subfigure}[b]{0.48\textwidth}
         \centering
         \includegraphics[width=\textwidth]{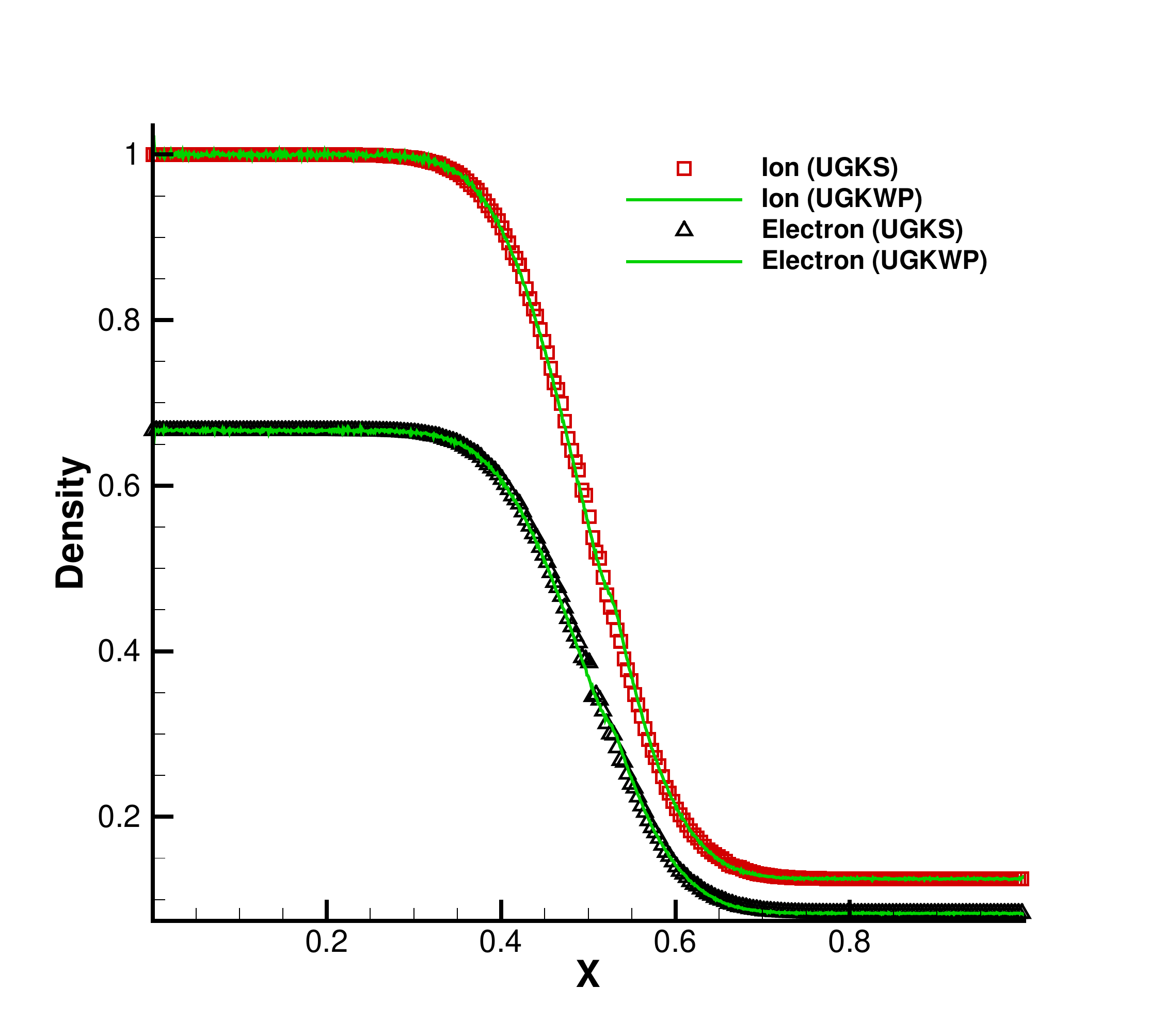}
         \caption{Density distribution}
     \end{subfigure}
     \hfill
     \begin{subfigure}[b]{0.48\textwidth}
         \centering
         \includegraphics[width=\textwidth]{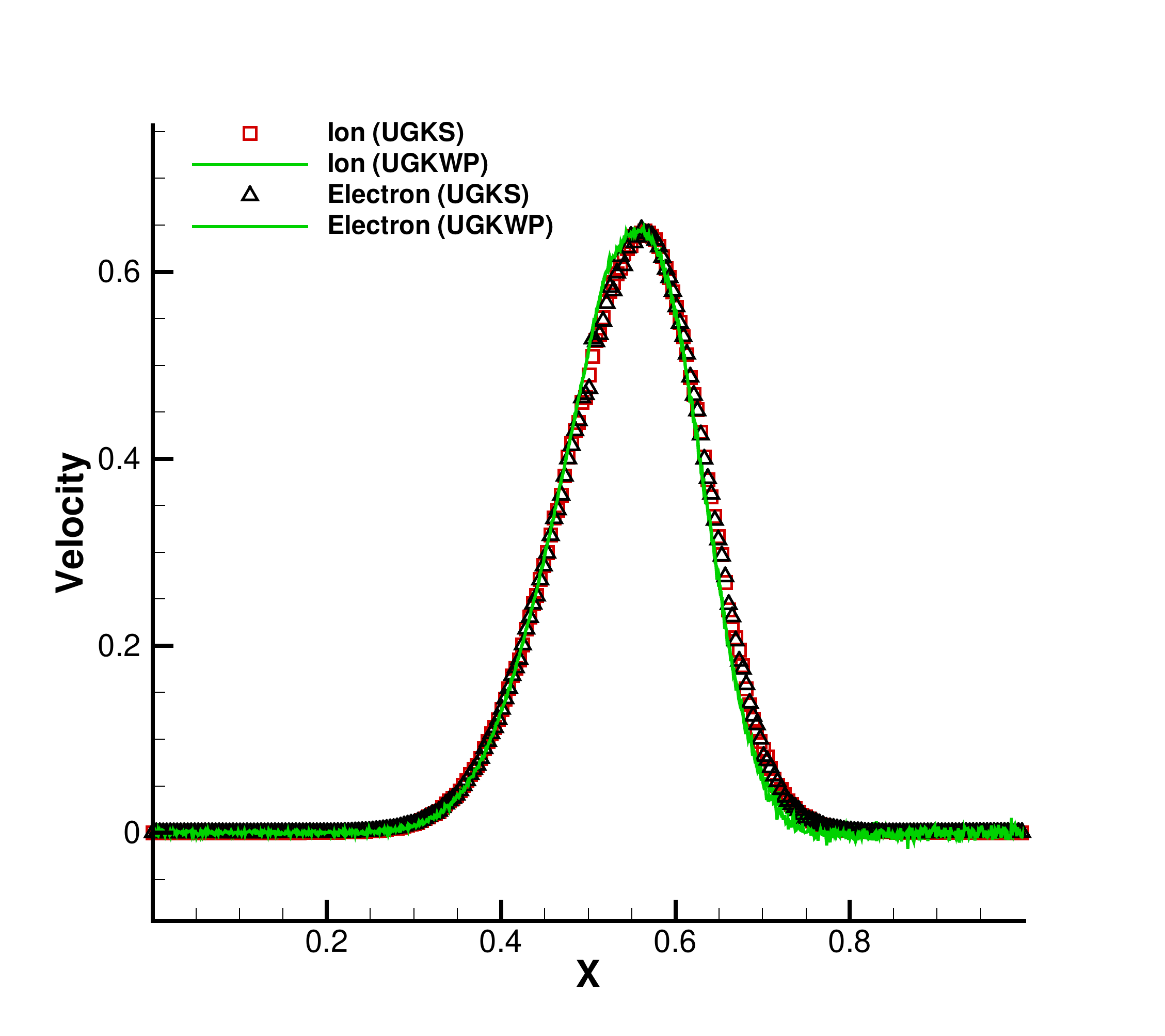}
         \caption{X-velocity distribution}
     \end{subfigure}
     \vfill
     \begin{subfigure}[b]{0.48\textwidth}
         \centering
         \includegraphics[width=\textwidth]{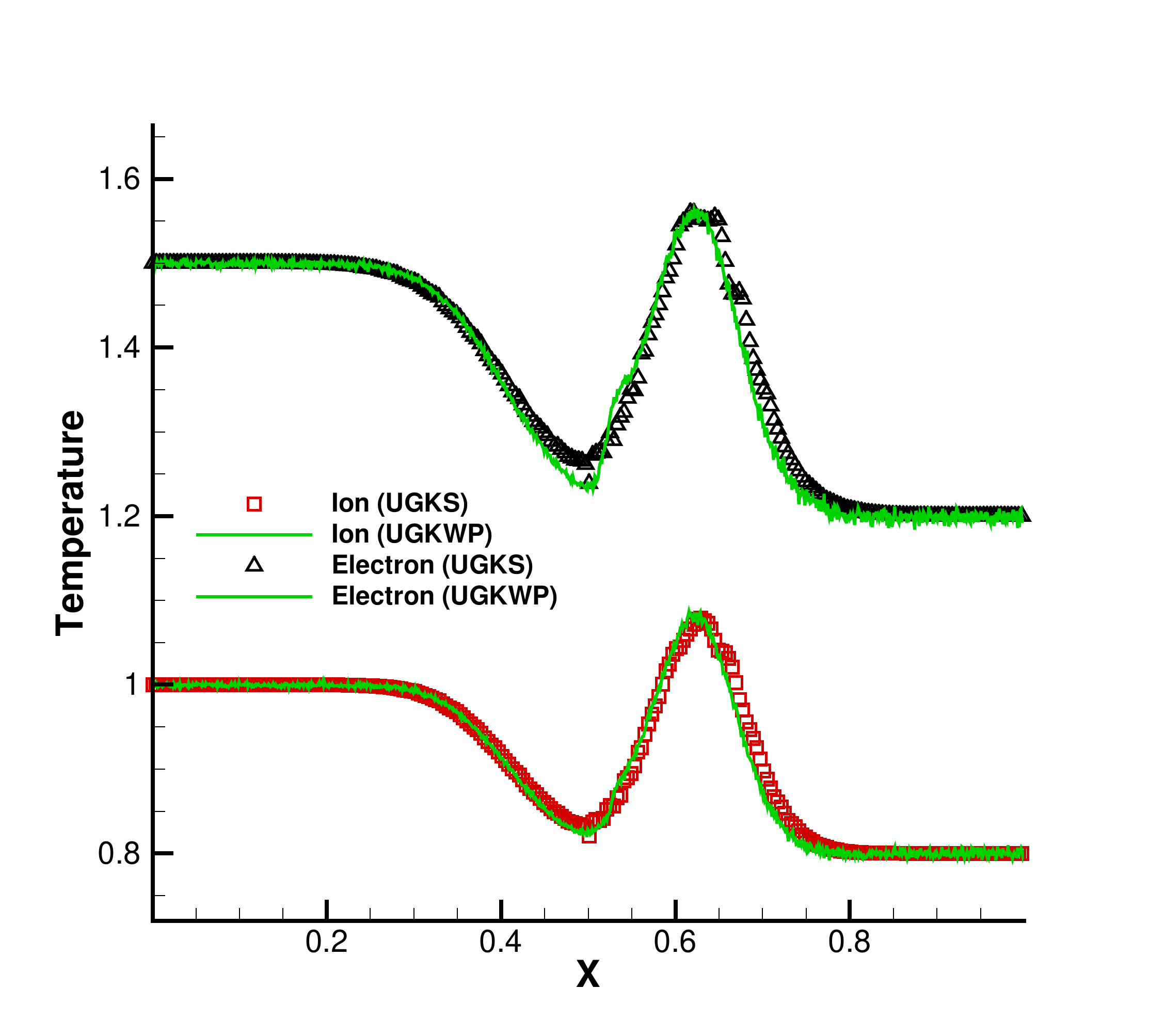}
         \caption{Pressure distribution}
     \end{subfigure}
     \hfill
     \begin{subfigure}[b]{0.48\textwidth}
         \centering
         \includegraphics[width=\textwidth]{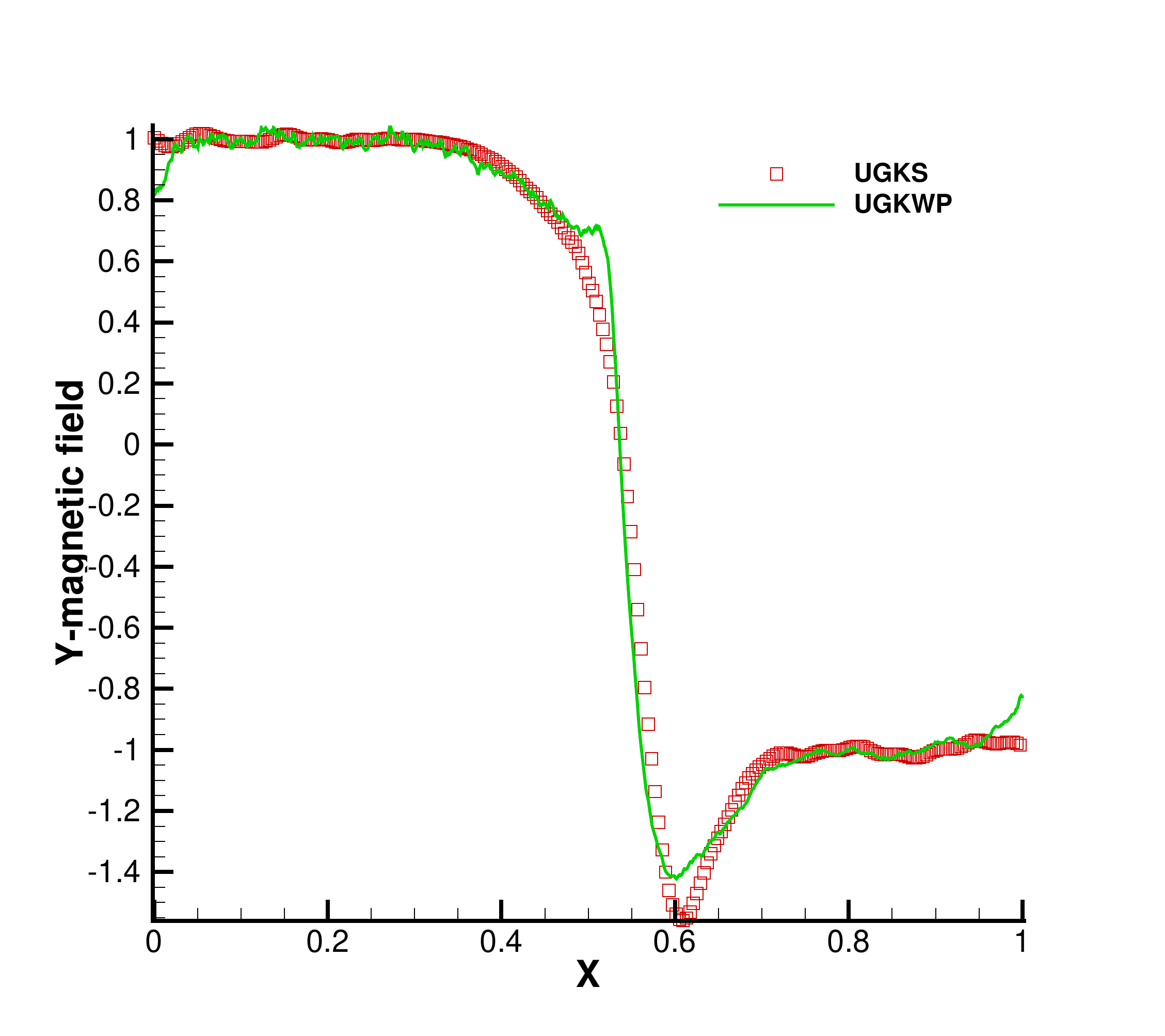}
         \caption{Y-magnetic field distribution}
     \end{subfigure}
\caption{Multiscale Brio-Wu shock tube problem with Kn=$10^{-2}$ and $r=10^{-3}$.
Lines show the UGKWP solution and symbols show the reference UGKS solution.}
\label{brio2}
\end{figure}

\begin{figure}
     \centering
     \begin{subfigure}[b]{0.48\textwidth}
         \centering
         \includegraphics[width=\textwidth]{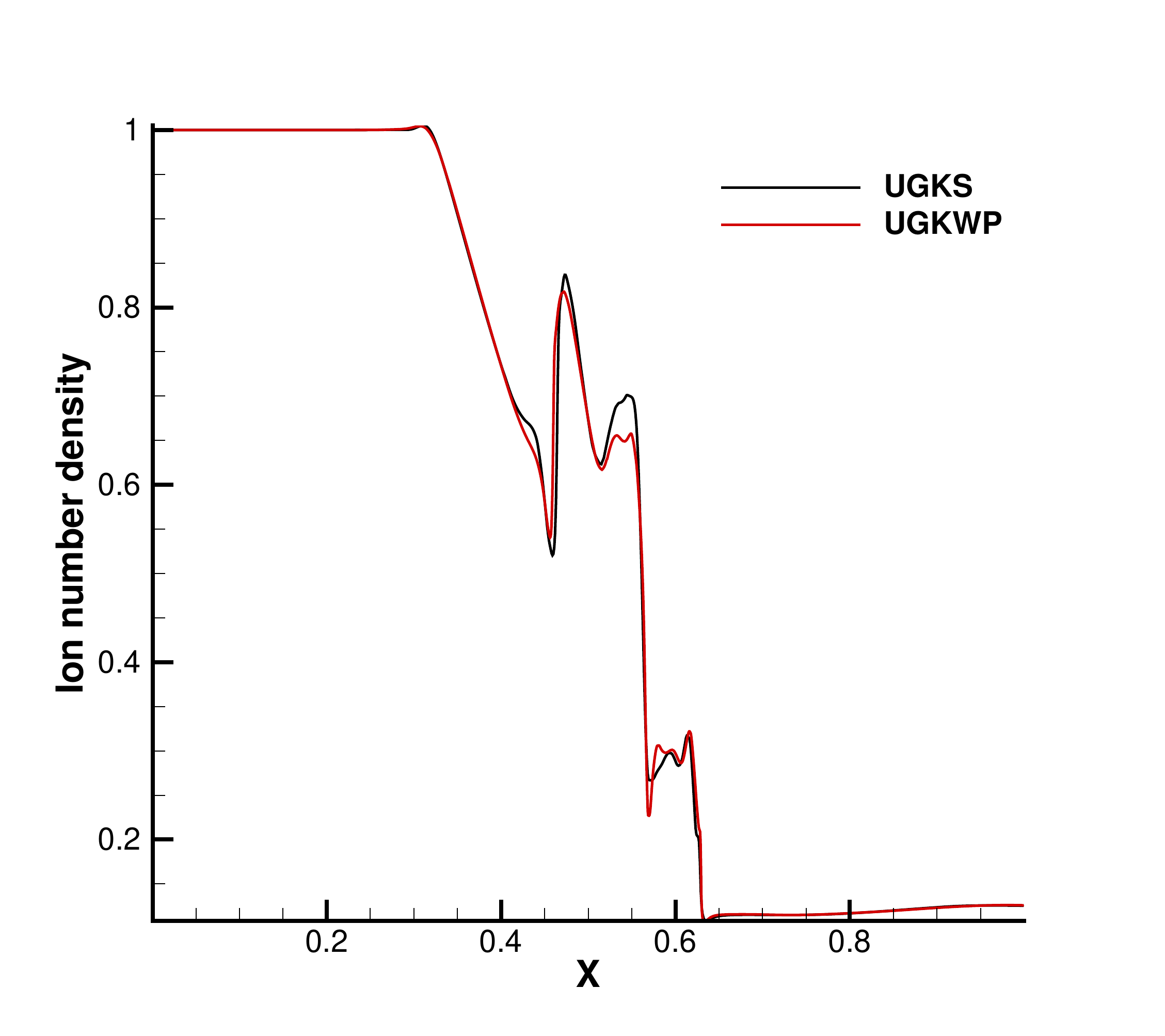}
         \caption{Density distribution}
     \end{subfigure}
     \hfill
     \begin{subfigure}[b]{0.48\textwidth}
         \centering
         \includegraphics[width=\textwidth]{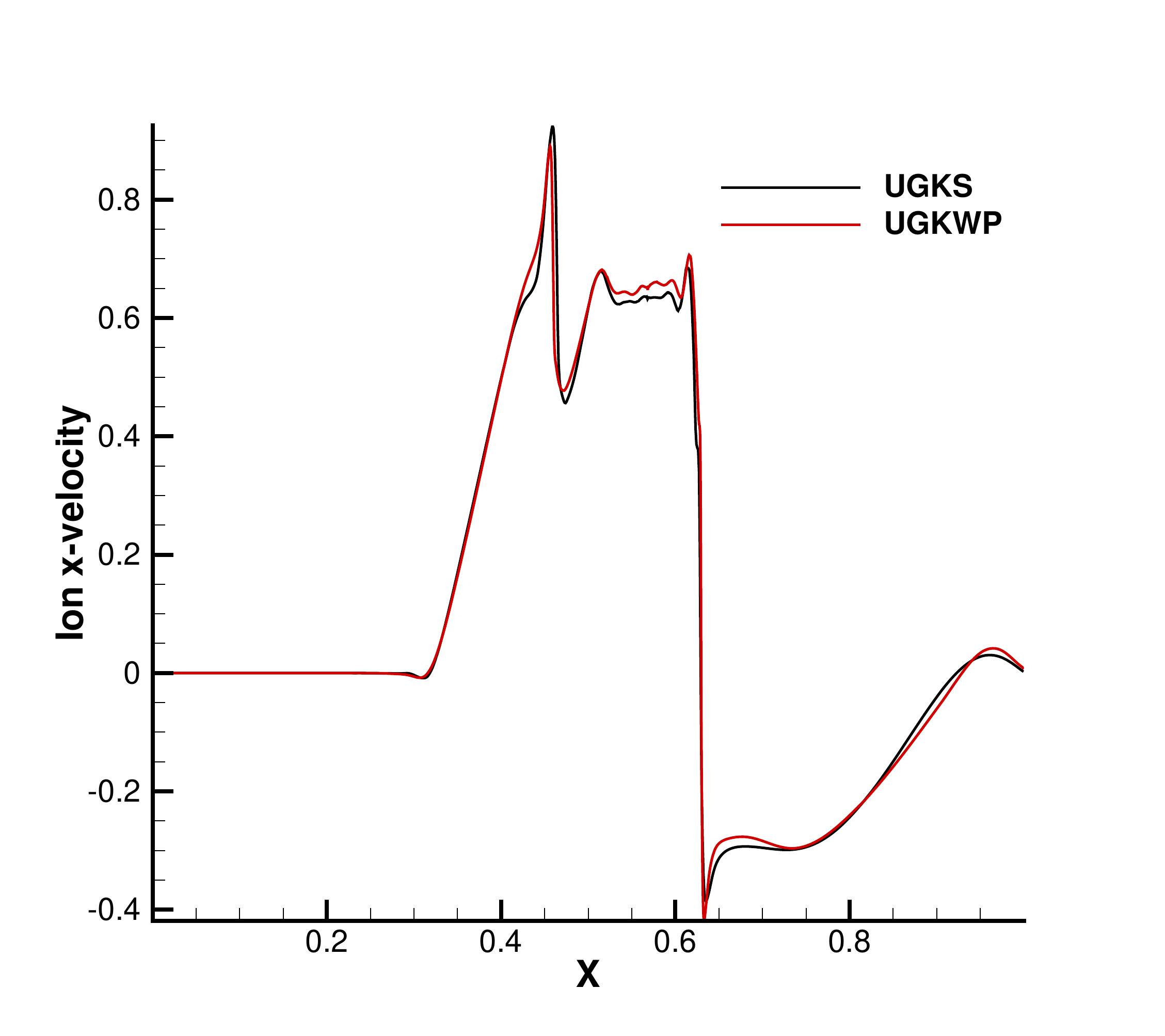}
         \caption{X-velocity distribution}
     \end{subfigure}
     \vfill
     \begin{subfigure}[b]{0.48\textwidth}
         \centering
         \includegraphics[width=\textwidth]{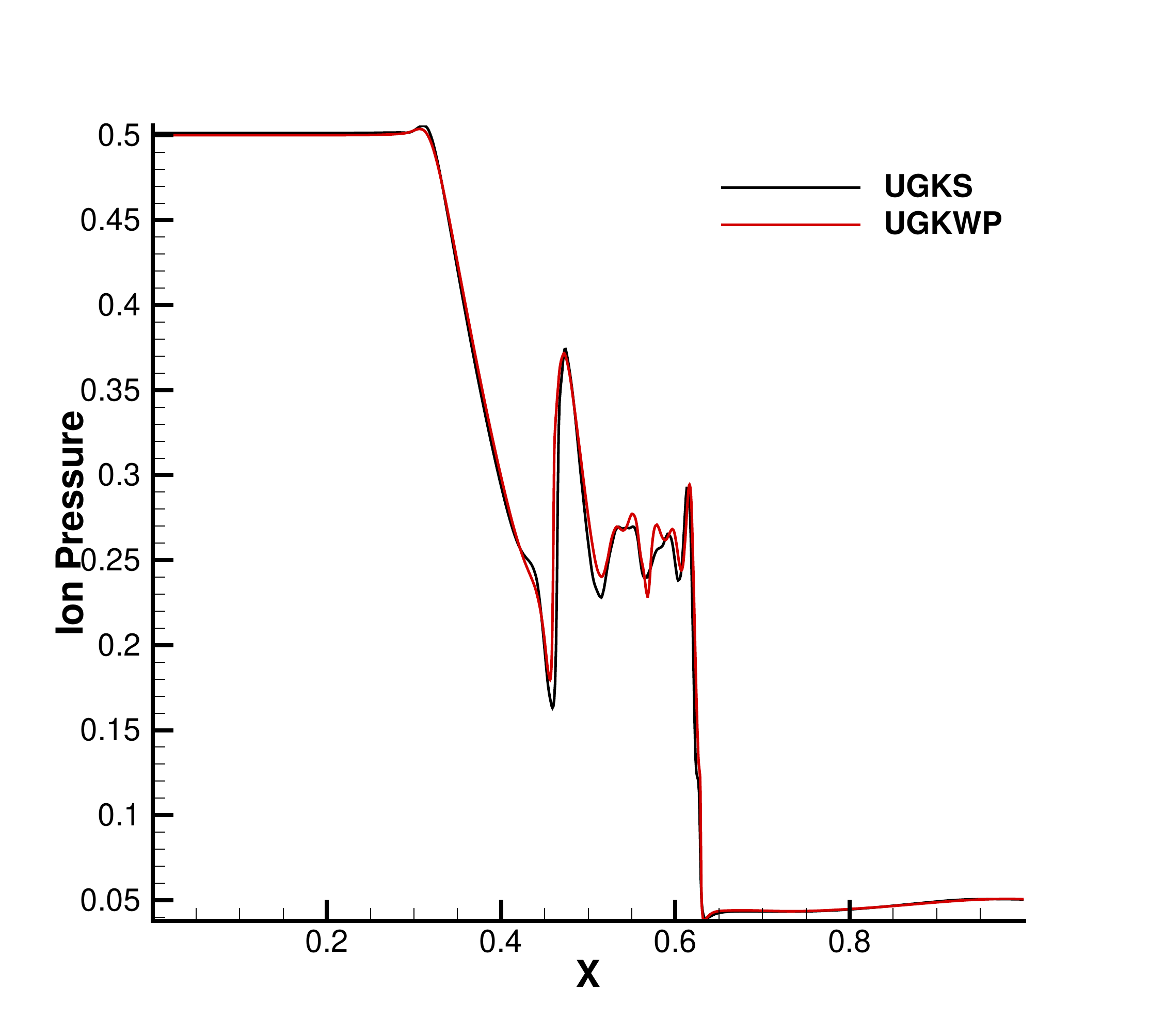}
         \caption{Y-velocity distribution}
     \end{subfigure}
     \hfill
     \begin{subfigure}[b]{0.48\textwidth}
         \centering
         \includegraphics[width=\textwidth]{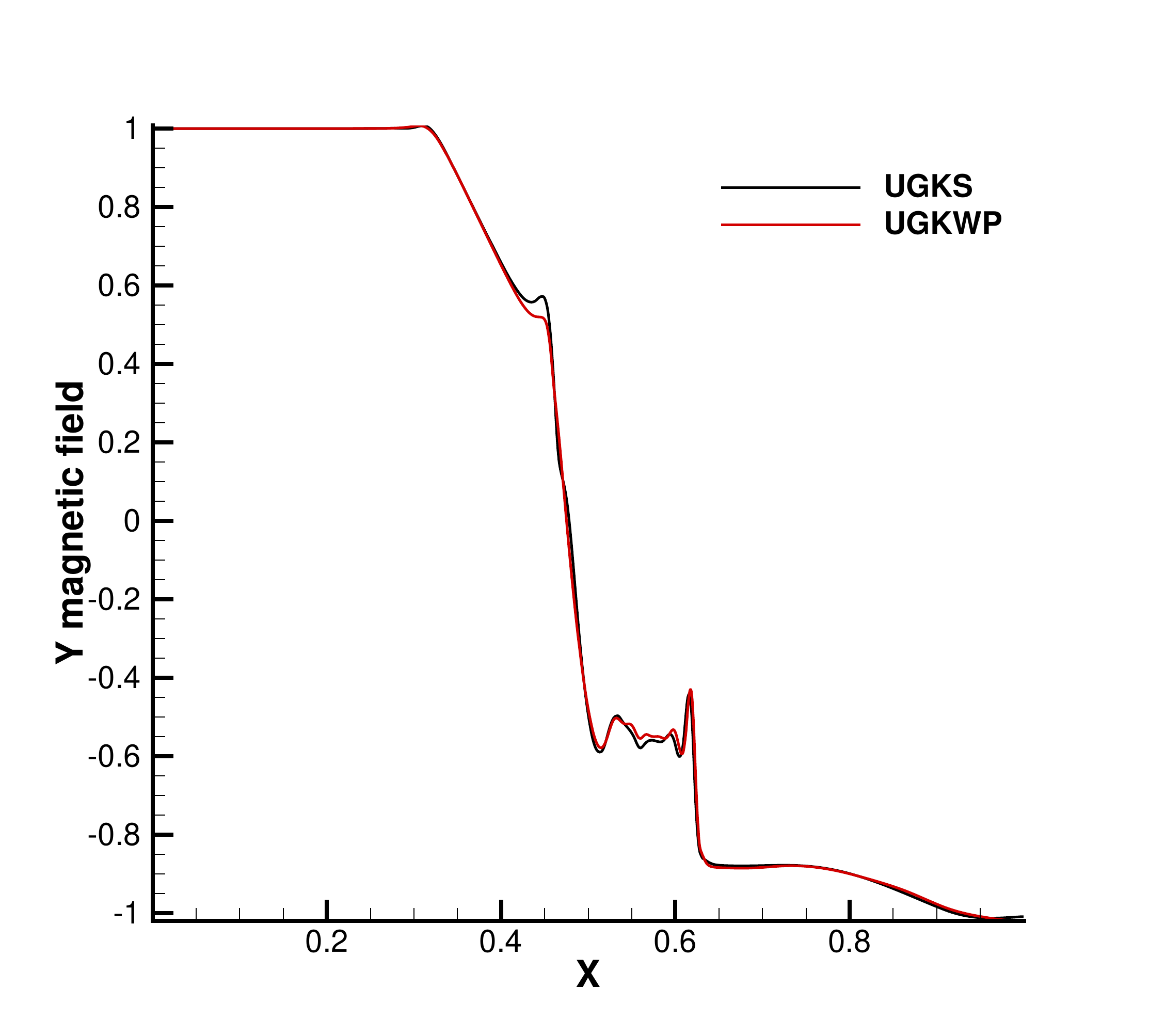}
         \caption{Y-magnetic field distribution}
     \end{subfigure}
\caption{Multiscale Brio-Wu shock tube problem with Kn=$10^{-4}$ and $r=10^{-3}$.
Lines show the UGKWP solution and symbols show the reference UGKS solution.}
\label{brio3}
\end{figure}

\begin{figure}
     \centering
     \begin{subfigure}[b]{0.48\textwidth}
         \centering
         \includegraphics[width=\textwidth]{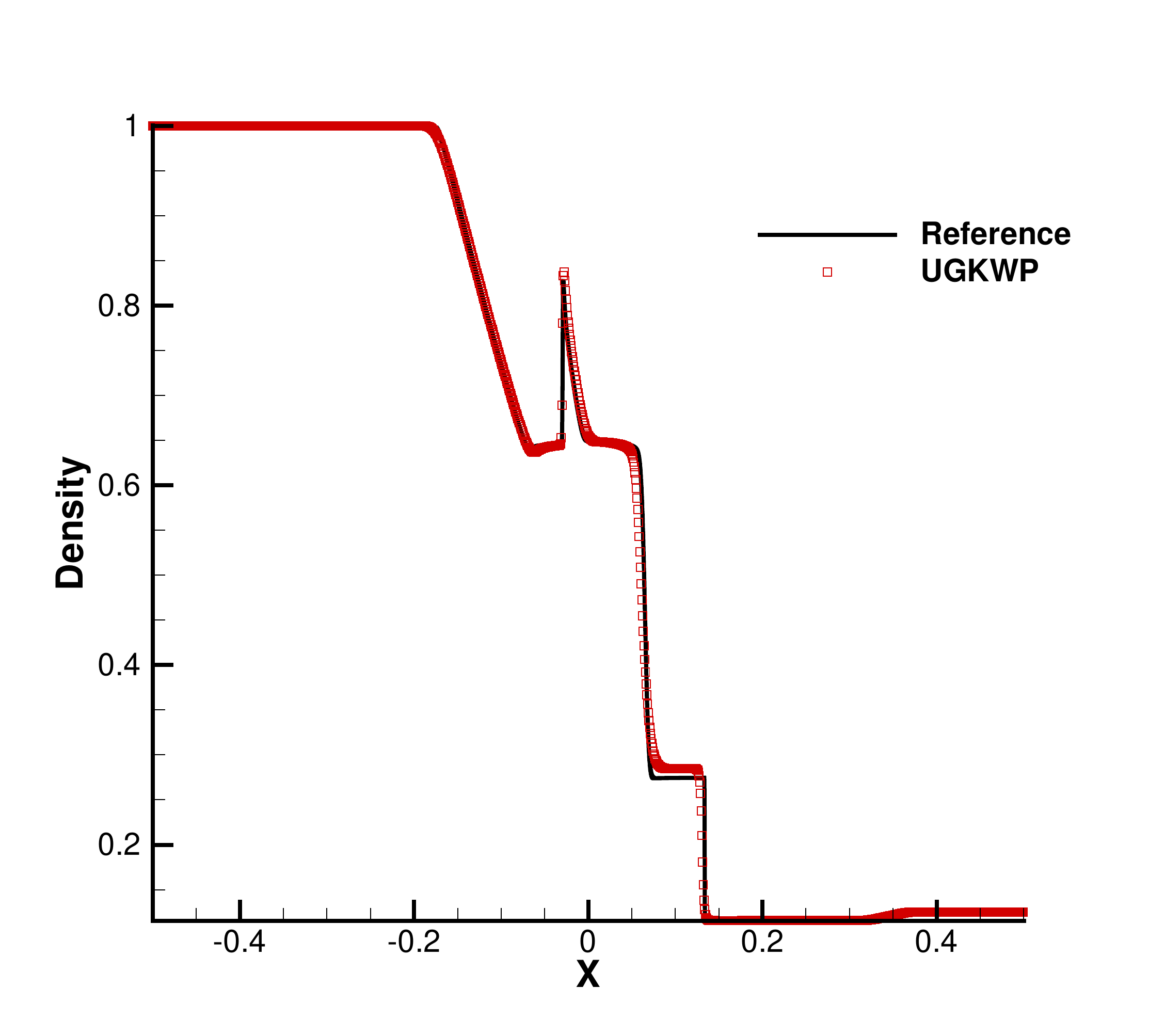}
         \caption{Density distribution}
     \end{subfigure}
     \hfill
     \begin{subfigure}[b]{0.48\textwidth}
         \centering
         \includegraphics[width=\textwidth]{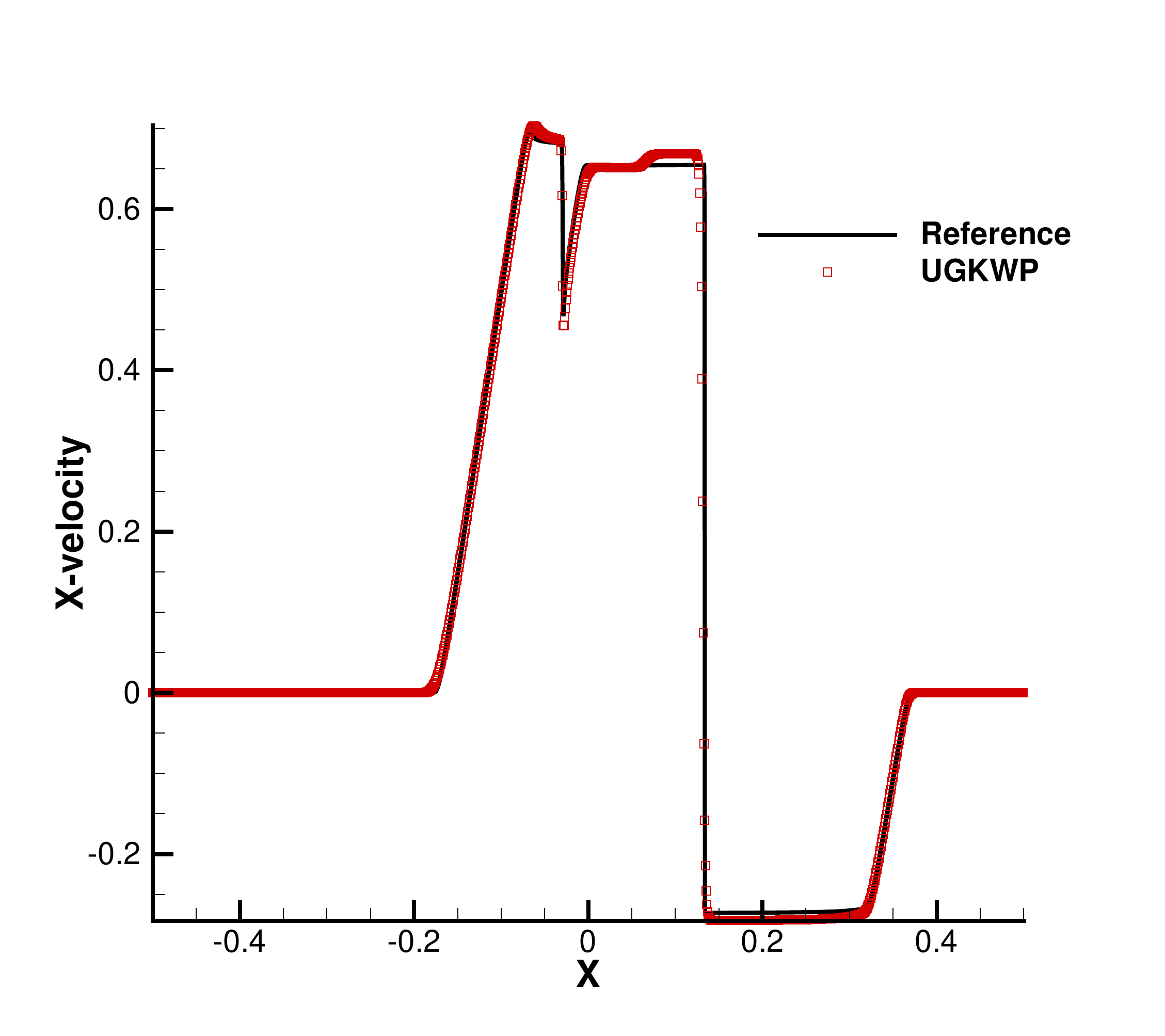}
         \caption{X-velocity distribution}
     \end{subfigure}
     \vfill
     \begin{subfigure}[b]{0.48\textwidth}
         \centering
         \includegraphics[width=\textwidth]{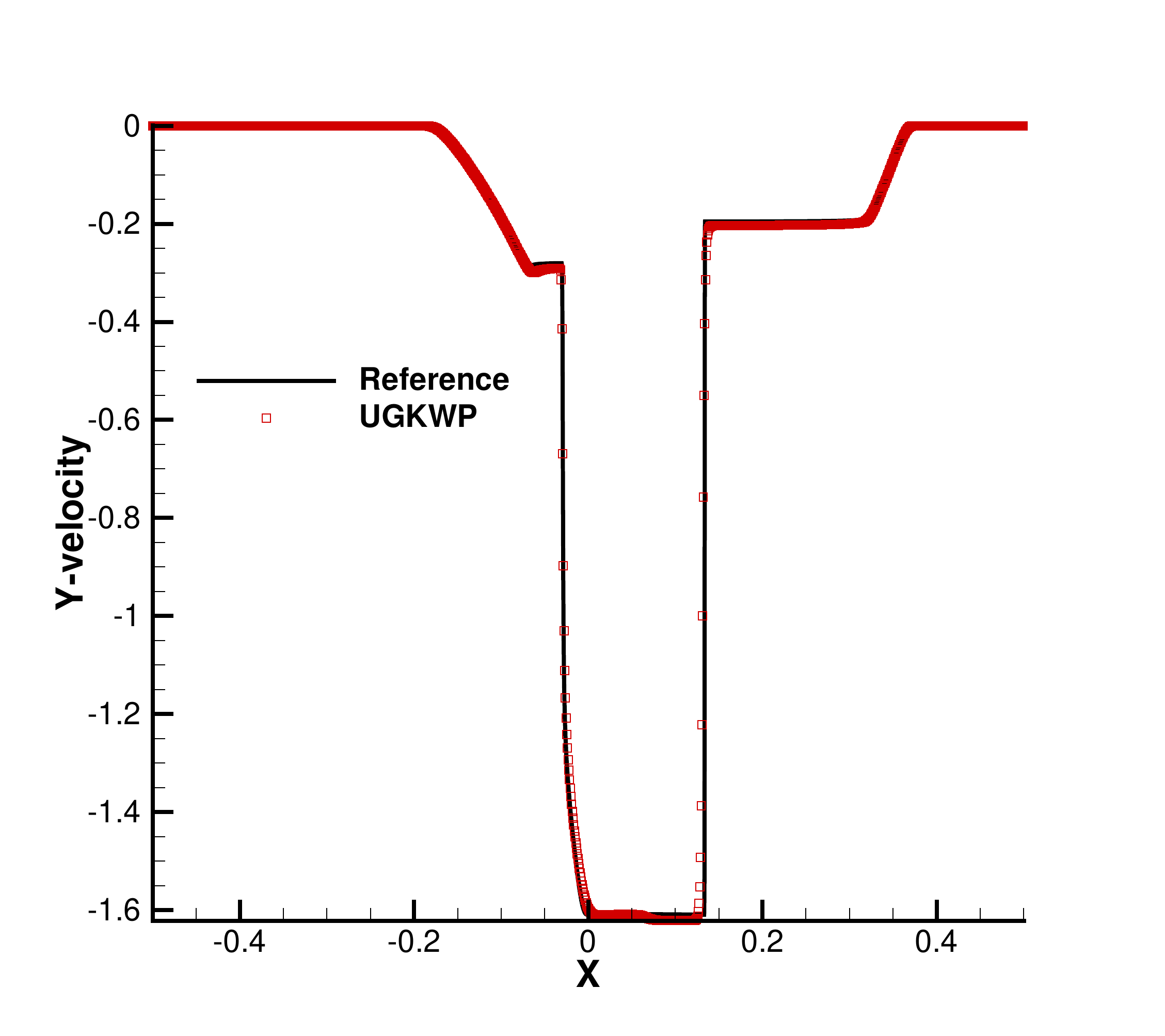}
         \caption{Y-velocity distribution}
     \end{subfigure}
     \hfill
     \begin{subfigure}[b]{0.48\textwidth}
         \centering
         \includegraphics[width=\textwidth]{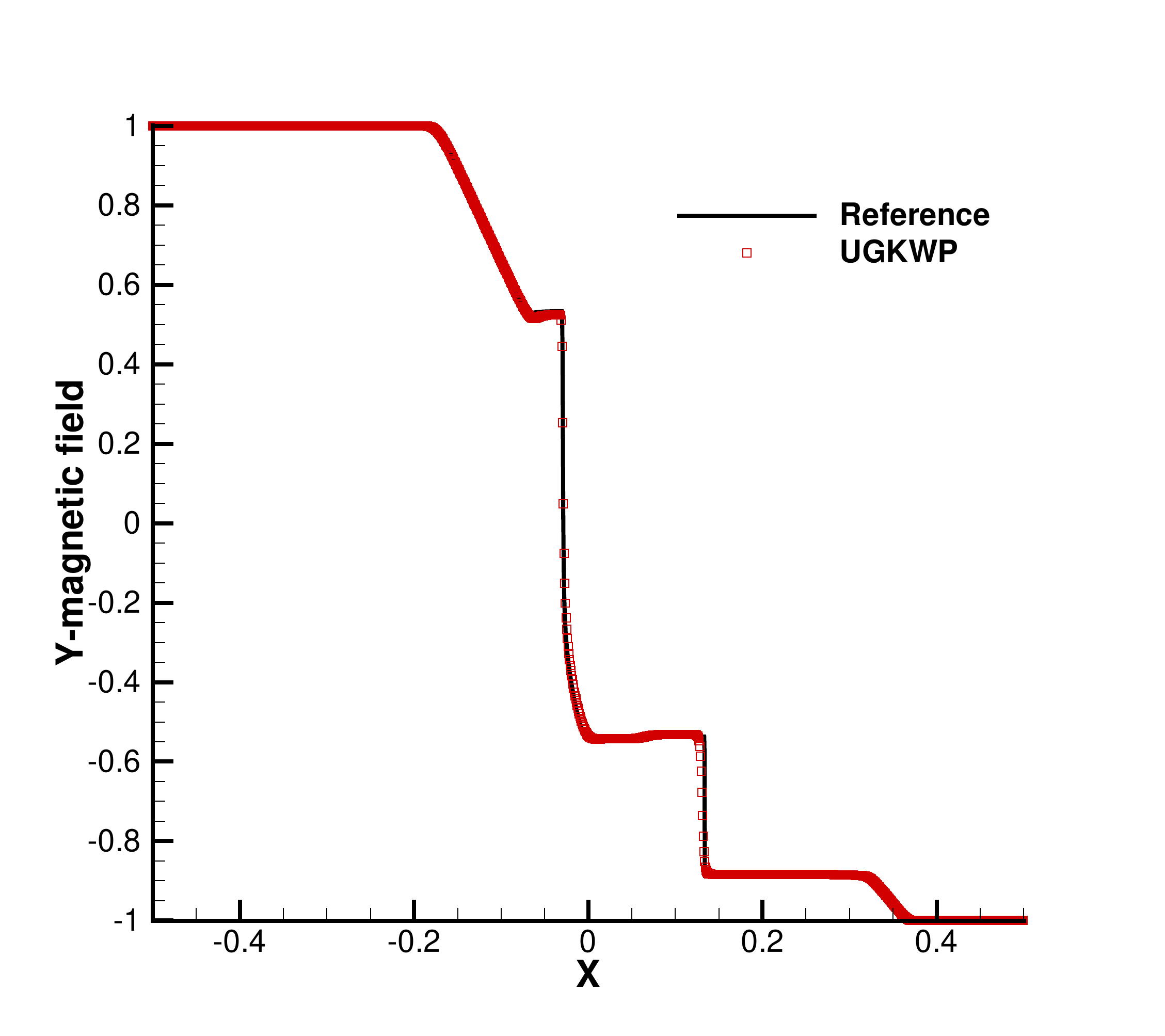}
         \caption{Y-magnetic field distribution}
     \end{subfigure}
\caption{Multiscale Brio-Wu shock tube problem with Kn=$10^{-4}$ and $r=0$.
Lines show the UGKWP solution and symbols show the reference MHD solution.}
\label{brio4}
\end{figure}

\begin{figure}
     \centering
     \begin{subfigure}[b]{0.48\textwidth}
         \centering
         \includegraphics[width=\textwidth]{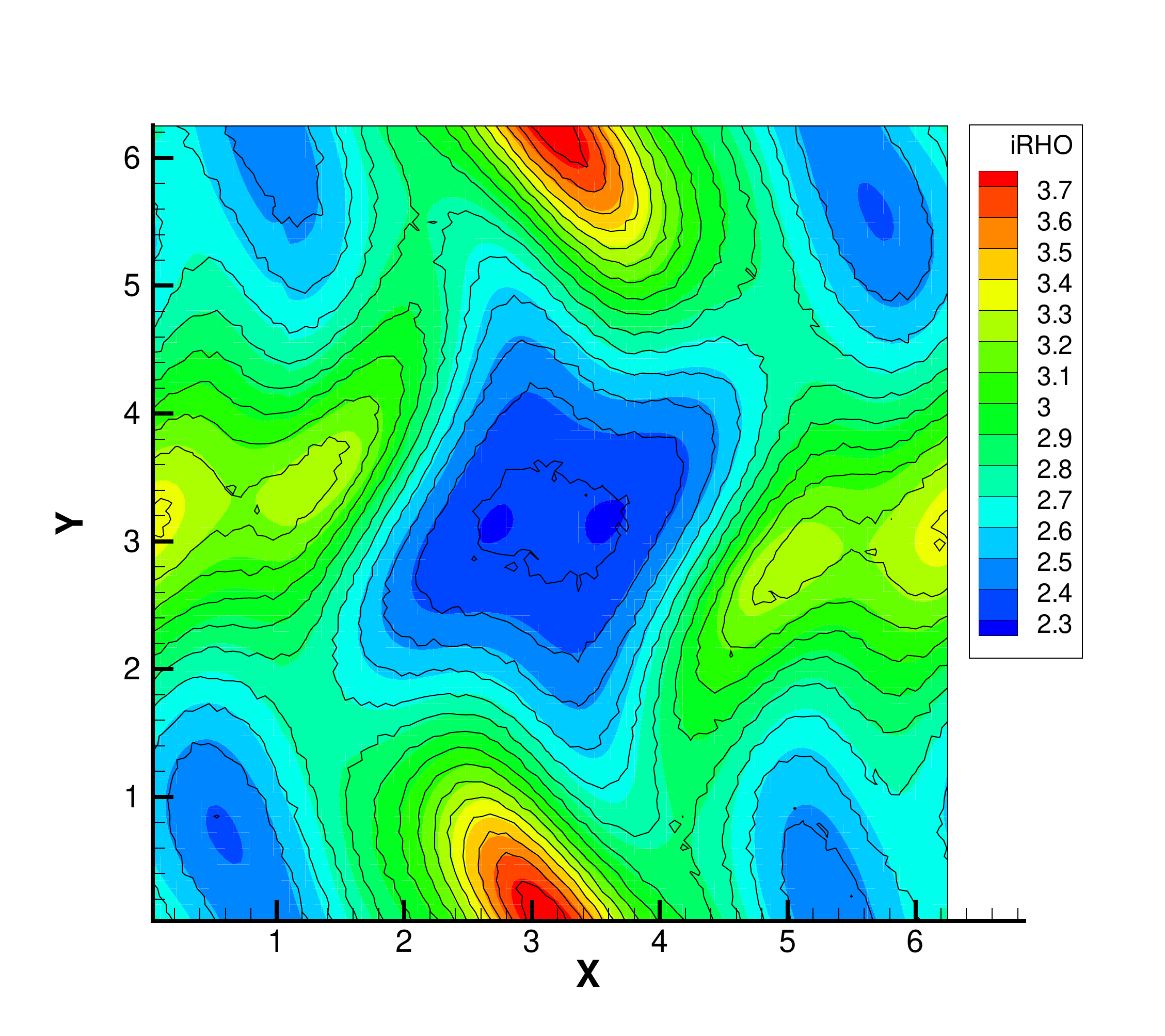}
         \caption{Ion density contour}
     \end{subfigure}
     \hfill
     \begin{subfigure}[b]{0.48\textwidth}
         \centering
         \includegraphics[width=\textwidth]{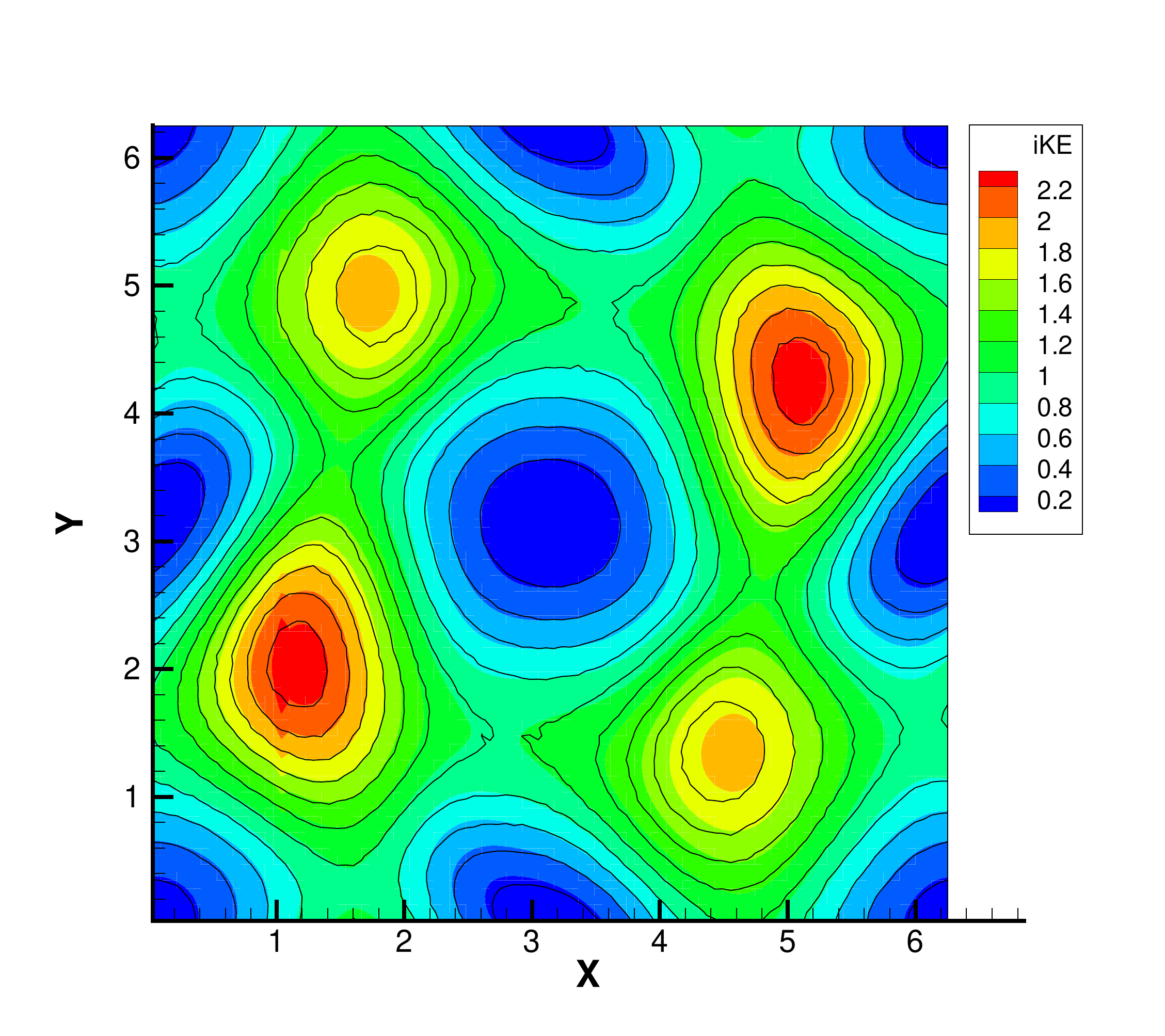}
         \caption{Ion kinetic energy contour}
     \end{subfigure}
     \vfill
     \begin{subfigure}[b]{0.48\textwidth}
         \centering
         \includegraphics[width=\textwidth]{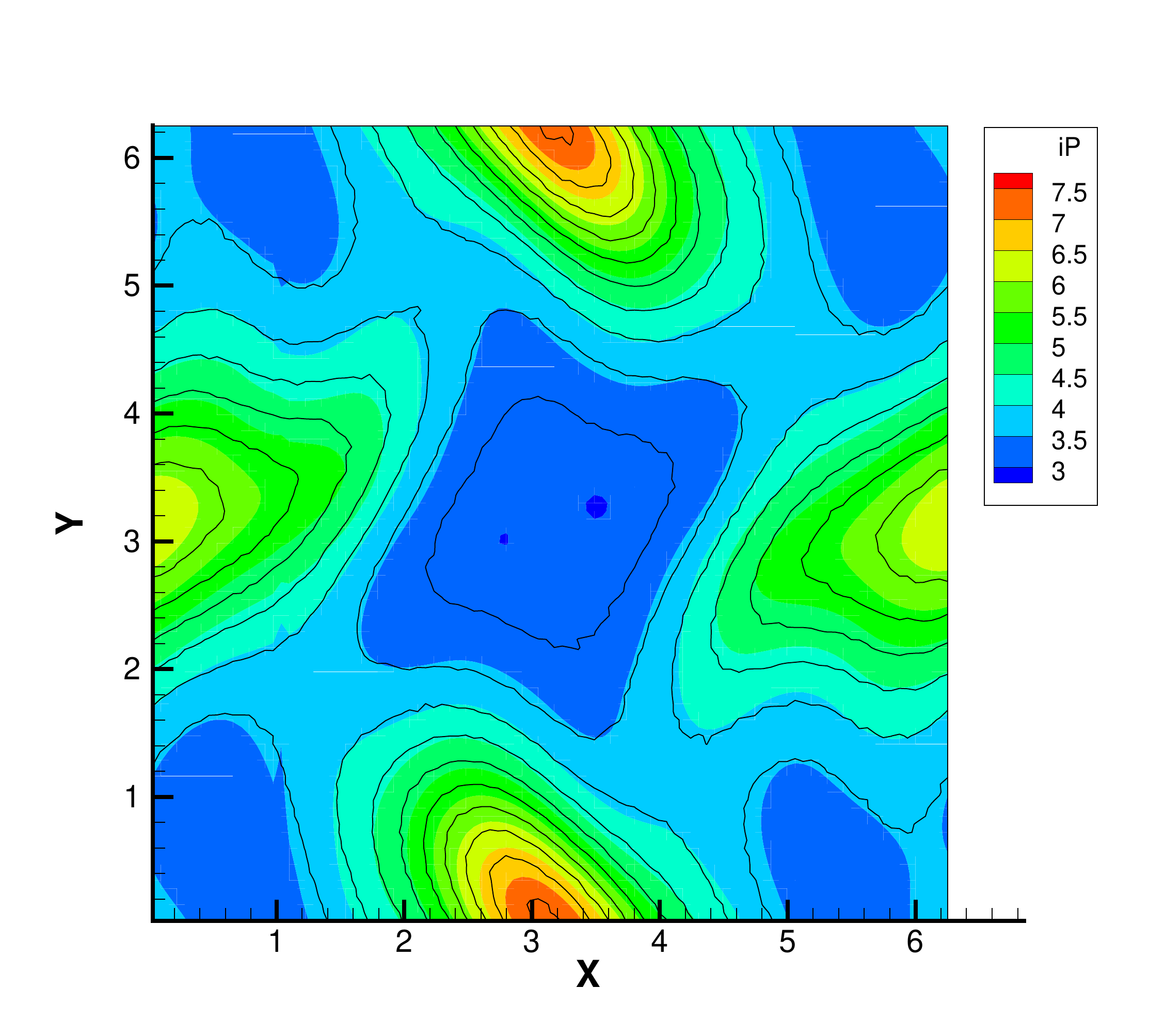}
         \caption{Ion pressure contour}
     \end{subfigure}
     \hfill
     \begin{subfigure}[b]{0.48\textwidth}
         \centering
         \includegraphics[width=\textwidth]{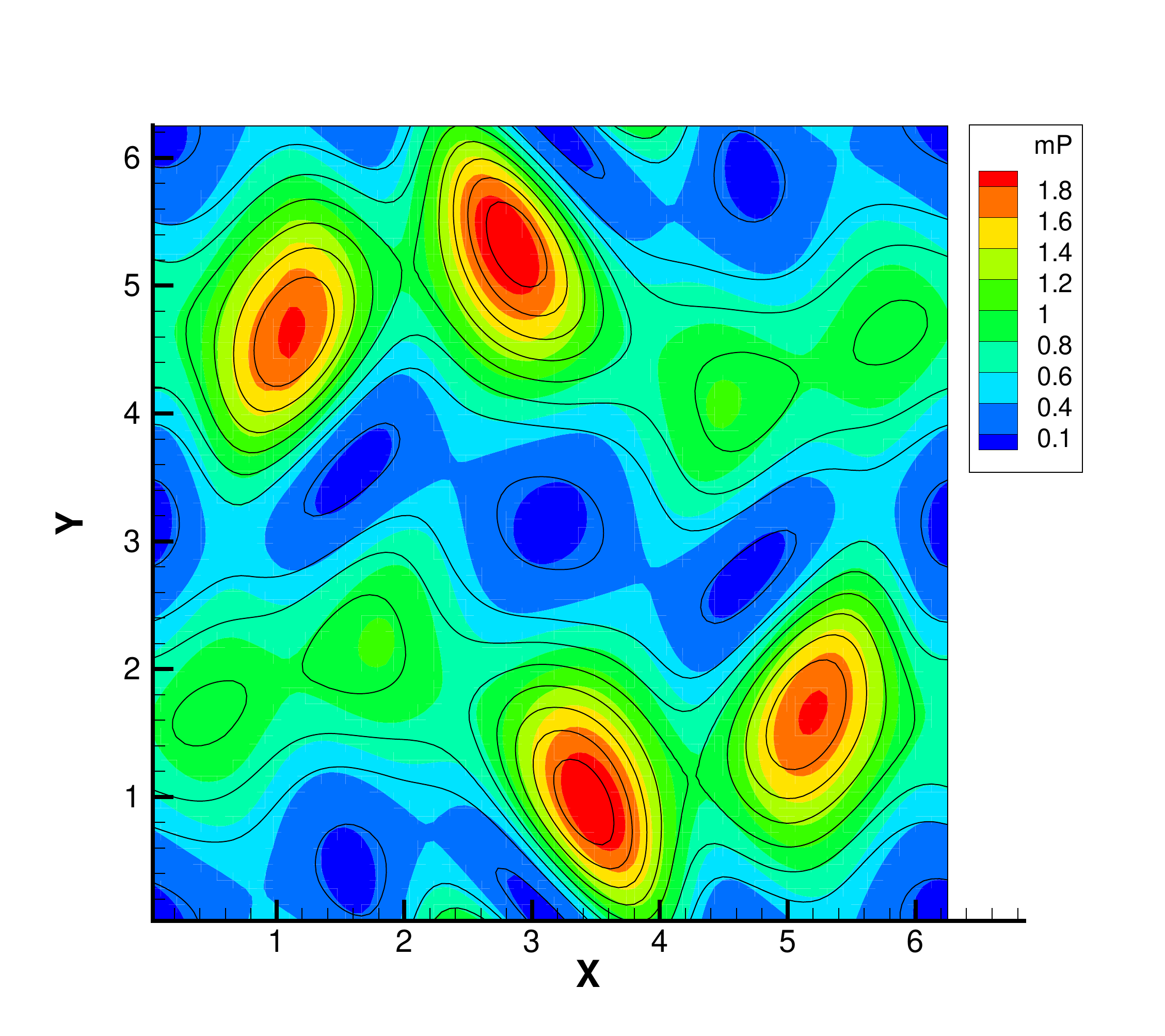}
         \caption{Magnetic pressure contour}
     \end{subfigure}
\caption{The results of the multiscale Orszag-Tang vortex problem with $\text{Kn}=1$ and $r=10^{-3}$ at $t=1$.
Contour line shows the UGKWP solution and contour flood shows the UGKS solution.}
\label{orzag1}
\end{figure}

\begin{figure}
     \centering
     \begin{subfigure}[b]{0.48\textwidth}
         \centering
         \includegraphics[width=\textwidth]{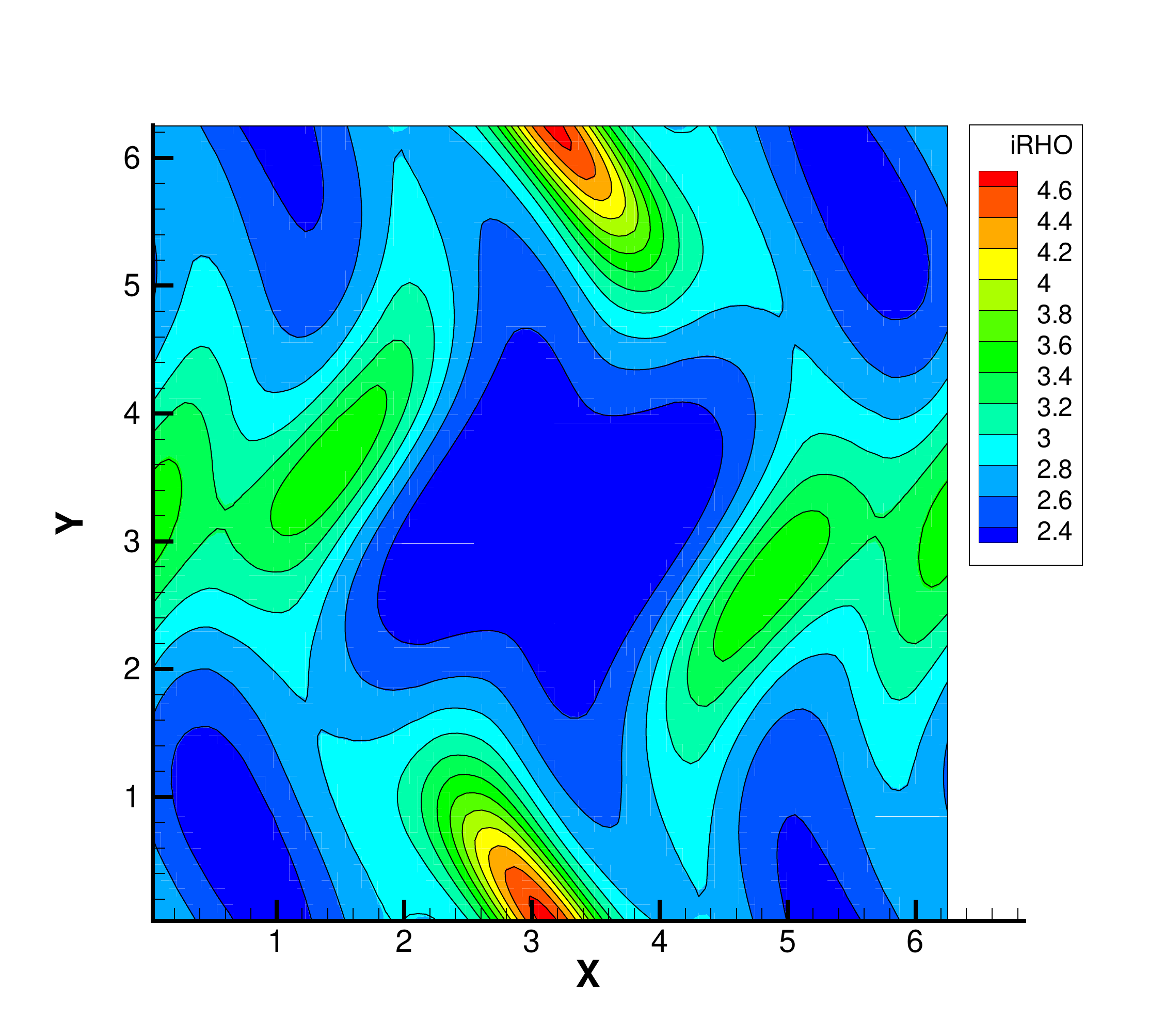}
         \caption{Ion density contour}
     \end{subfigure}
     \hfill
     \begin{subfigure}[b]{0.48\textwidth}
         \centering
         \includegraphics[width=\textwidth]{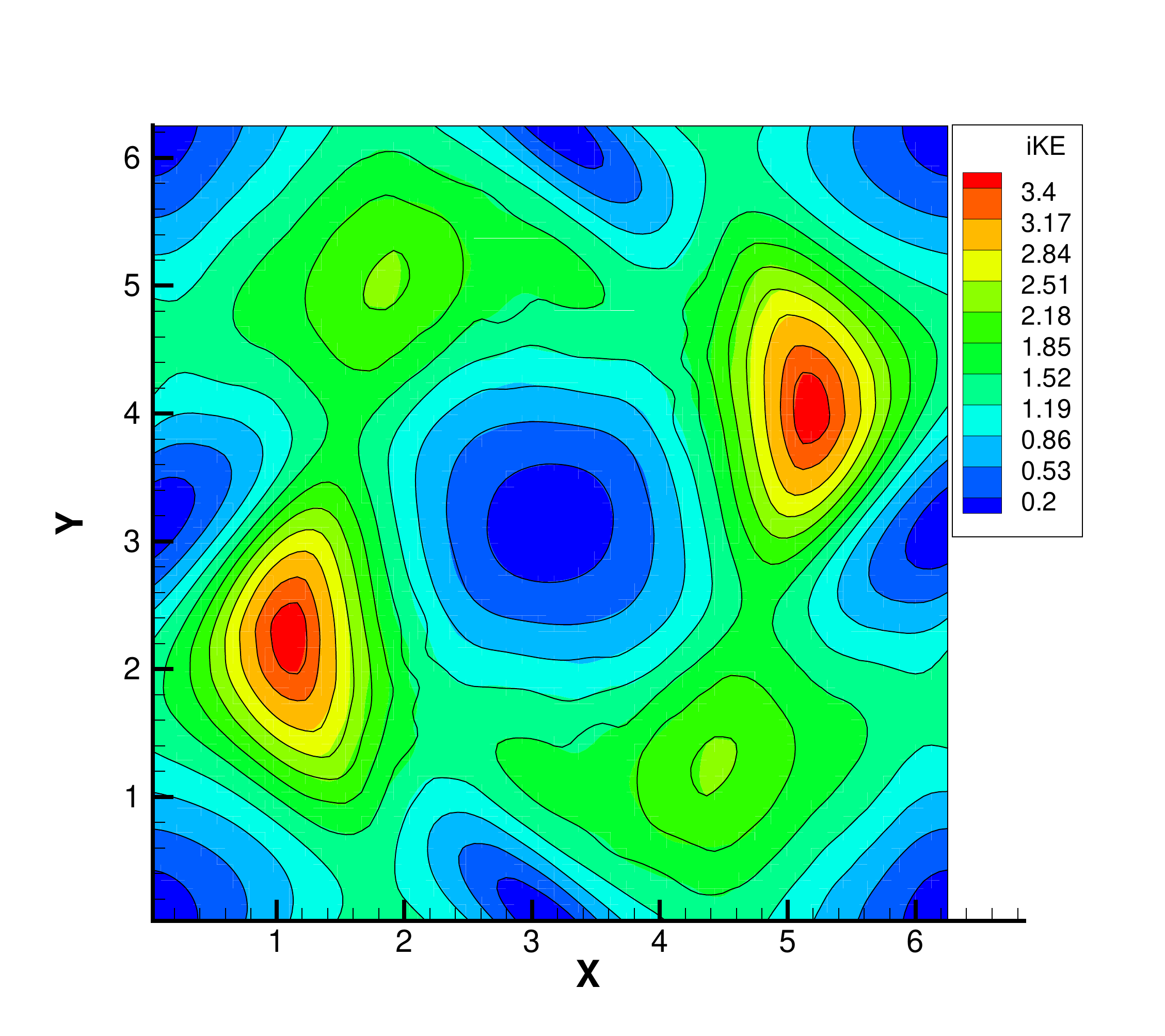}
         \caption{Ion kinetic energy contour}
     \end{subfigure}
     \vfill
     \begin{subfigure}[b]{0.48\textwidth}
         \centering
         \includegraphics[width=\textwidth]{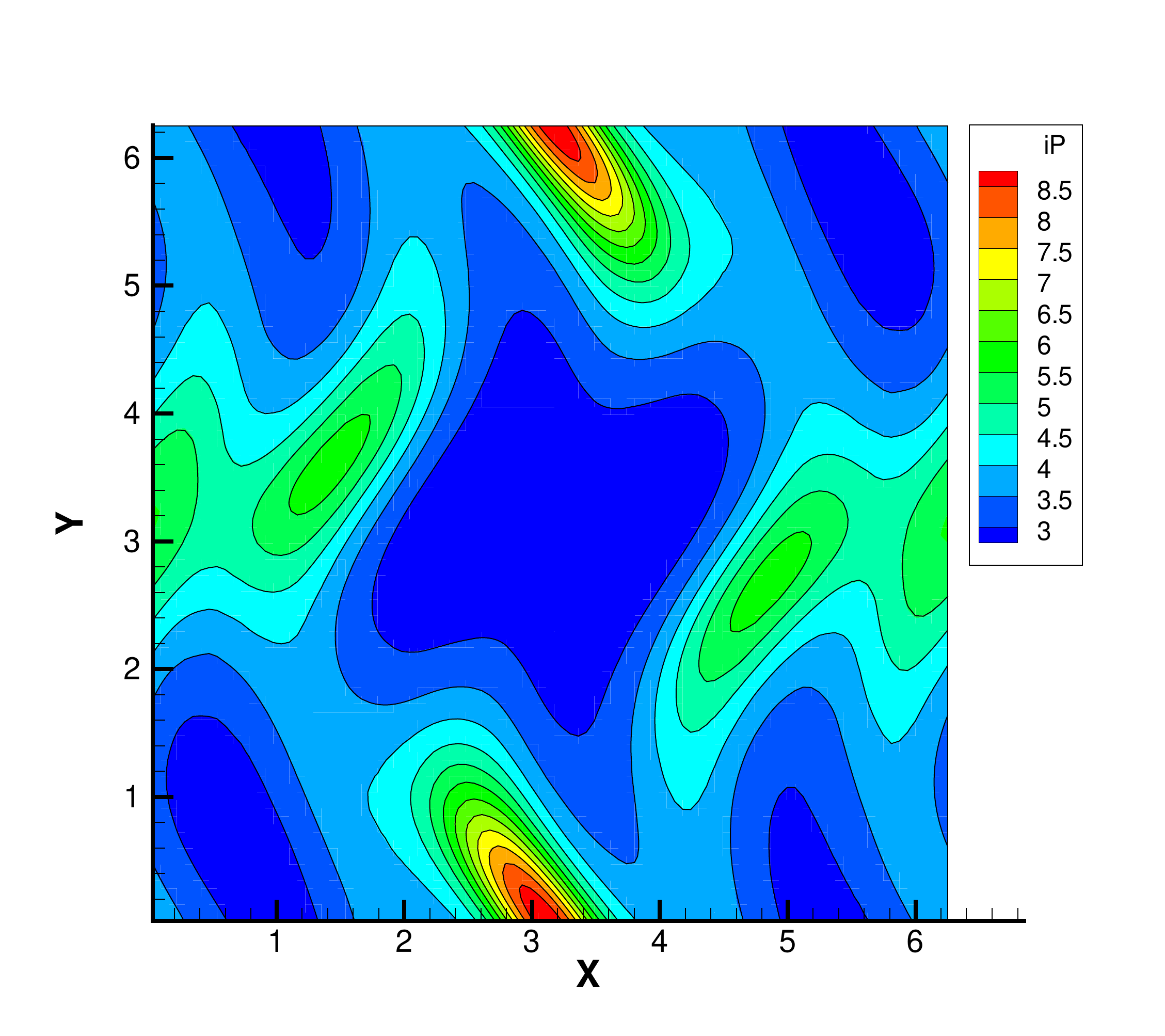}
         \caption{Ion pressure contour}
     \end{subfigure}
     \hfill
     \begin{subfigure}[b]{0.48\textwidth}
         \centering
         \includegraphics[width=\textwidth]{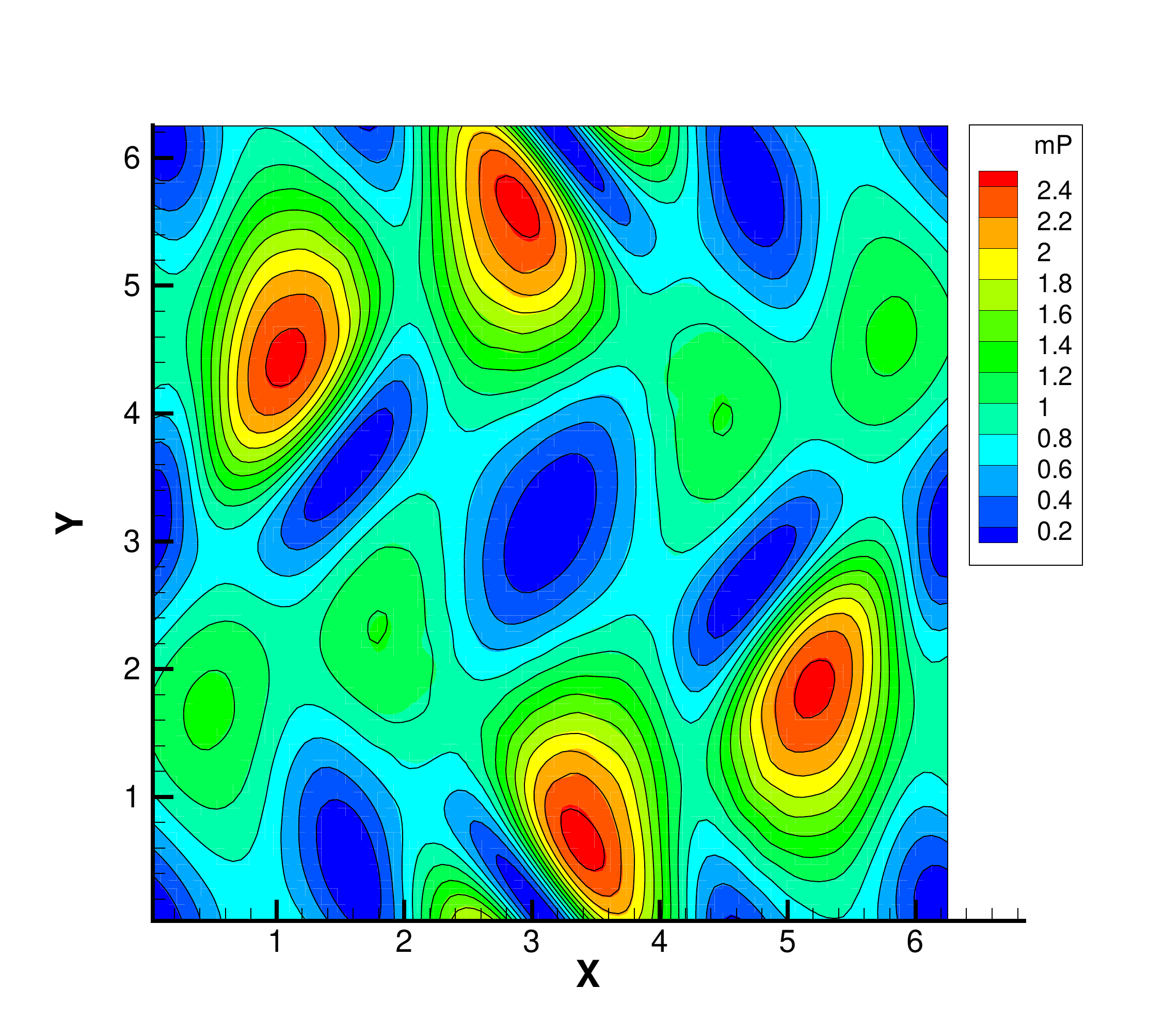}
         \caption{Magnetic pressure contour}
     \end{subfigure}
\caption{The results of the multiscale Orszag-Tang vortex problem with $\text{Kn}=10^{-4}$ and $r=10^{-3}$ at $t=1$.
Contour line shows the UGKWP solution and contour flood shows the UGKS solution.}
\label{orzag2}
\end{figure}

\begin{figure}
     \centering
     \begin{subfigure}[b]{0.48\textwidth}
         \centering
         \includegraphics[width=\textwidth]{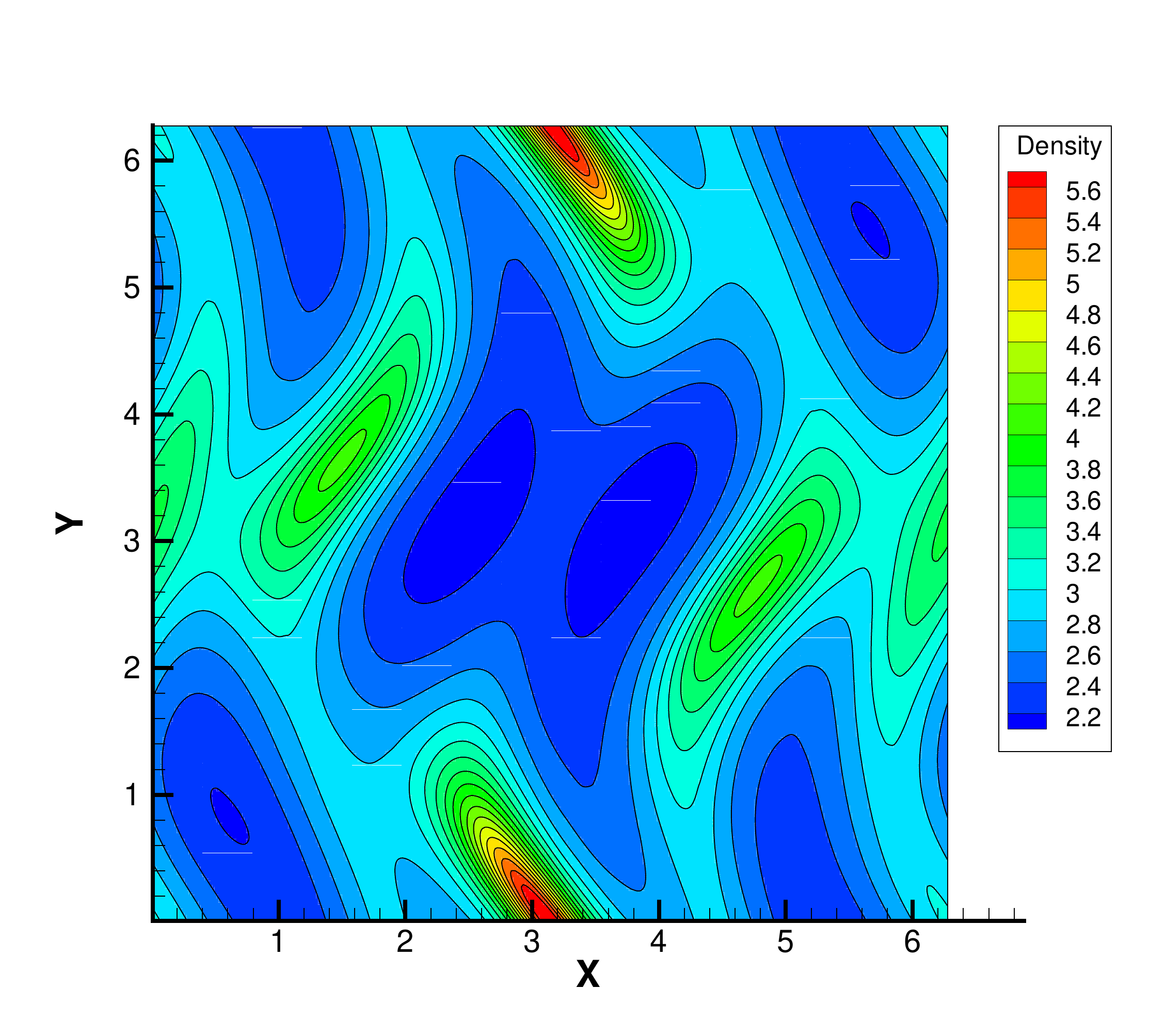}
         \caption{Density contour}
     \end{subfigure}
     \hfill
     \begin{subfigure}[b]{0.48\textwidth}
         \centering
         \includegraphics[width=\textwidth]{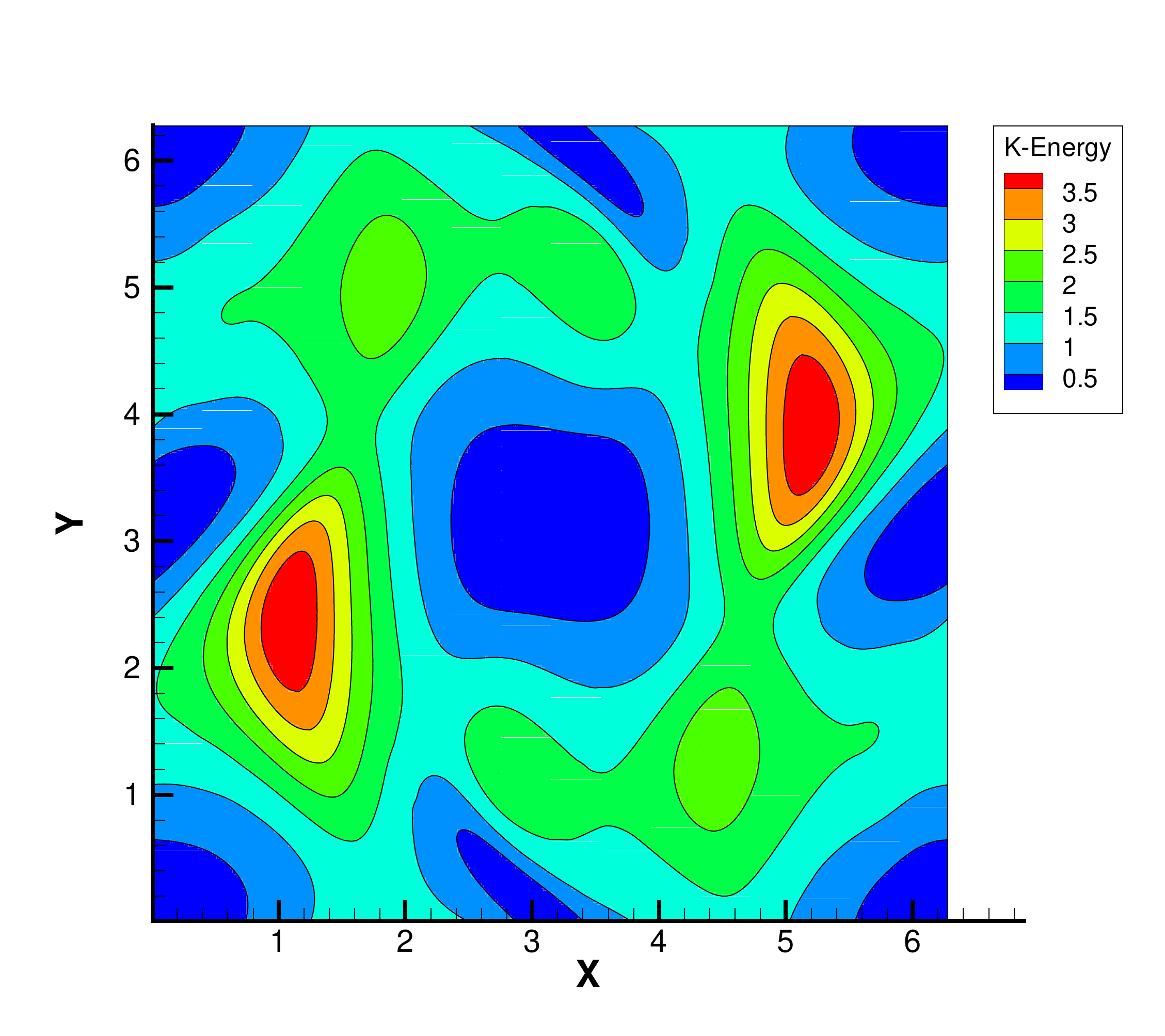}
         \caption{Kinetic energy contour}
     \end{subfigure}
     \vfill
     \begin{subfigure}[b]{0.48\textwidth}
         \centering
         \includegraphics[width=\textwidth]{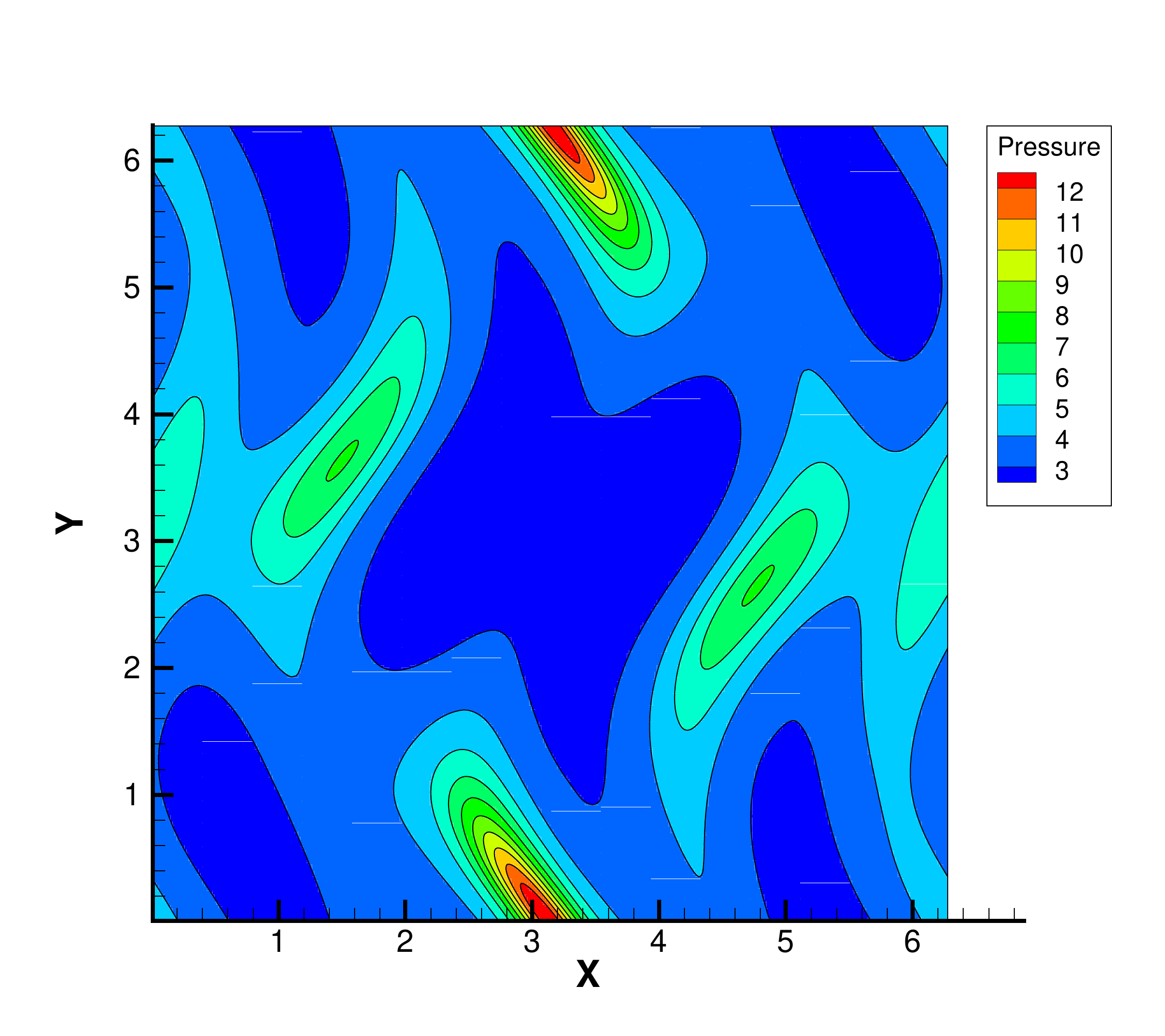}
         \caption{Pressure contour}
     \end{subfigure}
     \hfill
     \begin{subfigure}[b]{0.48\textwidth}
         \centering
         \includegraphics[width=\textwidth]{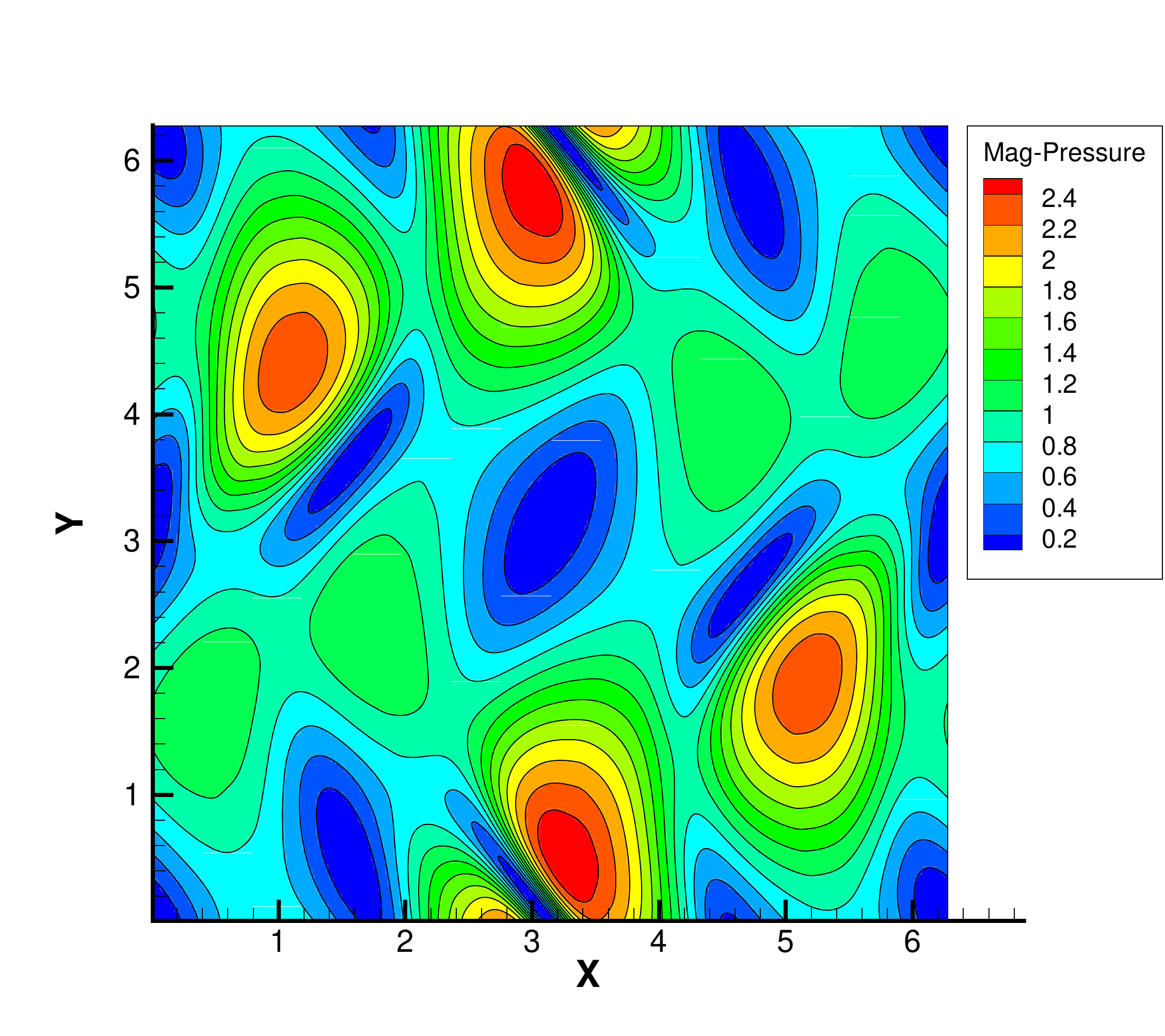}
         \caption{Magnetic pressure contour}
     \end{subfigure}
\caption{The UGKWP results of the multiscale Orszag-Tang vortex problem with $\text{Kn}=10^{-4}$ and $r=0$ at $t=1$.}
\label{orzag3}
\end{figure}

\begin{figure}
     \centering
     \begin{subfigure}[b]{0.48\textwidth}
         \centering
         \includegraphics[width=\textwidth]{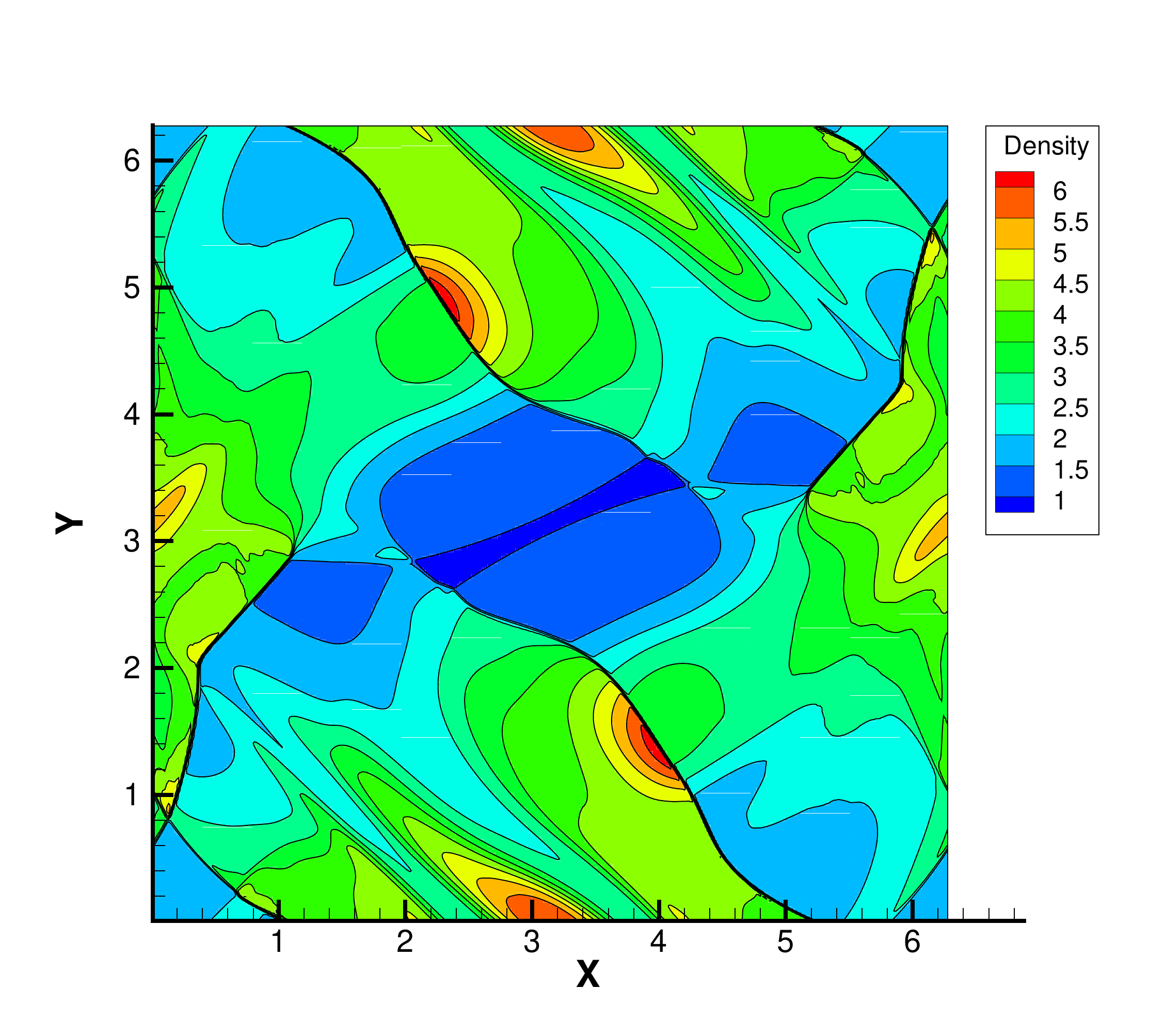}
         \caption{Density contour}
     \end{subfigure}
     \hfill
     \begin{subfigure}[b]{0.48\textwidth}
         \centering
         \includegraphics[width=\textwidth]{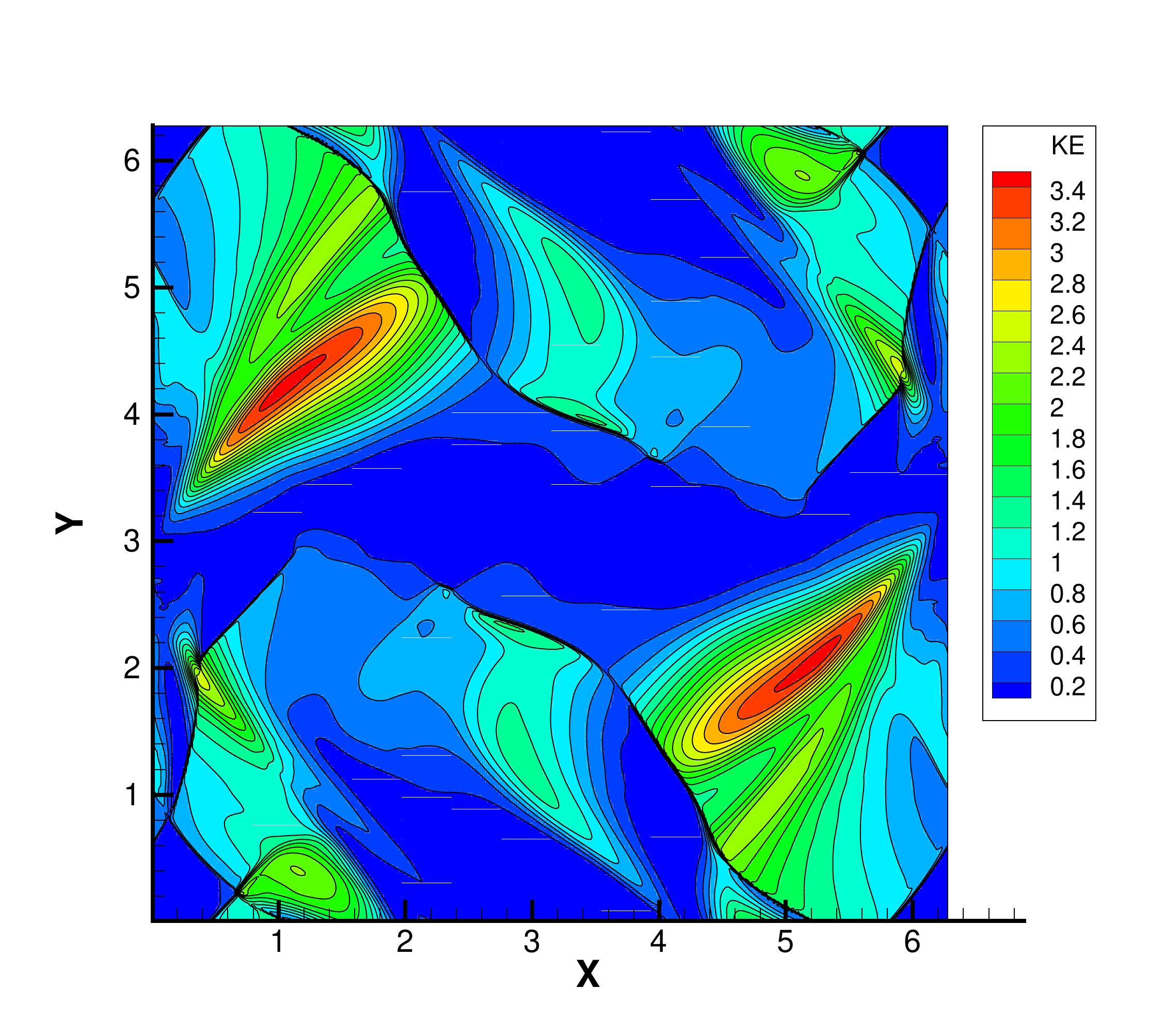}
         \caption{Kinetic energy contour}
     \end{subfigure}
     \vfill
     \begin{subfigure}[b]{0.48\textwidth}
         \centering
         \includegraphics[width=\textwidth]{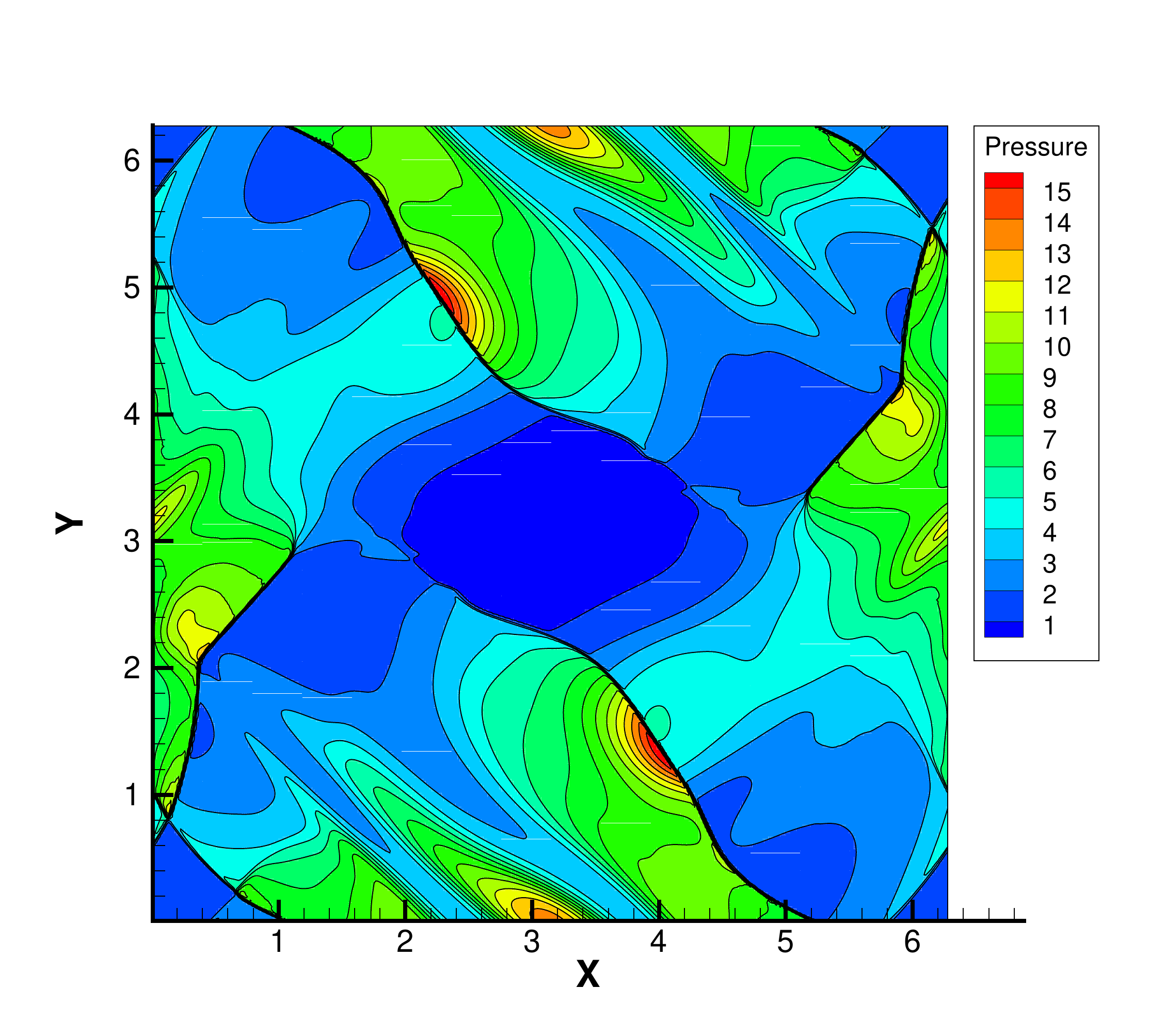}
         \caption{Pressure contour}
     \end{subfigure}
     \hfill
     \begin{subfigure}[b]{0.48\textwidth}
         \centering
         \includegraphics[width=\textwidth]{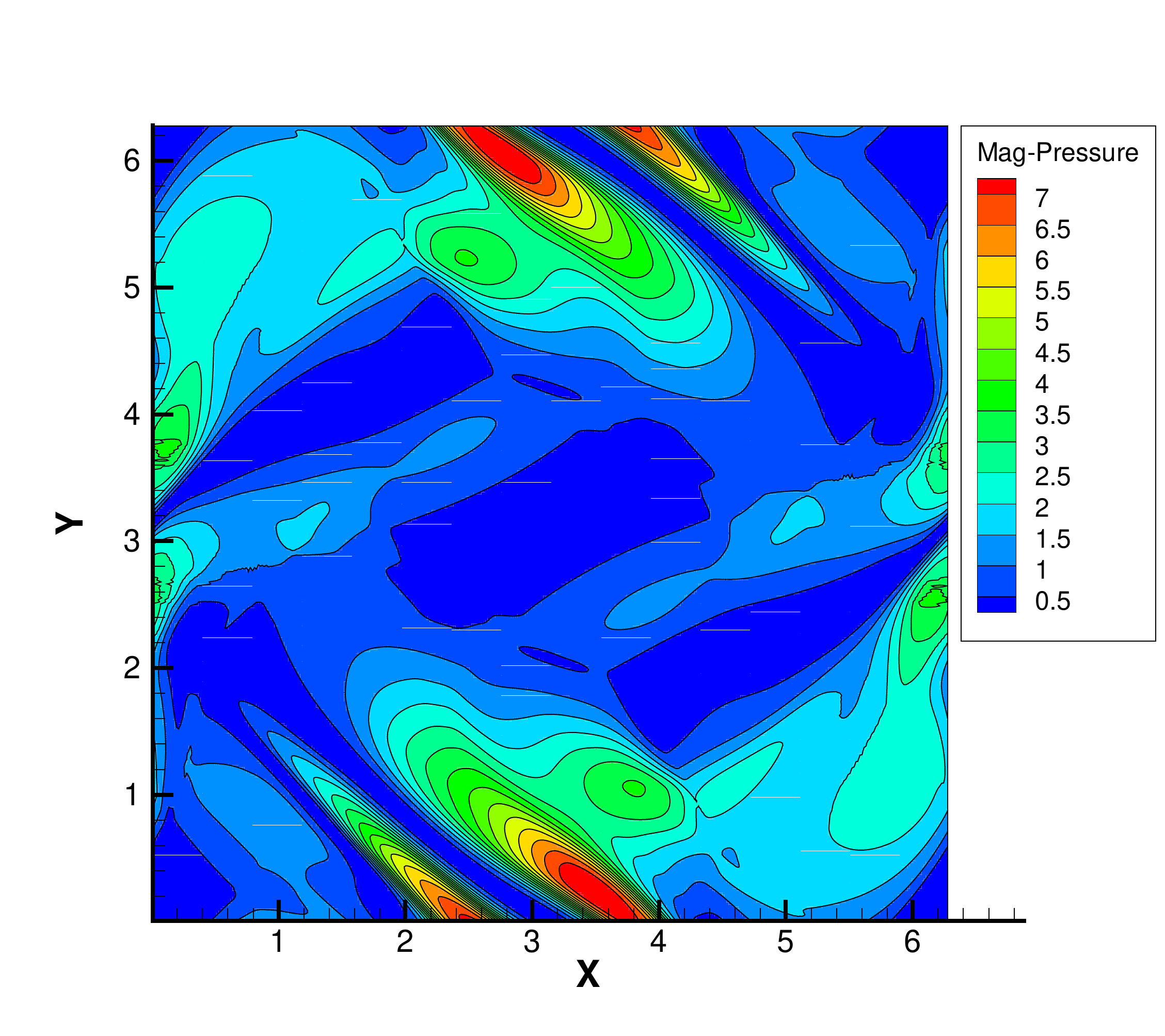}
         \caption{Magnetic pressure contour}
     \end{subfigure}
\caption{The UGKWP results of the multiscale Orszag-Tang vortex problem with $\text{Kn}=10^{-4}$ and $r=0$ at $t=2$.}
\label{orzag4}
\end{figure}

\begin{figure}
     \centering
     \begin{subfigure}[b]{0.48\textwidth}
         \centering
         \includegraphics[width=\textwidth]{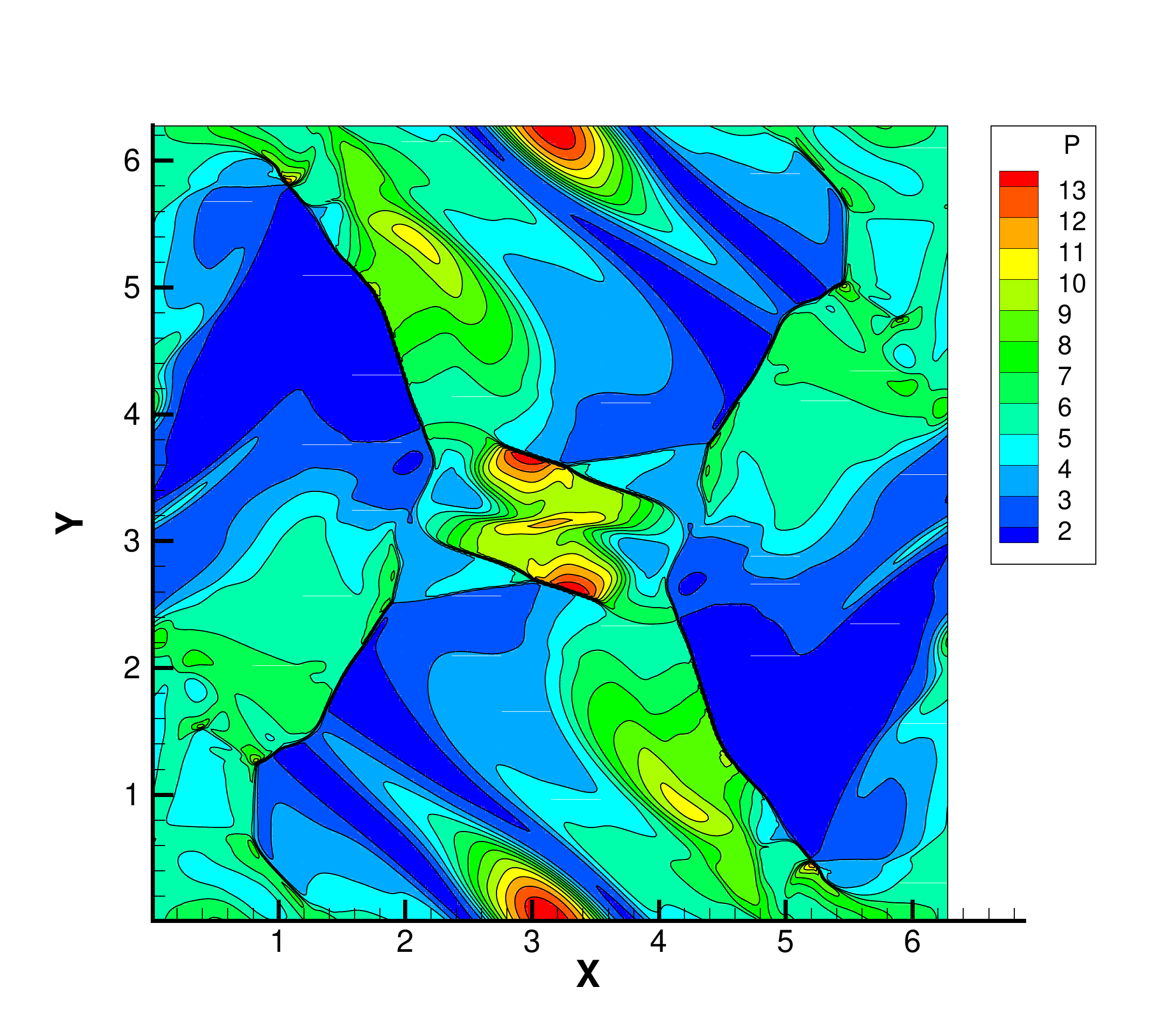}
         \caption{Pressure contour}
     \end{subfigure}
     \hfill
     \begin{subfigure}[b]{0.48\textwidth}
         \centering
         \includegraphics[width=\textwidth]{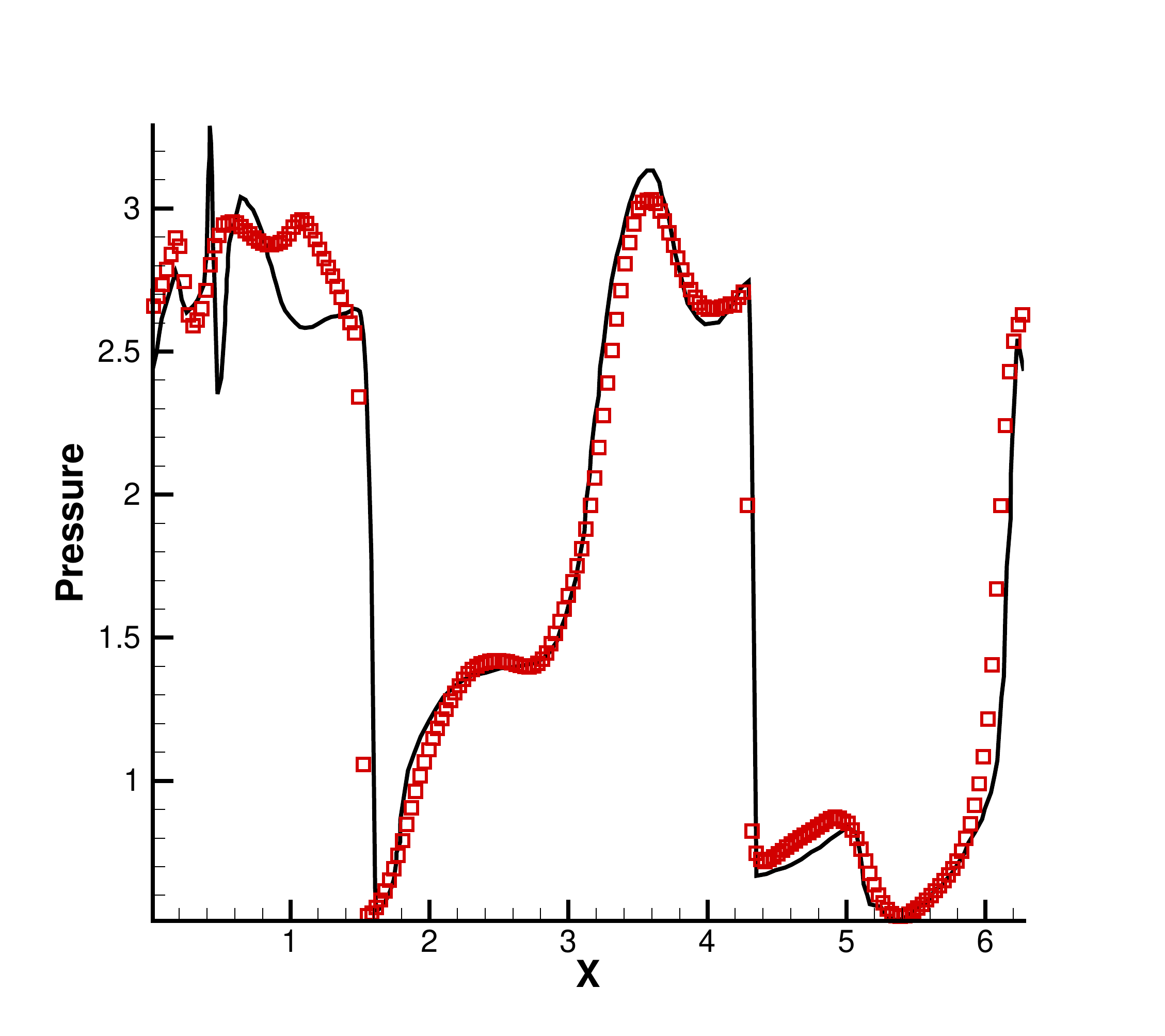}
         \caption{Pressure distribution along $y=0.625\pi$}
     \end{subfigure}
\caption{The UGKWP results of the multiscale Orszag-Tang vortex problem with $\text{Kn}=10^{-4}$ and $r=0$ at $t=3$.
Sub-figure (a) shows the UGKWP pressure contour, and sub-figure (b) shows the comparison of UGKWP and MHD pressure distribution along $y=0.625\pi$.}
\label{orzag5}
\end{figure}

\begin{figure}
     \centering
     \begin{subfigure}[b]{0.48\textwidth}
         \centering
         \includegraphics[width=\textwidth]{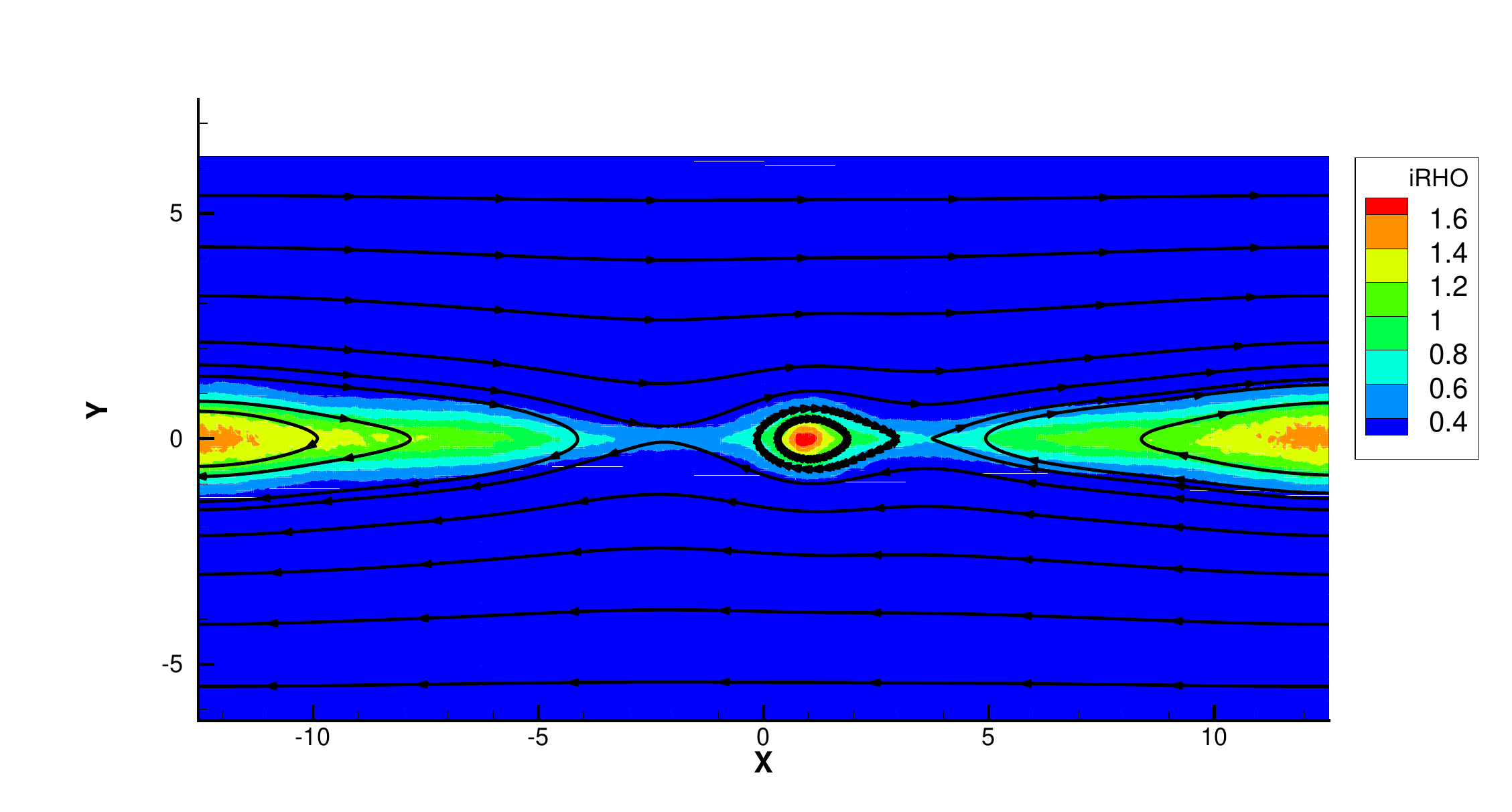}
         \caption{Magnetic field and ion density contour}
     \end{subfigure}
     \hfill
     \begin{subfigure}[b]{0.48\textwidth}
         \centering
         \includegraphics[width=\textwidth]{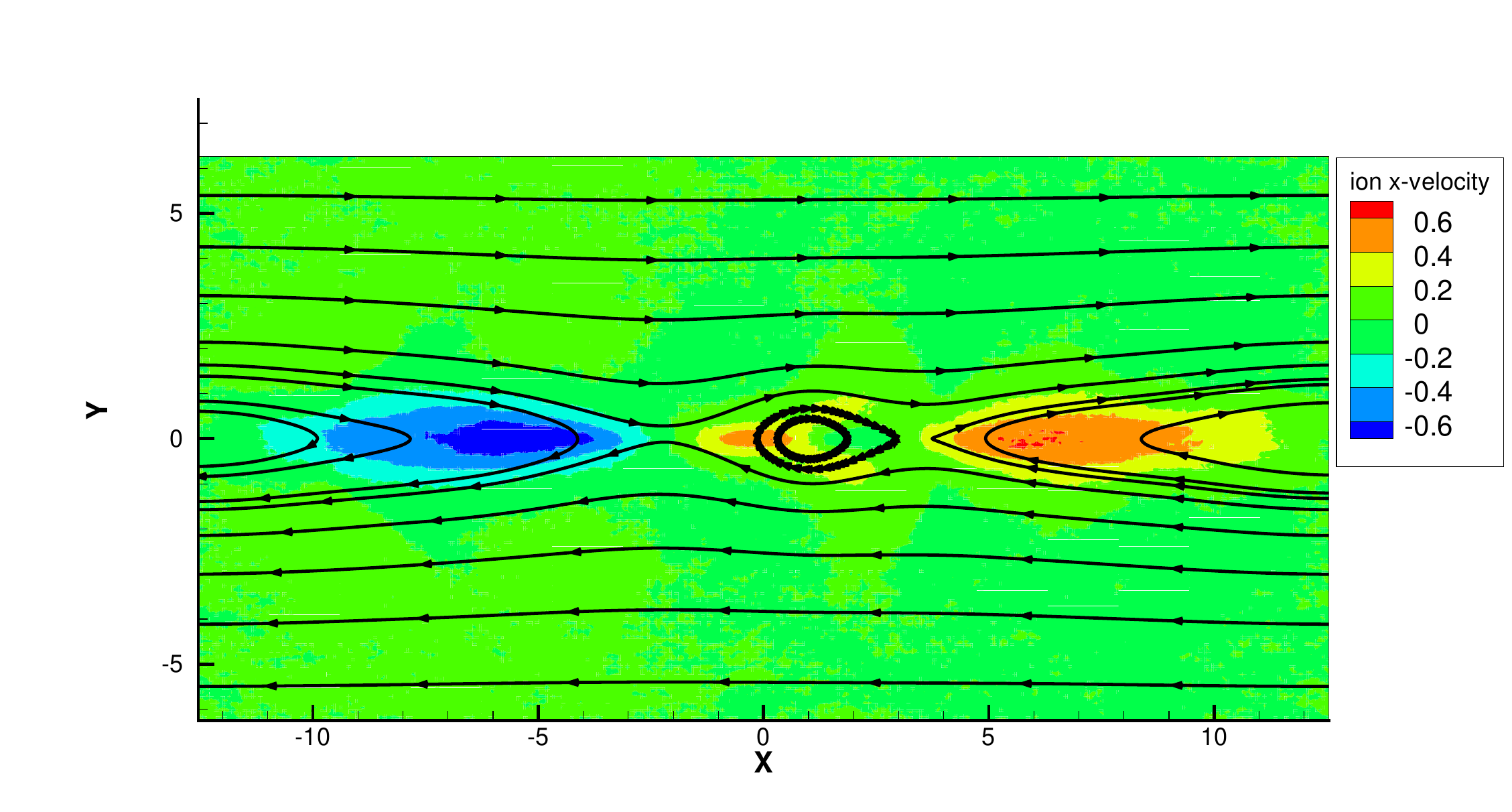}
         \caption{Magnetic field and ion x-velocity contour}
     \end{subfigure}
     \vfill
     \begin{subfigure}[b]{0.48\textwidth}
         \centering
         \includegraphics[width=\textwidth]{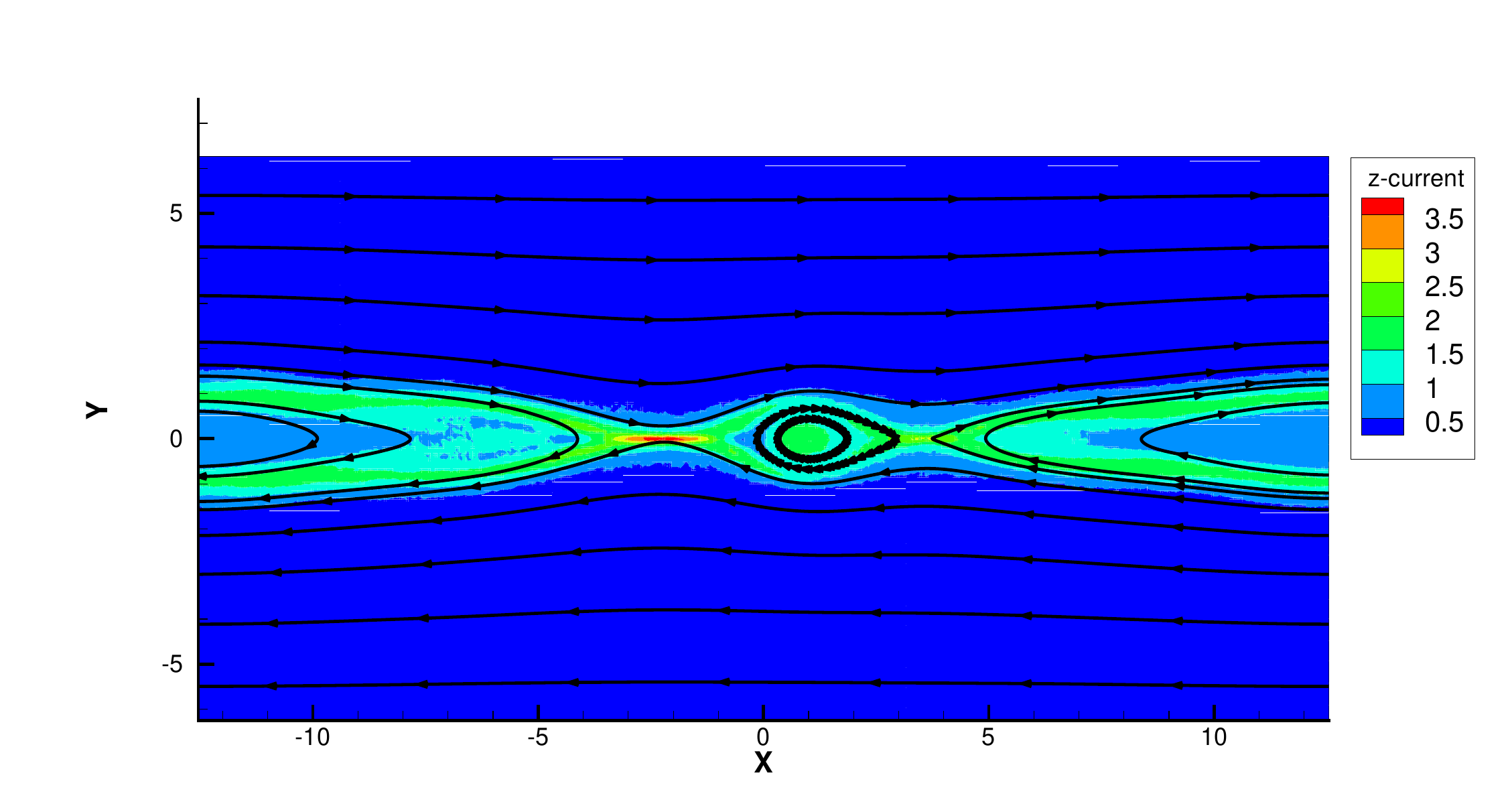}
         \caption{Magnetic field and z-current contour}
     \end{subfigure}
     \hfill
     \begin{subfigure}[b]{0.48\textwidth}
         \centering
         \includegraphics[width=\textwidth]{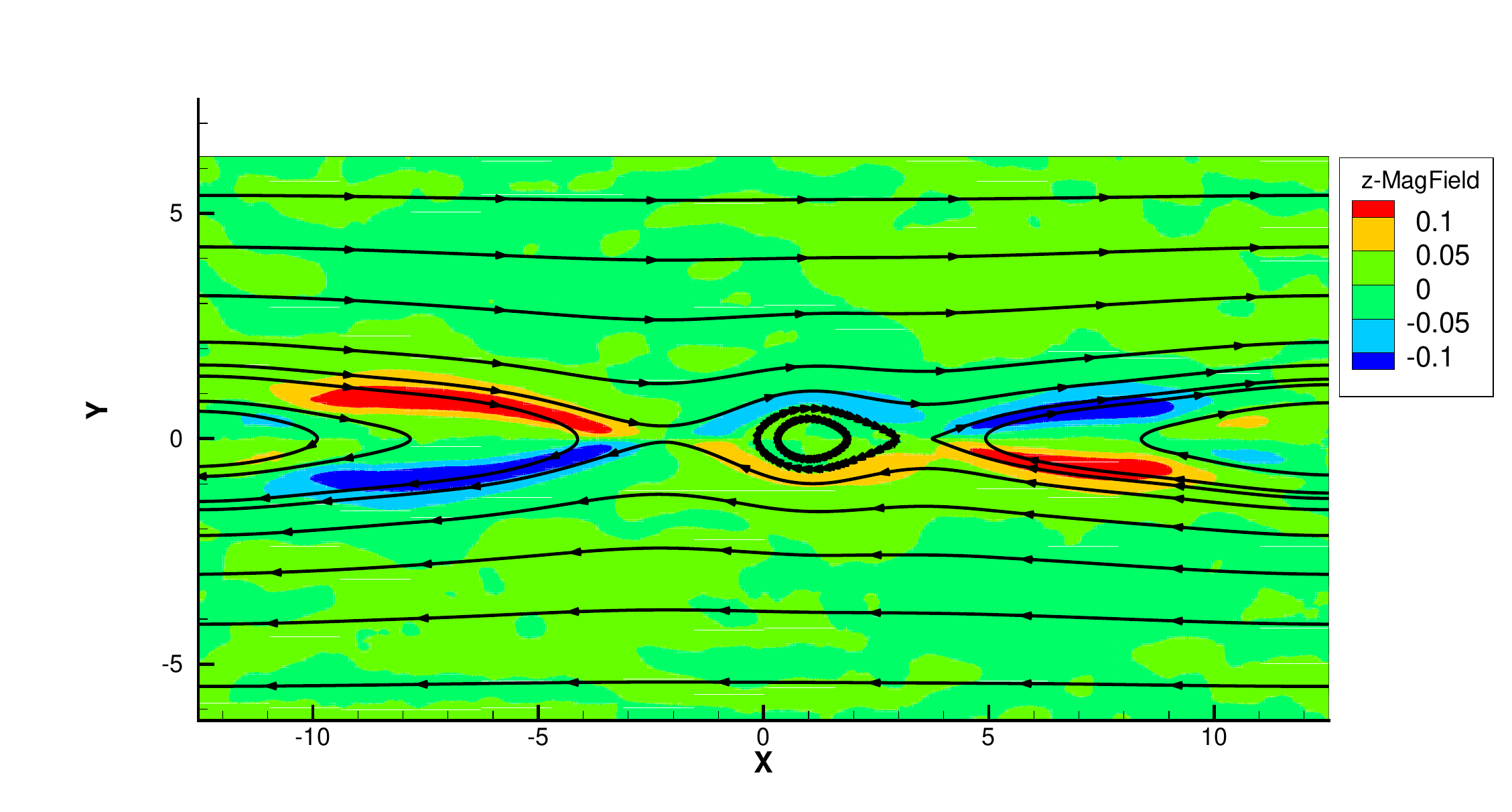}
         \caption{Magnetic field and z-magnetic field contour}
     \end{subfigure}
\caption{The UGKWP results of magnetic reconnection with Kn=$10^{-3}$ at $\omega_{pi} t=15$.}
\label{reconnection1}
\end{figure}

\begin{figure}
     \centering
     \begin{subfigure}[b]{0.48\textwidth}
         \centering
         \includegraphics[width=\textwidth]{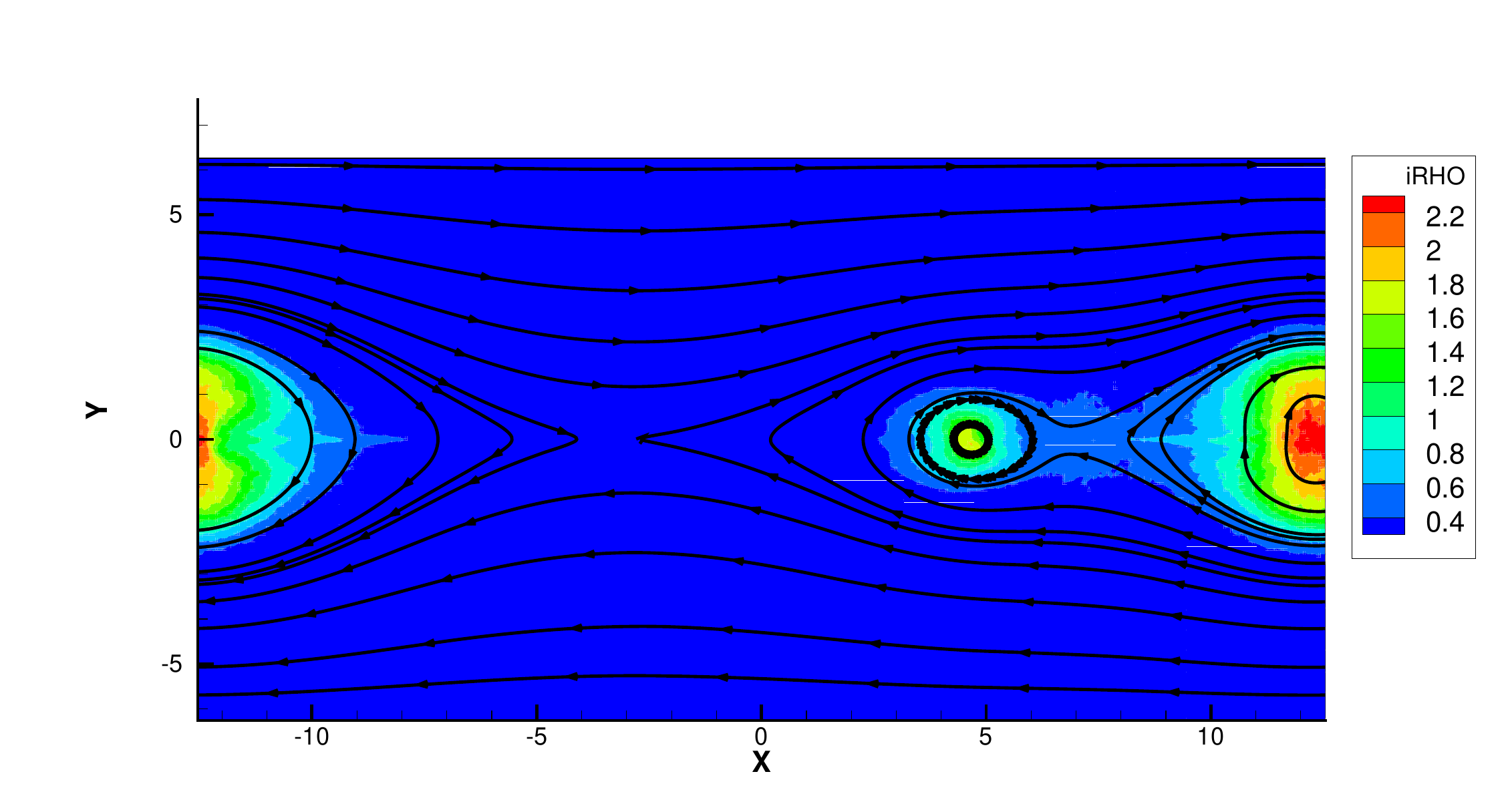}
         \caption{Magnetic field and ion density contour}
     \end{subfigure}
     \hfill
     \begin{subfigure}[b]{0.48\textwidth}
         \centering
         \includegraphics[width=\textwidth]{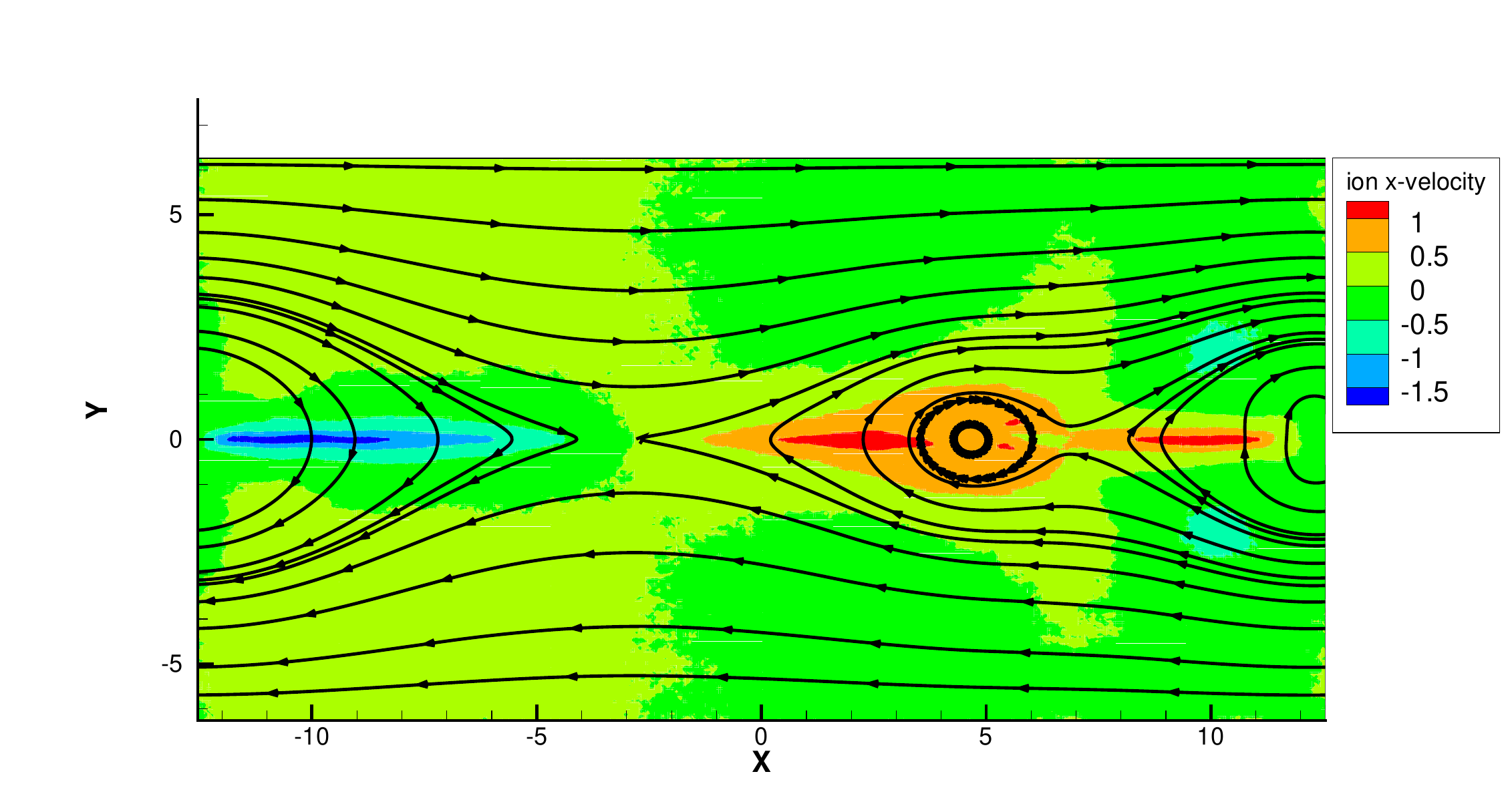}
         \caption{Magnetic field and ion x-velocity contour}
     \end{subfigure}
     \vfill
     \begin{subfigure}[b]{0.48\textwidth}
         \centering
         \includegraphics[width=\textwidth]{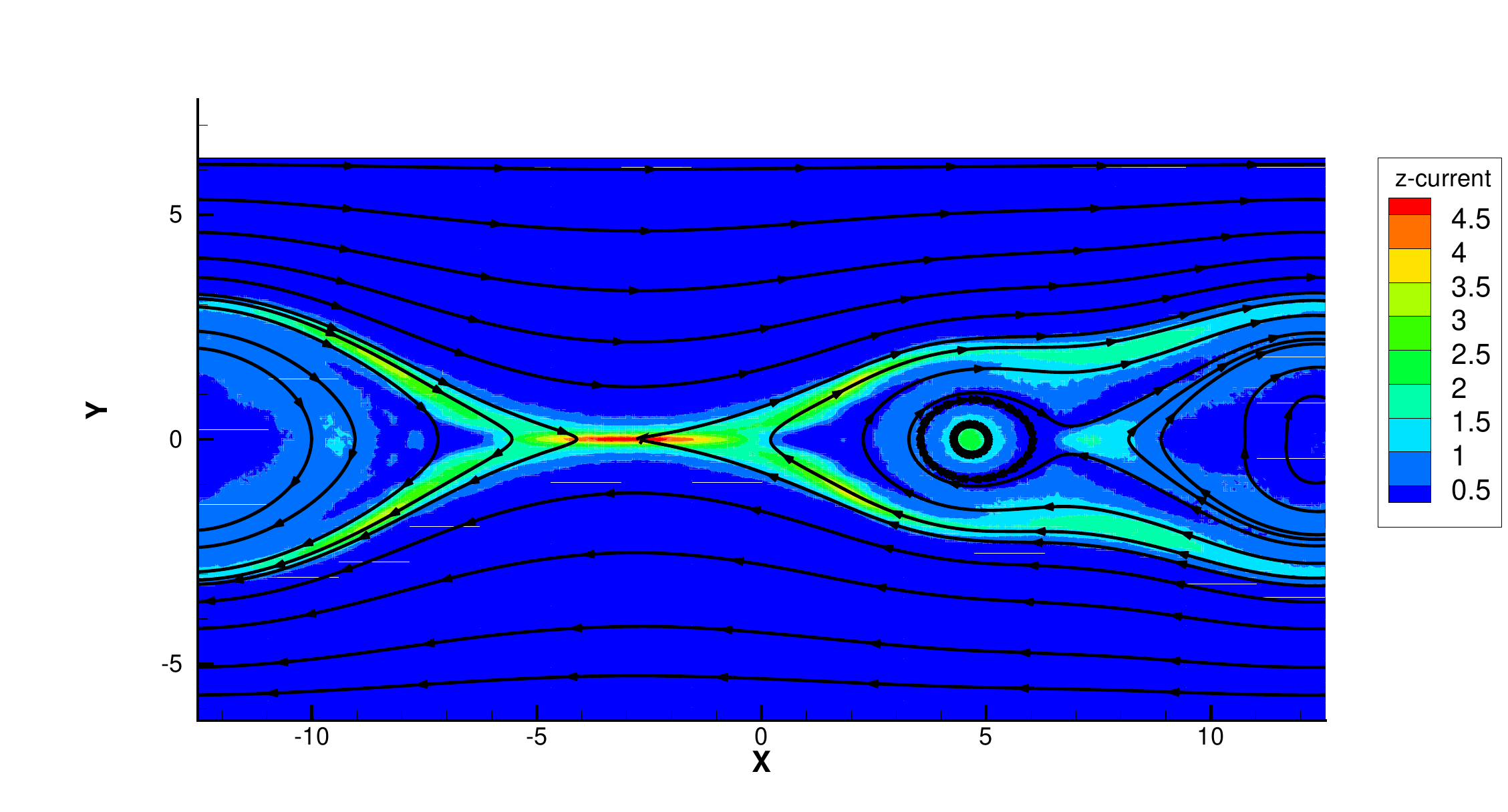}
         \caption{Magnetic field and z-current contour}
     \end{subfigure}
     \hfill
     \begin{subfigure}[b]{0.48\textwidth}
         \centering
         \includegraphics[width=\textwidth]{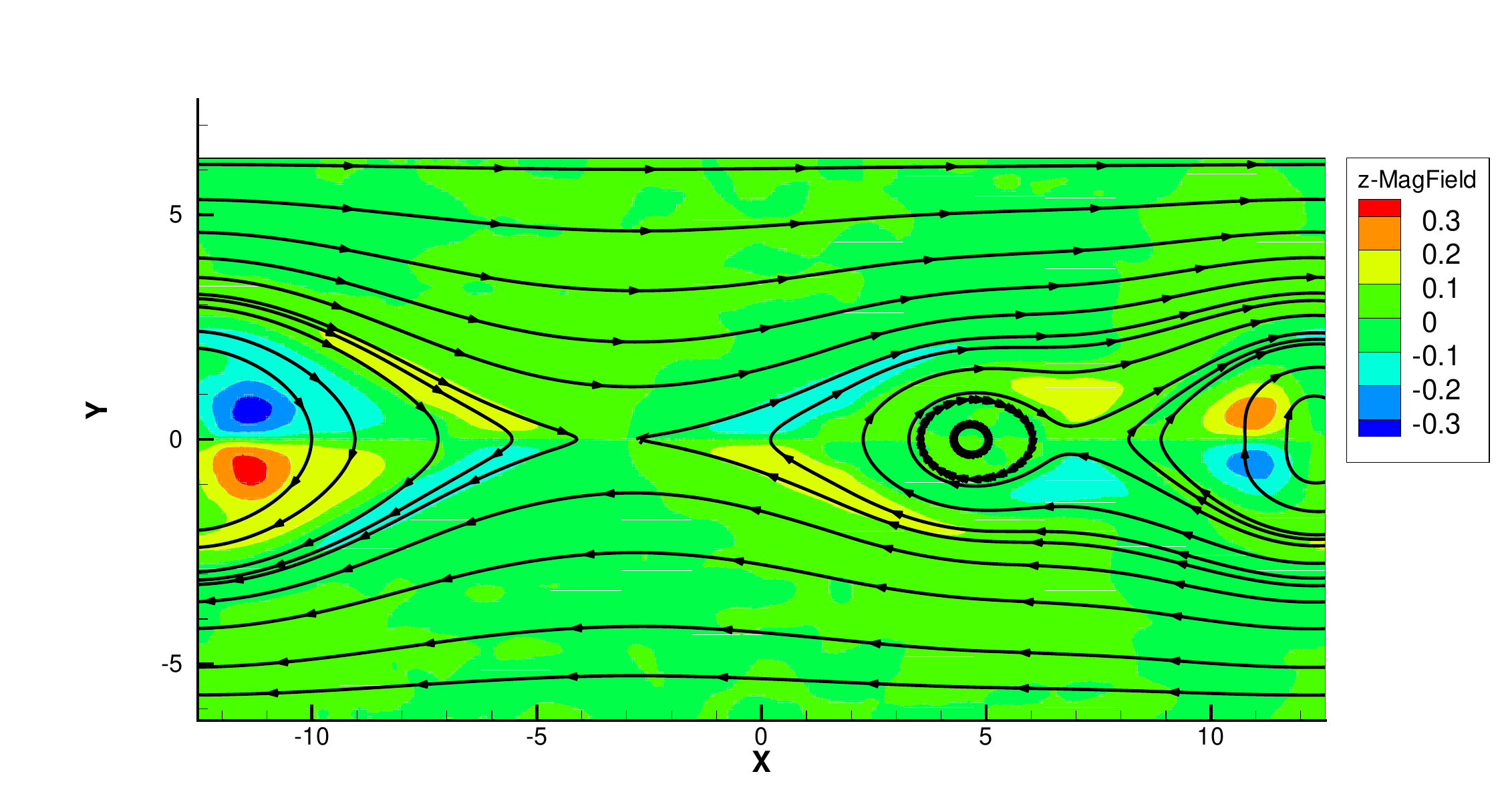}
         \caption{Magnetic field and z-magnetic field contour}
     \end{subfigure}
\caption{The UGKWP results of magnetic reconnection with Kn=$10^{-3}$ at $\omega_{pi} t=30$.}
\label{reconnection2}
\end{figure}

\begin{figure}
     \centering
     \begin{subfigure}[b]{0.48\textwidth}
         \centering
         \includegraphics[width=\textwidth]{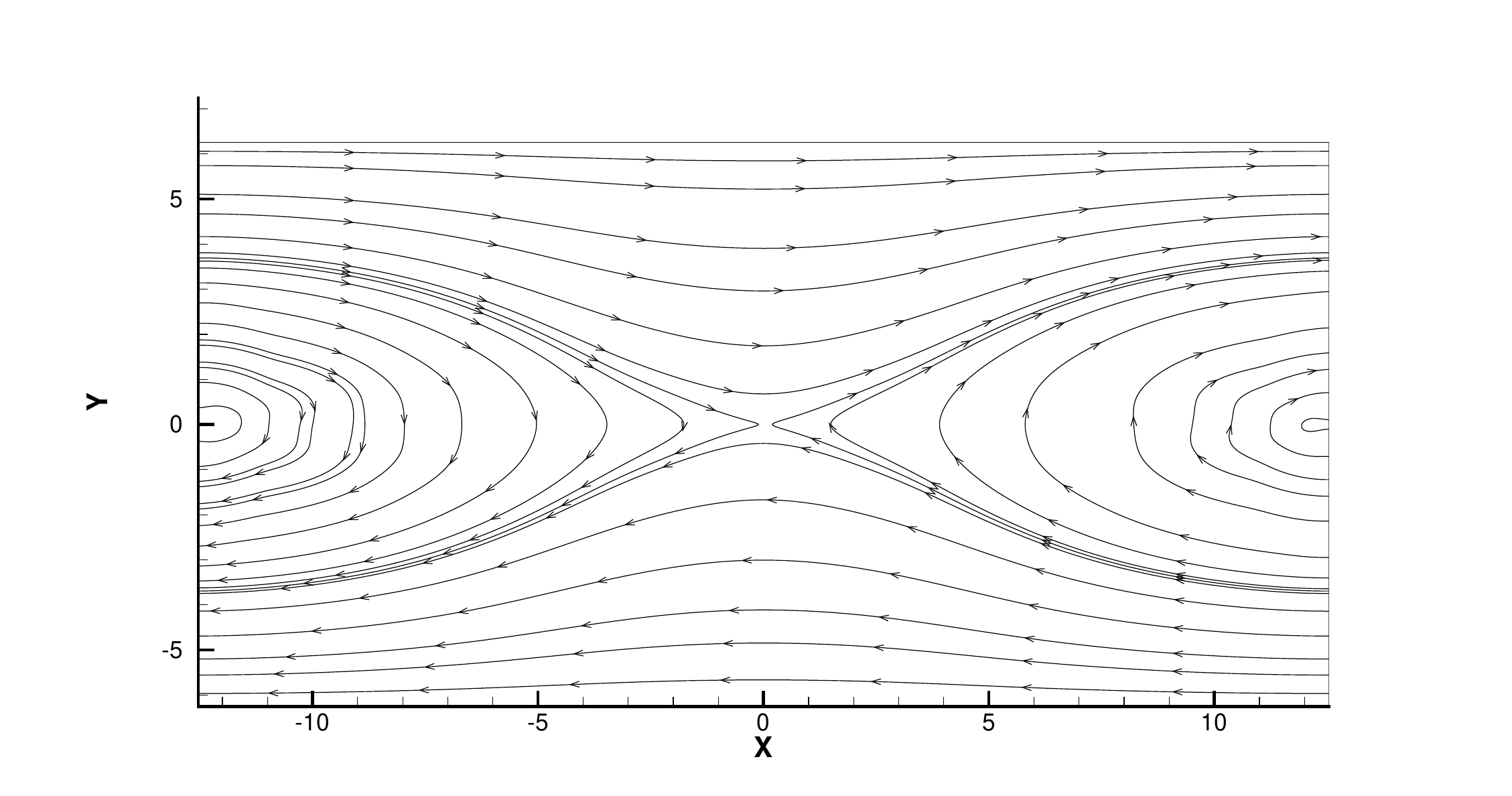}
         \caption{Magnetic field}
     \end{subfigure}
     \hfill
     \begin{subfigure}[b]{0.48\textwidth}
         \centering
         \includegraphics[width=\textwidth]{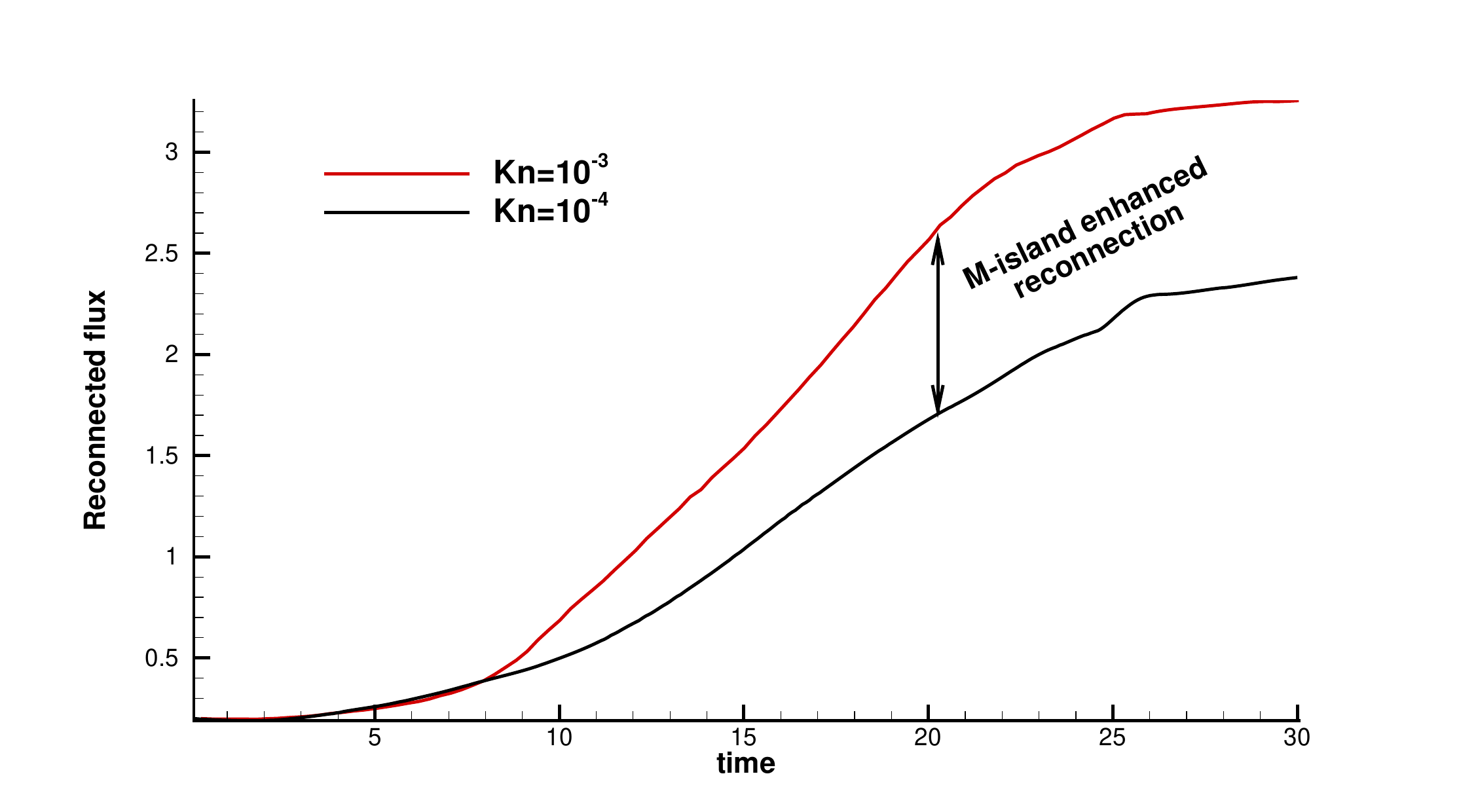}
         \caption{Magnetic reconnection rate}
     \end{subfigure}
\caption{Sub-figure (a) shows the magnetic field at $\omega_{pi} t=30$ in continuum regime with Kn=$10^{-4}$.
And sub-figure (b) shows the reconnected flux in transitional regime with Kn=$10^{-3}$ and in continuum regime with Kn=$10^{-4}$.}
\label{reconnection3}
\end{figure}

\end{document}